\documentclass[harvard,11pt]{article}
\usepackage{amssymb}
\usepackage{longtable}
\usepackage{amsfonts}
\usepackage{amsmath}
\usepackage{adjustbox}
\usepackage{placeins}
\usepackage{amsmath}
\usepackage[abs]{overpic}
\usepackage{linegoal}
\usepackage{hyperref}
\hypersetup{
    colorlinks=true,
    linkcolor=blue,
    filecolor=black,      
    urlcolor=blue,
    citecolor=blue,	
}
\usepackage[FIGBOTCAP]{subfigure}
\usepackage{bbm}
\usepackage{bm}
\usepackage{booktabs}
\usepackage[round]{natbib}
\usepackage{amsmath}
\usepackage{tikz}
\usetikzlibrary{calc}
\usepackage{caption}
\usepackage{tabu}
\usepackage{amsmath}
\usepackage{physics}
\usepackage{booktabs}
\usepackage{graphicx,epstopdf}
\usepackage{setspace,caption}
\usepackage{amsmath}
\usepackage[margin=1in]{geometry}
\usepackage{subfigure}
\usepackage{xargs}

\makeatletter
\newcommand*{\rom}[1]{\expandafter\@slowromancap\romannumeral #1@}
\makeatother
\newcommand{\R}{\mathbb{R}}
\newcommand{\E}{\mathbb{E}}

\newtheorem{theorem}{Theorem}

\newtheorem{corollary}{Corollary}

\newtheorem{definition}{Definition}

\newtheorem{proposition}{Proposition}

\newenvironment{proof}[1][Proof]{\textbf{#1.} }{\  \rule{0.5em}{0.5em}}

\renewcommand{\cite}{\citeasnoun}
\begin{document}

\title{{Pair copula constructions of point-optimal sign-based tests for predictive linear and nonlinear regressions}}
\author{Kaveh Salehzadeh Nobari\thanks{%
Department of Mathematics and Statistics, Lancaster University, Lancaster, LA1 4YR, UK
(\href{emailto: k.salehzadehnobari@lancs.ac.uk}{k.salehzadehnobari@lancs.ac.uk}),}
\\Lancaster University}
\date{\today}
\maketitle

\begin{abstract}
We propose pair copula constructed point-optimal sign tests in the context of linear and nonlinear predictive regressions with endogenous, persistent regressors, and disturbances exhibiting serial (nonlinear) dependence. The proposed approach entails considering the entire dependence structure of the signs to capture the serial dependence, and building feasible test statistics based on pair copula constructions of the sign process. The tests are exact and valid in the presence of heavy tailed and nonstandard errors, as well as heterogeneous and persistent volatility. Furthermore, they may be inverted to build confidence regions for the parameters of the regression function. Finally, we adopt an adaptive approach based on the split-sample technique to maximize the power of the test by finding an appropriate alternative hypothesis. In a Monte Carlo study, we compare the performance of the proposed \textquotedblleft quasi\textquotedblright-point-optimal sign tests based on pair copula constructions by comparing its size and power to those of certain existing tests that are intended to be robust against heteroskedasticity. The simulation results maintain the superiority of our procedures to existing popular tests.   
\end{abstract}

\noindent \textbf{Keywords}: Predictive regressions, D-vine, truncation, sequential estimation, persistent regressors, sign test, exact inference, adaptive method, power envelope.

\noindent \textbf{JEL Codes}: C12, C13, C15, C22, C52

\newpage
\section{Introduction \label{Introduction}}

{\hskip 1.5em} Predictive regressions are frequently encountered in the economics and finance literature. Within this framework, the regressors are often highly persistent (and potentially nonstationary) with errors that exhibit contemporaneous correlation with the disturbances in the predictive regression. This leads to endogeneity, as a result of which the least-squared estimator of the regression parameters is biased. In such settings, \textit{t}-type tests possess a non-standard distribution and inference using asymptotic critical values is no longer valid [see \citet{mankiw1986we} and \citet{stambaugh1999predictive} among others]. The econometric analysis of predictive regressions has been addressed extensively and numerous papers have suggested solutions to overcome this problem. These include reducing the bias in finite samples [see \citet{nelson1993predictable} and \citet{stambaugh1999predictive} among others] or self generated instrumental variables (IVs) that eliminate endogeneity [see \citet{magdalinos2009limit}, \citet{phillips2013predictive} among others]. However, many of these studies impose strict assumptions on the degree of persistency of the regressors  [see \citet{phillips2014confidence} for an overview]. In this paper, we propose pair copula constructed point-optimal sign-based tests (PCC-POS-based tests hereafter) in the context of linear and nonlinear predictive regressions. The proposed tests are robust in the presence of persistent/endogenous regressors and errors, heterogeneous and persistent volatility, and disturbances that exhibit serial (nonlinear) dependence. Moreover, they are exact, consider the entire dependence structure of the signs and may be inverted to build confidence regions for the parameters of the regression function.

Sign-based tests, such as those proposed by \citet{dufour1995exact, campbell1997exact}, \citet{ luger2003exact}, and \citet{dufour2010exact} are randomized tests with a randomized distribution under the null hypothesis of unpredictability [see \citet{pratt2012concepts} for a review of randomized tests]. Hence, under mild assumptions these procedures are distribution-free and do not suffer from the issues encountered by \textit{t}-type statistics in finite samples. These class of tests are valid in the presence of non-normal distributions and heteroskedasiticty of unknown form [see \citet{boldin1997sign} and \citet{taamouti2015finite} for a review of sign-based tests]. Furthermore, \citet{dufour2010exact} show that the heteroskedasticity and autocorrelation corrected tests developed by \citet{white1980heteroskedasticity} (more commonly referred to as \textquotedblleft HAC\textquotedblright{ }procedures) are plagued with low power when the errors follow GARCH-type structures or there is a break in variance. To address these issues, \citet{dufour2010exact} propose point-optimal sign-based (POS-based) inference to test whether the conditional median of a response variable is zero against a linear regression alternative, where these procedures are further extended to nonlinear models.

In an earlier paper, we proposed an extension of the POS-based tests within a predictive regression framework. However, in order to obtain \textit{feasible} test statistics, we imposed a Markovian assumption on the sign process. This paper relaxes the Markovian assumption on the signs, by decomposing the joint distribution of the signs using models proposed by \citet{panagiotelis2012pair} for multivariate discrete data based on pair copula constructions. The latter allow us to build feasible test statistics that are robust against heavy-tailed and asymmetric distributions, provided that the errors have zero median conditional on their own past and the explanatory variables. The PCC-POS-based tests are shown to be robust against non-standard distributions and heteroskedasticity of unknown form, and have the highest power among certain parametric and nonparametric tests that are intended to be robust against heteroskedasticity. Moreover, as in \citet{dufour2010exact}, they can be inverted to produce a confidence region for the vector (sub-vector) of the parameters.  

Although, the literature surrounding sign-based and sign-ranked inference is vast, the focus of the POS-based tests constructed by \citet{dufour2010exact} is to maximize power at a nominated point in the alternative parameter space, such that the power of the POS-based test is close to that of the power envelope - i.e. maximum attainable power for a given testing problem [see \citet{king1987towards}]. Similarly, the PCC-POS-based tests are Neyman-Pearson type tests based on the signs, and as in \citet{dufour2010exact} a practical problem concerns finding an alternative at which the power of the PCC-POS-based tests is close to that of the power envelope. By conducting simulations exercises, \citet{dufour2010exact} find that the power of the POS-based tests is shifted close to the power envelope, when 10\% of the sample is used to estimate the alternative and the remaining portion to calculate the test-statistic. Our simulations using a variety of split-sample PCC-POS-based tests confirm these findings.

Due to the nonlinear nature of the signs, there is inherent uncertainty regarding the structure of sign dependence. Therefore, it is important to consider the entire dependence structure of the signs. One approach for computing the joint distribution of the signs $s(y_1),\cdots,s(y_T)$, where $s(y_i)=\mathbbm{1}_{\R^+ \cup\{0\} }\{y_i\}$, entails taking advantage of copula functions [see \citet{sklar1959fonctions}], which express the joint distribution of the signs in terms of i) the marginal distributions of the individual signs; and ii) the copula models capturing the dependence of the $T$ signs. As the signs are discrete, the likelihood function of the POS-based tests under the alternative hypothesis can then be calculated using rectangle probabilities and in turn estimated using copulas with closed analytical form. However, this approach would not yield feasible test statistics, as the number of \textit{multivariate} copulas that need to be evaluated increase at an exponential rate with growing sample sizes $T$. As a result of this curse of dimensionality, the literature concerning calculating probability mass functions (p.m.f hereafter) using discrete data is limited to low-dimensional data and copulas that are fast to calculate [see \citet{nikoloulopoulos2008multivariate,nikoloulopoulos2009finite} and \citet{li2010two}]. 

To propose feasible POS-based test statistics within a predictive regression framework, we use a discrete analogue of the vine PCCs introduced by \citet{panagiotelis2012pair}. The likelihood function of the signs under the alternative hypothesis can then be decomposed as a vine PCC under certain set of conditions that are later outlined in the paper. The most important advantage of this method is that for a sample of size $T$, only $2T(T-1)$ \textit{bivariate} copula evaluations are required, as opposed to $2^T$ \textit{multivariate} copula evaluations using rectangle probabilities. Another advantage of the vine PCC methodology is that model selection techniques can be used to identify the conditional independence in the process of signs in order to create more parsimonious PCC models.

The rest of the paper is organized as follows: in Section \ref{Framework}, we motivate the use of the discrete analogue of the vine PCC for building POS-based tests. In Section \ref{Point-optimal sign
test based on PCC}, we outline the conditions under which vine PCCs can be implemented and we also discuss the choice of the PCC model. We then propose PCC-POS-based test statistics for linear and nonlinear models. Secion \ref{EstimationC3}, discusses the estimation approach implemented for the vine PCC models. In Section \ref{optimal alternative hypothesis}, we derive the power envelope and highlight the choice of the alternative hypothesis for computing the PCC-POS-based test statistic. In Section \ref{projectiontechniqueC3}, we discuss the problem of finding a confidence set for a vector (subvector) of parameters using the projection techniques. In Section \ref{sec: Monte Carlo study}, we assess the performance of the proposed tests in terms of size and power. Finally, in Section \ref{ConclusionC3}, we conclude the findings of the paper.

\section{Framework \label{Framework}}

 {\hskip 1.5em}Consider a stochastic process $Z=\{Z_t=(y_t,\bm{x}_{t-1}'):\Omega\rightarrow\R^{(k+1)},t=1,2,\cdots\}$ defined on a probability space $(\Omega,\mathcal{F},P)$. Within a framework of a predictive regression $y_{t}$ can linearly be explained by a vector variable $\bm{x}_{t-1}$%
\begin{equation}
y_{t}=\bm{\beta}'\bm{x}_{t-1}+\varepsilon_{t},\quad t=1,\cdots,T,  \label{eq: DGP}
\end{equation}%
where $y_t$ is the dependent variable and $\bm{x}_{t-1}$ is an $(k+1)\times 1$ vector of stochastic explanatory 
variables, say $\bm{x}_{t-1}=[1,x_{1,t-1},\cdots,x_{k,t-1}]'$, $\bm{\beta} \in \mathbb{R}^{(k+1)}$ is an unknown vector of parameters with $\bm{\beta}=[\beta_0,\beta_1,\cdots,\beta_k]'$ and
\[
 \varepsilon_t\mid X\sim F_t(.\mid X)
\] 
where $F_{t}(.\mid X)$ is an unknown conditional distribution function and X=$[\bm{x}_0',\cdots,\bm{x}_{T-1}']'$ is an $T\times (k+1)$ information matrix. 

We follow \citet{coudin2009finite} by considering the median as an alternative measure of central tendency, which differs from the Martingale Difference Sequence (MDS hereafter) assumption. In the latter, it is generally assumed that for an adapted stochastic sequence $\{Z_t,\mathcal{F}_{t}, t=1,2,\cdots\}$, where $\mathcal{F}_{t}$ is a $\sigma$-field in $\Omega$, $\mathcal{F}_s\subseteq\mathcal{F}_t$ for $s< t$, $\mathbb{E}\{\varepsilon_t\mid\mathcal{F}_{t-1}\}=0$, $\forall t\geq 1$. Thus the alternative implies imposing a median-based analogue of the MDS on the error process - namely we suppose that $\varepsilon_t$ is a strict conditional mediangale
\begin{equation}\label{eq: median}
P[\varepsilon_{t}> 0\mid \bm{{\varepsilon}}_{t-1},X]=P[\varepsilon_{t}<0\mid \bm{\varepsilon}_{t-1},X]=\frac{1}{2},
\end{equation}%
with
\[
\bm{\varepsilon}_{0}=\{\emptyset\},\quad\bm{\varepsilon}_{t-1}=\{\varepsilon_1,\cdots,\varepsilon_{t-1}\},\quad\text{for}\quad t\geq2.
\]

 Note (\ref{eq: median}) entails
that $\varepsilon _{t}\mid X$ has no mass at zero for all $t$, which is only true if $\varepsilon_{t}\mid X$\ is a
continuous variable. Model (\ref{eq: DGP}) in conjunction with assumption (\ref{eq: median}) allows the error terms to possess asymmetric, heteroskedastic and serially dependent distributions, so long as the conditional medians are zero. Assumption (\ref{eq: median}) allows for many dependent schemes, such as those of the form $\varepsilon_1=\sigma_1(x_1,\cdots,x_{t-2})\epsilon_1$, $\varepsilon_t=\sigma_1(x_1,\cdots,x_{t-2},\varepsilon_1,\cdots,\varepsilon_{t-1})\epsilon_t$ ,$t=2,\cdots,n$, where $\epsilon_1,\cdots,\epsilon_n$ are independent with a zero median. In time-series context this includes models such as ARCH, GARCH or stochastic volatility with non-Gaussian errors. Furthermore, in the mediangale framework the disturbances need not be second order stationary. 

Suppose, we wish to test the null hypothesis of unpredictability
\begin{equation}\label{eq: null}
H_0:\bm{\beta}=\bm{0},
\end{equation}
against the alternative
\begin{equation}\label{eq: alt}
H_1:\bm{\beta}=\bm{\beta}_1.
\end{equation}
where $\bm{0}$ is a $(k+1)\times 1$ zero vector. Define the vector of signs as follows
\begin{equation*}
U(T)=(s(y_1),\cdots,s(y_T))',
\end{equation*}
where for $t=1,\cdots,T$
\begin{equation*}
s(y_{t})=\left\{ 
\begin{tabular}{l}
$1,$ $if$ $y_{t}\geq 0$ \\ 
$0,$ $if$ $y_{t}<0$%
\end{tabular}%
\ \right. \text{.}
\end{equation*}
We consider Neyman-Pearson type test based on the signs. Thus, to build POS-based tests for testing the null hypothesis (\ref{eq: null}) against the alternative (\ref{eq: alt}), we first define the likelihood function of the sample in terms of signs $s(y_1),\cdots,s(y_T)$ conditional on $X$
\begin{equation}\label{eq: likelihood}
L(U(T),\bm{\beta},X)=P[s(y_1)=s_1,\cdots,s(y_T)=s_T\mid X]=\prod\limits_{t=1}^{T}P\left[s(y_t)=s_t\mid \text{\b{S}}_{t-1}=\text{\b{s}}_{t-1},X\right],
\end{equation}
with 
\begin{equation*}
\text{\b{S}}_{0}=\{\emptyset\},\quad \text{\b{S}}_{t-1}=\{s(y_1),\cdots,s(y_{t-1})\},\quad\text{for}\quad t\geq 2,
\end{equation*}
and
\[
P[s(y_1)=s_1\mid\text{\b{S}}_{0}=\text{\b{s}}_{0},X]=P[s(y_1)=s_1\mid X],
\]
where each $s_t$ for $1\leq t\leq T$ takes two possible values of 0 and 1. Under model (\ref{eq: DGP}) and assumption (\ref{eq: median}), the variables $s(\varepsilon_1),\cdots,s(\varepsilon_T)$ and in turn $s(y_1),\cdots,s(y_T)$ are i.i.d conditional on $X$, according to the distribution
\[
P[s(\varepsilon_1)=1\mid X]=P[s(\varepsilon_1)=0\mid X]=\frac{1}{2},\quad t=1,\cdots,T.
\]
This results holds true iff for any combination of $t=1,\cdots,T$ there is a permutation $\pi: i\rightarrow j$ such that the mediangale assumption holds for $j$. Then the signs $s(\varepsilon_1),\cdots,s(\varepsilon_T)$ are i.i.d [see Theorem \ref{Theorem1} of \citet{coudin2009finite}]. Therefore, under the null hypothesis of unpredictability we have

\begin{equation}
P[s(y_t)=1\mid X]=P[s(y_t)=0\mid X]=\frac{1}{2},\quad t=1,\cdots,T.
\end{equation}
Consequently, under the null hypothesis of orthogonality, the log-likelihood function conditional on $X$ is given by
\[
L_0(U(T),\bm{0},X)=\prod\limits_{t=1}^{T}P[s(y_t)=s_t\mid X]=\left(\frac{1}{2}\right)^T.
\]
On the other hand, under the alternative we have
\[
L_1(U(T),\bm{\beta}_1,X)=\prod\limits_{t=1}^{T}P[s(y_t)=s_t\mid \text{\b{S}}_{t-1}=\text{\b{s}}_{t-1}, X],
\]
where now for $t=1,\cdots,T$
\[
y_t=\bm{\beta}_1'\bm{x}_{t-1}+\varepsilon_t.
\]

In an earlier paper, we considered optimal sign-based tests (in the Neyman-Pearson sense), which maximize power under the constraint $P[\text{Reject }  H_0\mid H_0]\leq \alpha$, where $\alpha$ is an arbitrary significance level [see \citet{lehmann2006testing}]. Let $H_0$ and $H_1$ be defined by (\ref{eq: null}) and (\ref{eq: alt}) respectively. Then under the assumptions (\ref{eq: DGP}) and (\ref{eq: median}), the log-likelihood ratio
\begin{equation}\label{eq: teststat}
SL_T(\bm{\beta_1})=\ln\left\{\frac{L_1(U(T),\bm{\beta}_1,X)}{L_0(U(T),\bm{0},X)}\right\}>c,
\end{equation} 
is most powerful for testing $H_0$ against $H_1$ among level $\alpha$ tests based on the signs $(s(y_1),\cdots,s(y_T))'$, where $c$ is the smallest constant such that 
\[
P[SL_T(\bm{\beta}_1)>c\mid H_0]\leq \alpha,
\]
and $\alpha$ is an arbitrary significance level.

 For POS-based tests within a predictive regression framework, the test statistic requires the calculation of $P[y_t\geq0\mid\text{\b{S}}_{t-1}=\text{\b{s}}_{t-1},X]$ and $P[y_t<0\mid\text{\b{S}}_{t-1}=\text{\b{s}}_{t-1},X]$. The latter is not easy to compute, as it involves the distribution of the joint process of signs  $s(y_1),\cdots,s(y_T)$, conditional on $X$, which is unknown. Therefore, to obtain feasible test statistics, we may impose a Markovian assumption on the sign process. However, it may be important to capture the dependence structure of the entire process.

In this paper, we consider the entire dependence structure of the vector of signs by taking advantage of copulas. The Theorem of  \citet{sklar1959fonctions} states that there exists a copula $C$ such that
\begin{equation}
F(s_1,\cdots,s_T\mid X)=C(F_1(s_1\mid X),\cdots,F_T(s_T\mid X)),
\end{equation}
where $F$ is a conditional joint cumulative distribution function (CDF hereafter) of the signs vector $\text{\b{S}}=(s(y_1),\cdots,s(y_T))'$ with conditional marginal distribution functions $F_j$ for $j=1,2,\cdots,T$. Copula $C(.)$ is unique for continuous variables, but for discrete variables, it is unique only on the set
\[
\text{Range}(F_1)\times\cdots\times\text{Range}(F_T),
\]
which is the Cartesian product of the ranges of the conditional marginal distribution functions. To illustrate an example of non-uniqueness in the discrete case, let us consider a sample of two discrete binary variables, say $s(y_1)$ and $s(y_2)$, with corresponding conditional marginal distribution functions $F_1$ and $F_2$. We know that $F_j\sim\text{Bernoulli}(p_j)$ for $j=1,2$, such that
\begin{equation}\label{eq: BernoulliCDF}
F_j(s_j\mid X)=\left\{ 
\begin{tabular}{lll}
$0,$ &$\text{for}$& $s_j< 0$ \\ 
$1-p_j,$ &$\text{for}$ &$0\leq s_j< 1$ \\
$1,$& $\text{for}$& $s_j\geq 1$
\end{tabular}%
\right.
\end{equation} 
Thus, $\text{Range}(F_1)=\{0,1-p_1,1\}$ and $\text{Range}(F_2)=\{0,1-p_2,1\}$, with the copula only being unique for $C(1-p_1,1-p_2)$, noting that $C(0,1-p_j)=0$ and $C(1,1-p_j)=1-p_j$ for $j=1,2$. However, this non-uniqueness does not preclude the use of parametric copulas for modelling discrete data [see. \citet{joe1997multivariate}, \citet{song2009joint}]. 

By considering this bivariate example, the conditional p.m.f can be expressed in terms of rectangle probabilities, 
\begingroup
\allowdisplaybreaks
\begin{align*}
P[s(y_1)=s_1,s(y_2)=s_2\mid X]&=P[s_1-1<s(y_1)\leq s_1,s_2-1<s(y_2)\leq s_2\mid X]\\
&=F(s_1,s_2\mid X)-F(s_1-1,s_2\mid X)\\
&\textcolor{white}{=}-F(s_1,s_2-1\mid X)+F(s_1-1,s_2-1\mid X),
\end{align*}
\endgroup
and in turn in terms of copulas as follows
\begingroup
\allowdisplaybreaks
\begin{align*}
P[s(y_1)=s_1,s(y_2)=s_2\mid X]&=F(s_1,s_2\mid X)-F(s_1-1,s_2\mid X)\\
&\textcolor{white}{=}-F(s_1,s_2-1\mid X)+F(s_1-1,s_2-1\mid X)\\
&=C(F_1(s_1\mid X),F_2(s_2\mid X))-C(F_1(s_1-1\mid X),F_2(s_2\mid X))\\
&\textcolor{white}{=}-C(F_1(s_1\mid X),F_2(s_2-1\mid X))+C(F_1(s_1-1\mid X),F_2(s_2-1\mid X)),
\end{align*}
\endgroup
which implies that the $T$-variate conditional likelihood function (\ref{eq: likelihood}) can be expressed in terms of $2^T$ finite differences
\begingroup
\allowdisplaybreaks
\begin{align*}
P[s(y_1)=s_1,\cdots,s(y_T)=s_T\mid X]&=\sum\limits_{i_1=0,1}\cdots\sum\limits_{i_T=0,1}(-1)^{i_1+\cdots+i_T}P[s(y_1)\leq s_1-i_1,\cdots,s(y_T)\leq s_T-i_T\mid X]\\
&=\sum\limits_{i_1=0,1}\cdots\sum\limits_{i_T=0,1}(-1)^{i_1+\cdots+i_T}C(F_1(s_1-i_1\mid X),\cdots,F_T(s_T-i_T\mid X)).
\end{align*}
\endgroup 

Evidently, the calculation of the conditional likelihood function (\ref{eq: likelihood}) using this approach would require $2^T$ \textit{multi\-variate} copula evaluations, which is not computationally feasible. However, by employing the vine PCC introduced later in the paper, we will show that this number can be reduced to only $2T(T-1)$ \textit{bivariate} copula evaluations. The latter method provides us with flexibility, since any multivariate discrete distribution can be decomposed as a vine PCC under a set of conditions that are discussed in the following Section. 
\section{Pair copula constructions of point-optimal sign tests for predictive regressions \label{Point-optimal sign test based on PCC}}
{\hskip 1.5em}In this Section, we derive POS-based tests in the context of linear and nonlinear regression models based on vine PCC decomposition. Following a structure similar to \citet{dufour2010exact}, we first consider the problem of testing whether the conditional median of a vector of observations is zero against a linear regression alternative. We further consider the conditions under which the conditional likelihood function under the alternative hypothesis can be decomposed as a vine PCC, and as such, choose an appropriate vine model. These results are later generalized to test whether the coefficients of a possibly nonlinear median regression function have a given value against an alternative nonlinear median regression.  
\subsection{Testing orthogonality hypothesis in linear predictive regressions \label{Linear}}

{\hskip 1.5em}Consider the problem of testing the null hypothesis of unpredictability (\ref{eq: null}) against the alternative (\ref{eq: alt}), using the test statistic (\ref{eq: teststat}) and given the assumptions (\ref{eq: DGP}) and (\ref{eq: median}). As it was shown in Section \ref{Framework}, under the alternative hypothesis the conditional likelihood function can be expressed as
\begin{equation}\label{eq: likelihood1}
L_1(U(T),\bm{\beta}_1,X)=\prod\limits_{t=1}^{T}P\left[s(y_t)=s_t\mid \text{\b{S}}_{t-1}=\text{\b{s}}_{t-1},X\right].
\end{equation}
Let $s(y_j)$ be a scalar element of $\text{\b{S}}_{t-1}$, with $\text{\b{S}}_{t-1}^{\backslash j}=\text{\b{S}}_{t-1}\backslash s(y_j)$ such that
\[
 \text{\b{S}}_{t-1}^{\backslash j}=\left\{s(y_1),s(y_2),\cdots,s(y_{j-1}),s(y_{j+1}),\cdots,s(y_{t-1})\right\}
\] 
and $s(y_t)\notin \text{\b{S}}_{t-1}$. By choosing a single element of $\text{\b{S}}_{t-1}$, say $s(y_j)$, we would have
\begingroup
\allowdisplaybreaks
\begin{align}\label{eq: bayes}
\begin{split}
P\left[s(y_t)=s_t\mid\text{\b{S}}_{t-1}=\text{\b{s}}_{t-1},X\right]&=\frac{P[s(y_t)=s_t,s(y_j)=s_j\mid\text{\b{S}}_{t-1}^{\backslash j}=\text{\b{s}}_{t-1}^{\backslash j},X]}{P[s(y_j)=s_j\mid\text{\b{S}}_{t-1}^{\backslash j}=\text{\b{s}}_{t-1}^{\backslash j},X]}\\
&=\sum\limits_{k_t=0,1}\sum\limits_{k_j=0,1}(-1)^{k_t+k_j}\times\\
&\textcolor{white}{=}\left\{P[s(y_t)\leq s_t-k_t,s(y_j)\leq s_j-k_j\mid \text{\b{S}}_{t-1}^{\backslash j}=\text{\b{s}}_{t-1}^{\backslash j},X]\right\}\\
&\textcolor{white}{=}/P[s(y_j)=s_j\mid \text{\b{S}}_{t-1}^{\backslash j}=\text{\b{s}}_{t-1}^{\backslash j},X],
\end{split}
\end{align}
\endgroup
where the bivariate conditional probability in (\ref{eq: bayes}) can be expressed in terms of copulas as follows
\begingroup
\allowdisplaybreaks
\begin{align}\label{eq: copulas}
\begin{split}
P[s(y_t)=s_t\mid \text{\b{S}}_{t-1}=\text{\b{s}}_{t-1},X]&=\sum\limits_{k_t=0,1}\sum\limits_{k_j=0,1}(-1)^{k_t+k_j}\bigg\{\\
&\textcolor{white}{=}C_{s(y_t),s(y_j)\mid \text{\b{S}}_{t-1}^{\backslash j} }\left(F_{s(y_t)\mid\text{\b{S}}_{t-1}^{\backslash j}}(s_t-k_t\mid\text{\b{s}}_{t-1}^{\backslash j},X),F_{s(y_j)\mid\text{\b{S}}_{t-1}^{\backslash j}}(s_j-k_j\mid \text{\b{s}}_{t-1}^{\backslash j},X)\right)\bigg\}\\
&\textcolor{white}{=}/P[s(y_j)=s_j\mid \text{\b{S}}_{t-1}^{\backslash j}=\text{\b{s}}_{t-1}^{\backslash j},X].
\end{split}
\end{align}
\endgroup
Further, let $\text{\b{S}}_{t-1}^{\backslash i,j}=\text{\b{S}}_{t-1}^{\backslash j}\backslash s(y_ i)$, such that $s(y_ i)$ is a scalar element of $\text{\b{S}}_{t-1}^{\backslash j}$. Then the arguments $F_{s(y_t)\mid\text{\b{S}}_{t-1}^{\backslash j}}$ and $F_{s(y_j)\mid\text{\b{S}}_{t-1}^{\backslash j}}$ in copula expression (\ref{eq: copulas}) can be presented by the general form
\begingroup
\allowdisplaybreaks
\begin{align}\label{eq: copulas1}
\begin{split}
&F_{s(y_t)\mid s(y_{i}),\text{\b{S}}_{t-1}^{\backslash i,j}}(s_t-k_t\mid s_{i},\text{\b{s}}_{t-1}^{\backslash i,j},X)=\\
&\left\{C_{s(y_t),s(y_ i)\mid \text{\b{S}}_{t-1}^{\backslash i,j}}\left(F_{s(y_t)\mid \text{\b{S}}_{t-1}^{\backslash i,j}}(s_t-k_t\mid \text{\b{s}}_{t-1}^{\backslash i,j},X),F_{s(y_ i)\mid \text{\b{S}}_{t-1}^{\backslash i,j}}(s(y_ i)\mid \text{\b{s}}_{t-1}^{\backslash i,j},X)\right)\right.-\\
&\textcolor{white}{=}\left.C_{s(y_t),s(y_ i)\mid \text{\b{S}}_{t-1}^{\backslash i,j}}\left(F_{s(y_t)\mid \text{\b{S}}_{t-1}^{\backslash i,j}}(s_t-k_t\mid \text{\b{s}}_{t-1}^{\backslash i,j},X),F_{s(y_ i)\mid \text{\b{S}}_{t-1}^{\backslash i,j}}(s(y_ i)-1\mid \text{\b{s}}_{t-1}^{\backslash i,j},X)\right)\right\}\\
&/P[s(y_ i)=s_ i\mid \text{\b{S}}_{t-1}^{\backslash i,j}=\text{\b{s}}_{t-1}^{\backslash i,j},X].
\end{split}
\end{align}
\endgroup
Thus, decomposition (\ref{eq: copulas}), and in turn (\ref{eq: copulas1}) can be applied recursively to the elements of the conditional likelihood function (\ref{eq: likelihood}), such that it is expressed in terms of bivariate copulas. Let $\text{\b{S}}_{t-1}=\{s(y_1),\cdots,s(y_{t-1})\}$ be the variables that $s(y_t)$ for $t=2,\cdots,T$ is conditioned on. We follow \citet{joe2014dependence}, by letting $\underline{\sigma}_{t-1}=\{\sigma(1,t),\cdots,\sigma(t-1,t)\}$ be a permutation of $\text{\b{S}}_{t-1}$, such that $s(y_t)$ is paired sequentially first with $\sigma(1,t)$, then $\sigma(2,t)$ and finally $\sigma(t-1,t)$, where in the $r^{\text{th}}$ step $(2\leq r\leq t-1)$, $\sigma(r,t)$ is paired to $t$ conditional on $\sigma(1,t),\cdots,\sigma(r-1,t)$. For $n\leq3$ (i.e. $t=2,3$) there are only three possible permutations with $\underline{\sigma}_1=\{s(y_1)\}$ for $t=2$, and $\underline{\sigma}_2=\{s(y_1),s(y_2)\}$, as well as $\underline{\sigma}_2=(s(y_2),s(y_1))$ for $t=3$ respectively. Therefore, under assumptions (\ref{eq: DGP}) and (\ref{eq: median}%
), and with $T\leq 3$, let $H_{0}$ and $H_{1}$ be defined by (\ref{eq: null}) - (\ref%
{eq: alt}), then the Neyman-Pearson type test-statistic based on the signs $(s(y_1),\cdots,s(y_T))'$ can be expressed as
\begingroup
\allowdisplaybreaks
\begin{align*}
SL_{T}(\bm{\beta}_{1})=\ln P\left[s(y_1)=s_1\mid X\right]&+\sum\limits_{t=2}^{T}\ln\Delta_{s_t^-}^{s_t^+}\Delta_{s_{t-1}^-}^{s_{t-1}^+}C_{t,t-1\mid t-2}\\
&-\sum\limits_{t=2}^{T}\ln P[s(y_{t-1})=s_{t-1}\mid \text{\b{S}}_{t-2}=\text{\b{s}}_{t-2},X]-T\ln\left\{\frac{1}{2}\right\},
\end{align*}
\endgroup
for $t=2,\cdots,T$, where
\begingroup
\allowdisplaybreaks
\begin{align*}
\Delta_{s_t^-}^{s_t^+}\Delta_{s_{t-1}^-}^{s_{t-1}^+}C_{t,t-1\mid t-2}&=\sum\limits_{k_t=0,1}\sum\limits_{k_{t-1}=0,1}(-1)^{k_t+k_{t-1}}\\
&\textcolor{white}{\times}\times \left(C_{s(y_t),s(y_{t-1})\mid \text{\b{S}}%
_{t-2}}\left(F_{s(y_t)\mid \text{\b{S}}%
_{t-2} }(s_t-k_t\mid \text{\b{s}}%
_{t-2},X),F_{s(y_{t-1})\mid \text{\b{S}}%
_{t-2}}(s_{t-1}-k_{t-1}\mid \text{\b{s}}%
_{t-2},X)\right)\right)
\end{align*}
\endgroup
and such that
\begingroup
\allowdisplaybreaks
\begin{align*}
\ln P\left[s(y_1)=s_1\mid \text{\b{S}}_{0}=\text{\b{s}}_{0},X\right]=s(y_1)\ln \left\{\frac{P[y_1\geq 0\mid X]}{P[y_1< 0\mid X]}\right\}+\ln P[y_1< 0\mid X].
\end{align*}
\endgroup

For $T>3$, the permutations $\underline{\sigma}_{t-1}$ are dependent on the choice of the permutations at stages $3,\cdots,t-1$. Therefore, an issue that requires considerable attention is whether there \textit{exists} a decomposition such as the one considered in the earlier example for $T>3$. Furthermore, the conditional likelihood function expressed in terms of bivariate copulas by recursively using (\ref{eq: copulas}) and (\ref{eq: copulas1}), assumes that a single copula is specified for each conditional bivariate distribution $F_{s(y_t),s(y_j)\mid \text{\b{S}}_{t-1}^{\backslash j}}$ in decomposition (\ref{eq: likelihood1}) over all possible values of $\text{\b{S}}_{t-1}^{\backslash j}$. This implies that the copula is unique for the Cartesian product of the ranges of conditional CDFs $F_{s(y_t)\mid \text{\b{S}}_{t-1}^{\backslash j}}$ and $F_{s(y_j)\mid \text{\b{S}}_{t-1}^{\backslash j}}$. Therefore, the decomposition must be such that each conditional bivariate distribution in said vine has a constant conditional copula [see \citet{panagiotelis2012pair}]. For a constant conditional copula to exist, the following conditions outlined by \citet{panagiotelis2012pair} must be satisfied.

\begin{definition}[Existence of constant conditional copula] Consider the conditional bivariate distribution function $F_{s(y_t),s(y_j)\mid \text{\b{S}}_{t-1}^{\backslash j}}$. We say that a copula $C=C_{s(y_t),s(y_j)\mid\text{\b{S}}_{t-1}^{\backslash j}}$ is constant over all possible values of $\text{\b{S}}_{t-1}^{\backslash j}$ if
\[
\sum\limits_{m=0,1}\sum\limits_{n=0,1}(-1)^{m+n}C(a_{k-m},b_{l-n})\geq 0,\quad \forall k,l\in\{1,2\}\times\{1,2\},
\] 
where $a_0<a_1<a_2$ and $b_0<b_1<b_2$, are the distinct points corresponding to the ranges of the conditional Bernoulli CDFs $F_{s(y_t)\mid \text{\b{S}}_{t-1}^{\backslash j}}$ and $F_{s(y_j)\mid \text{\b{S}}_{t-1}^{\backslash j}}$ respectively, such that $a_0=b_0=0$ and $a_{2}=b_{2}=1$, and where further, the following constraints are satisfied:
\begingroup
\allowdisplaybreaks
\begin{align*}
&C_{s(y_t),s(y_j)\mid \text{\b{S}}_{t-1}^{\backslash j}}\left(a_{s(y_t)\mid\text{\b{S}}_{t-1}^{\backslash j}},b_{s(y_j)\mid\text{\b{S}}_{t-1}^{\backslash j}}\right)=P[s(y_t)\leq s_t,s(y_j)\leq s_j\mid \text{\b{S}}_{t-1}^{\backslash j}=\text{\b{s}}_{t-1}^{\backslash j},X],\\
&C_{s(y_t),s(y_j)\mid  \text{\b{S}}_{t-1}^{\backslash j}}\left(1,b_{s(y_j)\mid\text{\b{S}}_{t-1}^{\backslash j}}\right)=b_{s(y_j)\mid\text{\b{S}}_{t-1}^{\backslash j}},\quad C_{s(y_t),s(y_j)\mid  \text{\b{S}}_{t-1}^{\backslash j}}\left(a_{s(y_t)\mid\text{\b{S}}_{t-1}^{\backslash j}},1\right)=a_{s(y_t)\mid\text{\b{S}}_{t-1}^{\backslash j}},
\end{align*}
\endgroup
with $a_{s(y_t)\mid\text{\b{S}}_{t-1}^{\backslash j}}:=P[s(y_t)\leq s_t\mid \text{\b{S}}_{t-1}^{\backslash j}=\text{\b{s}}_{t-1}^{\backslash j},X]$ and $b_{s(y_j)\mid\text{\b{S}}_{t-1}^{\backslash j}}:=P[s(y_j)\leq s_j\mid \text{\b{S}}_{t-1}^{\backslash j}=\text{\b{S}}_{t-1}^{\backslash j},X]$.
\end{definition}

To satisfy the above conditions, the vine decomposition must be such that the strength of the dependence of the conditional bivariate distribution does not vary much across different values of the conditioning set [see \citet{panagiotelis2012pair}]. As we are dealing with time-series data, the D-vine decomposition yields a constant dependence structure over different values of $\text{\b{S}}_{t-1}^{\backslash j}$, and is thus, the most appropriate and intuitive choice for the decomposition of the conditional likelihood function (\ref{eq: likelihood1}). 

The D-vine PCC (figure \ref{fig: D-vineC3}) is constructed by $T-1$ trees, say $D=\{\mathcal{T}_1,\cdots,\mathcal{T}_{T-1}\}$, comprised of the edges $\xi(D)=\xi(\mathcal{T}_1)\cup\cdots\cup \xi(\mathcal{T}_{T-1})$, where $\xi(\mathcal{T}_l)$ refers to the edges of the tree $\mathcal{T}_l$. In the first tree $T_1$, the marginals $F(s_1\mid X),F(s_2\mid X),\cdots,F(s_T\mid X)$, are arranged as nodes in consecutive order, say $N(\mathcal{T}_1):=\{1,2,\cdots,T-1,T\}$, where the nodes are of degree two, meaning that no more than two edges is connected to each node. The corresponding edges join the adjacent nodes, such that $\xi(\mathcal{T}_1):=\{12,23,\cdots,(T-1,T)\}$. Next, the edges of the first tree $\xi(\mathcal{T}_1)$ become the nodes of $\mathcal{T}_2$, a process which is completed in a recursive manner, such that $N(\mathcal{T}_{l+1})=\xi(\mathcal{T}_l)$, with the edges of each tree joining the adjacent nodes, and with the mutual elements between the nodes becoming the conditioning set.  
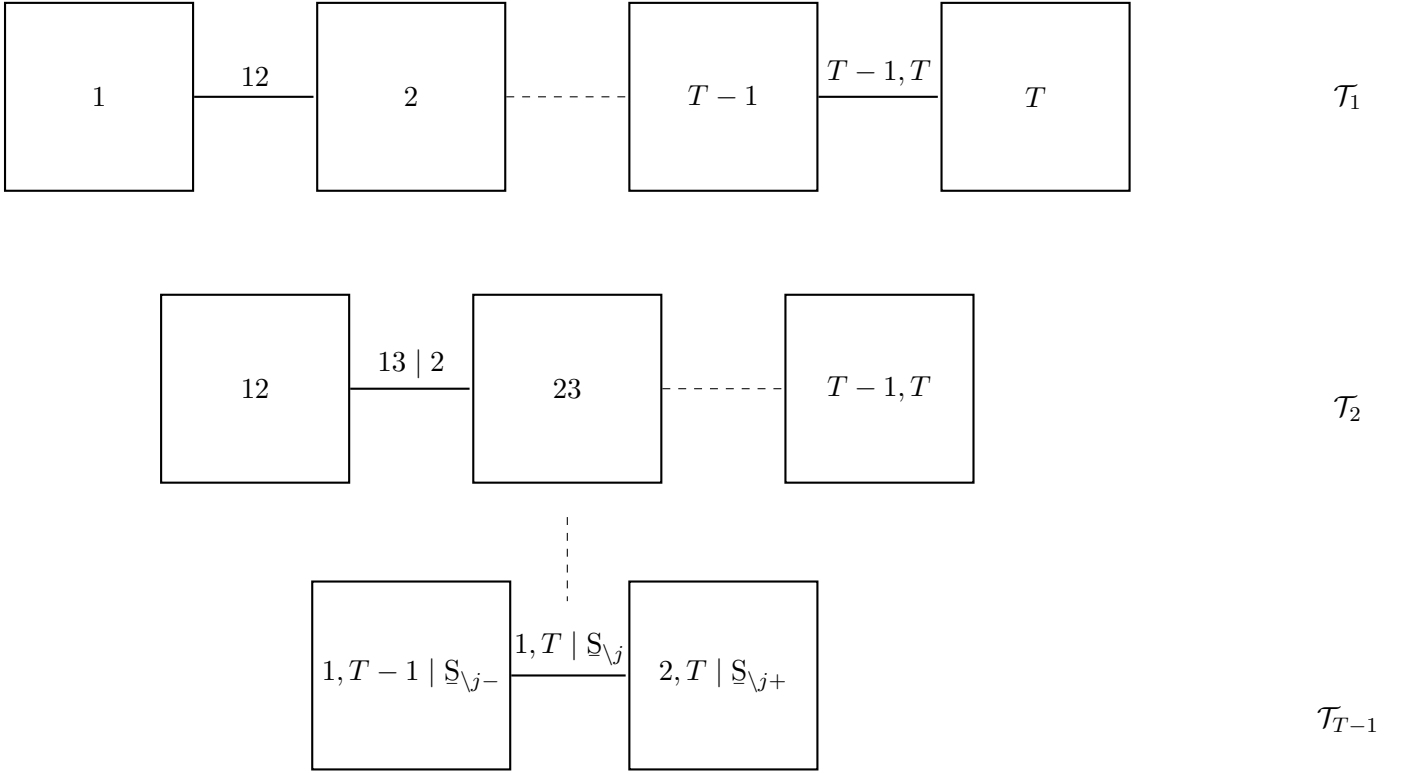
\begin{figure}[hbtp!]
\caption{D-vine PCC for the $T$-variate case}
\begin{center}
\begin{tikzpicture}[shorten >=1pt,node distance=4.15cm,auto]
\tikzstyle{state}=[shape=rectangle,thick,draw,minimum size=2.5cm]

\node[state] (1) {$1$};
\node[state,right of=1] (2) {$2$};
\node[state,right of=2] (T-1) {$T-1$};
\node[state,right of=T-1] (T) {$T$};
\node[right of=T] (T1) {$\mathcal{T}_1$};

\path[draw,thick]
(1) edge node (12) {$12$} (2) 
(T-1) edge node (T-1T){$T-1,T$} (T);

\path [draw,dashed]
(2) edge node (2T-1) {$$} (T-1);

\node[state,below of=12] (12T) {$12$};
\node[state,right of=12T] (23T) {$23$};
\node[state,right of=23T] (T-1TT) {$T-1,T$};
\node[below of=T1] (T2) {$\mathcal{T}_2$};

\path[draw,thick]
(12T) edge node (132) {$13\mid 2$} (23T); 

\path [draw,dashed]
(23T) edge node (34T-1T) {$$} (T-1TT);

\node[state,below of=132] (l1) {$1,T-1\mid \text{\b{S}}_{\backslash j-}$};
\node[state,right of=l1] (l2) {$2,T\mid \text{\b{S}}_{\backslash j+}$};
\node[below of=T2] (TT) {$\mathcal{T}_{T-1}$};
\path[draw,thick]
(l1) edge node (last) {$1,T\mid \text{\b{S}}_{\backslash j} $} (l2); 

\path[draw,dashed]
($($ (23T) !0.5! ($ (l1) !0.5! (l2) $) $)!0.20cm!(23T)$) edge node[near end] {$$} ($ ($ (l1) !0.5! (l2) $) !0.25! (23T) $); 

\end{tikzpicture}
\end{center}
\doublespacing
Note: D-vine for a sample size $T$ consists of $T-1$ trees. The first tree consists of the conditional marginals ordered consecutively as nodes, with the edges connecting the adjacent nodes, and with the elements shared by the two nodes going in the conditioning set. The edges of each tree $\mathcal{T}_l$ become the nodes of the tree $\mathcal{T}_{l+1}$. In this figure, $\text{\b{S}}_{\backslash i-}$ and  $\text{\b{S}}_{\backslash j+}$ correspond to the elements $\text{\b{S}}_{\backslash j-}:=\{2,3,\cdots,T-2\}$ and $\text{\b{S}}_{\backslash j+}:=\{3,\cdots,T-1\}$ respectively, with $\text{\b{S}}_{\backslash j}:=\{2,3,\cdots,T-1\}$.
\label{fig: D-vineC3}
\end{figure}

To express the conditional likelihood function (\ref{eq: likelihood1}) as a D-vine, we begin by calculating the conditional marginals $F_1,\cdots,F_T$, where $F_t\sim Bernoulli(p_t)$, for $t=1,\cdots,T$, with conditional CDFs that are expressed as (\ref{eq: BernoulliCDF}). Therefore, under assumptions (\ref{eq: DGP}) and (\ref{eq: median}), we have
\begin{equation}\label{eq: BernoulliCDF2}
F_t(s_t\mid X)=\left\{ 
\begin{tabular}{lll}
$0,$ &$\text{for}$& $s_t< 0$ \\ 
$1-P[\varepsilon_t\geq  -\beta' x_{t-1}\mid X],$ &$\text{for}$ &$0\leq s_t< 1$ \\
$1,$& $\text{for}$& $s_t\geq 1$
\end{tabular}%
\right.,\quad t=1,\cdots,T.
\end{equation} 
\begin{sloppypar}
Once the marginals are obtained, the next step consists of evaluating the copulas in the first tree - i.e. $C_{12}(F_1,F_2),\cdots,C_{T-1,T}(F_{T-1},F_T)$, corresponding to the edges $\xi(\mathcal{T}_1)$). In the second tree, the copulas $C_{13\mid 2}(F_{1\mid 2},F_{3\mid 2}),\cdots,C_{T-2,T \mid T-1}(F_{T-2\mid T-1},F_{T\mid T-1})$ are evaluated, then $C_{14\mid 23}(F_{1\mid 23},F_{4\mid 23}),\cdots,C_{T-3,T\mid T-2,T-1}(F_{T-3\mid T-2,T-1},F_{T\mid T-2,T-1})$ in the third tree, and so on. 
\end{sloppypar}

In the case of continuous variables, say $\{s^*(y_t)\in \mathbb{R},t=1,\cdots,T\}$, the construction of the D-vine involves an iterative copula evaluation process for the trees $\mathcal{T}_1,\cdots,\mathcal{T}_{T-1}$. This leads to $T(T-1)/2$ copula evaluations, which correspond to one copula evaluation for each edge [see Appendix]. On the other hand, for discrete variables, the conditional p.m.fs are expressed as in (\ref{eq: copulas}), which requires the evaluation of the following four copulas
\begingroup
\allowdisplaybreaks
\begin{align*}
C_{t,j\mid \mathbf{\backslash j}}^{++}(F_{t\mid \mathbf{\backslash j}}^+,F_{j\mid \mathbf{\backslash j}}^+),\quad C_{t,j\mid \mathbf{\backslash j}}^{+-}(F_{t\mid \mathbf{\backslash j}}^+,F_{j\mid \mathbf{\backslash j}}^-),\\
C_{t,j\mid \mathbf{\backslash j}}^{-+}(F_{t\mid \mathbf{\backslash j}}^-,F_{j\mid \mathbf{\backslash j}}^+),\quad C_{t,j\mid \mathbf{\backslash j}}^{--}(F_{t\mid \mathbf{\backslash j}}^-,F_{j\mid \mathbf{\backslash j}}^-),
\end{align*}
\endgroup
where $F_{t\mid \mathbf{\backslash j}}^+=P[s(y_t)\leq s_t \mid\text{\b{S}}_{t-1}^{\backslash j}=\text{\b{s}}_{t-1}^{\backslash j},X]$ and $F_{t\mid \mathbf{\backslash j}}^-=P[s(y_t)\leq s_t-1 \mid\text{\b{S}}_{t-1}^{\backslash j}=\text{\b{s}}_{t-1}^{\backslash j},X]$. Henceforth, $4\times T(T-1)/2$ bivariate copulas need to be evaluated in the case of discrete data. 

Let us express the conditional joint p.m.f of the signs as follows
\begin{equation}
P_1[s(y_1)=s_1\mid X]\times\prod\limits_{t=2}^{T}P_{t\mid 1:{t-1}}[s(y_t)=s_t\mid s(y_1)=s_1,\cdots,s(y_{t-1})=s_{t-1},X],
\end{equation}
where following the results by \citet{stoeber2013simplified}, if the D-vine is expressed as a vine-array $A=(\sigma_{lt})_{1\leq l\leq t\leq T}$, where $l=2,\cdots,T-1$ is the row with tree $\mathcal{T}_l$, and column $t$ has the permutation $\underline{\sigma}_{t-1}=(\sigma(1,t),\cdots,\sigma(t-1,t))$ of the previously added variables, such that
\begin{equation*}
\begin{bmatrix}
- & 12 & 23 &34 &\cdots & T-1,T\\
 &-&13\mid2&24\mid 3&\cdots&T-2,T\mid T-1\\
& & \ddots&\cdots&\cdots&\vdots\\
&&&-&1,T-1\mid \text{\b{S}}_{\backslash j-}&2,T \mid \text{\b{S}}_{\backslash j+}\\
&&&&-&1,T\mid \text{\b{S}}_{\backslash j}\\
&&&&&-
\end{bmatrix},\quad
A=
\begin{bmatrix}
1 & 1 & 2 &3 &\cdots & T-1\\
 &2&1&2&\cdots&T-2\\
& & \ddots&\cdots&\cdots&\vdots\\
&&&T-2&1&2\\
&&&&T-1&1\\
&&&&&T
\end{bmatrix}
\end{equation*}
then 
\begin{equation}\label{eq: stober}
P_{t\mid 1:{t-1}}[s(y_t)=s_t\mid s(y_1)=s_1:s(y_{t-1})=s_{t-1}, X]=\left\{\prod\limits_{l=t-1}^{2} c_{\sigma_{lt}t,\mid \sigma_{1t},\cdots,\sigma_{t-1,t}}\right\}\times c_{\sigma_{lt}t}\times P_t[s(y_t)=s_t\mid X],
\end{equation}
where following \citet{joe2014dependence}, the copula densities in expression (\ref{eq: stober}) are calculated by
\begin{equation}
c_{t,j\mid \mathbf{\backslash j}}=\frac{C_{t,j\mid \mathbf{\backslash j}}^{++}(F_{t\mid \mathbf{\backslash j}}^+,F_{j\mid \mathbf{\backslash j}}^+)-C_{t,j\mid \mathbf{\backslash j}}^{-+}(F_{t\mid \mathbf{\backslash j}}^-,F_{j\mid \mathbf{\backslash j}}^+)-C_{t,j\mid \mathbf{\backslash j}}^{+-}(F_{t\mid \mathbf{\backslash j}}^+,F_{j\mid \mathbf{\backslash j}}^-)+C_{t,j\mid \mathbf{\backslash j}}^{--}(F_{t\mid \mathbf{\backslash j}}^-,F_{j\mid \mathbf{\backslash j}}^-)}{P_{t\mid \mathbf{\backslash j}}[s(y_t)=s_t\mid\text{\b{S}}_{t-1}^{\backslash j}=\text{\b{s}}_{t-1}^{\backslash j},X]P_{j\mid \mathbf{\backslash j}}[s(y_j)=s_j\mid\text{\b{S}}_{t-1}^{\backslash j}=\text{\b{s}}_{t-1}^{\backslash j},X]},
\end{equation}
which leads to the following proposition.
\begin{proposition}\label{coroll1}
Let $A=(\sigma_{lt})_{1\leq l\leq t\leq T}$ be a D-vine array for the signs $s(y_1),\cdots,s(y_T)$. Under assumptions (\ref{eq: DGP}) and (\ref{eq: median}%
), let $H_{0}$ and $H_1$ be defined by (\ref{eq: null}) - (\ref%
{eq: alt})$,$ 
\begingroup
\allowdisplaybreaks
\begin{align*}
SL_{T}(\bm{\beta}_{1})=\sum\limits_{t=2}^{T}\sum\limits_{l=t-1}^{2}\ln c_{\sigma_{lt}t,\mid \sigma_{1t},\cdots,\sigma_{t-1,t}}+\sum\limits_{t=2}^{T}\ln c_{\sigma_{1t}t}+\sum\limits_{t=1}^{T}s(y_t)a_t(\bm{\beta}_1)>c_1(\bm{\beta}_1),
\end{align*}%
\endgroup
where 
\[
a_t(\bm{\beta}_1)=\ln\left\{\frac{1-P_t[\varepsilon_t\leq -\bm{\beta}' \bm{x}_{t-1}\mid X]}{P_t[\varepsilon_t\leq -\bm{\beta}'\bm{ x}_{t-1}\mid X]}\right\},
\]
and suppose the constant $c_{1}(\bm{\beta}_{1})$ satisfies $P\left[SL_T(\bm{\beta}_1)>c_{1}(\bm{\beta}_1)\right] =\alpha $
under $H_{0},$ with $0<\alpha <1.$ Then the test that rejects $H_{0}$ when 
\begin{equation}
SL_{T}(\bm{\beta}_{1})>c_{1}(\bm{\beta}_1) \label{eq: SLn critical region}
\end{equation}%
is most powerful for testing $H_{0}$ against $H_1$ among level-$\alpha $
tests based on the signs $\big(s(y_{1}),\cdots,s(y_{T})\big)%
'.$
\end{proposition}

Under the null hypothesis, the signs $s(y_1),\cdots,s(y_T)$ are i.i.d. according to Bernoulli $Bi(1,0.5)$, with the distribution of $SL_T(\bm{\beta}_1)$ only depending on the weights $a_t(\bm{\beta}_1)$, without the presence of any nuisance parameters. Assumption (\ref{eq: median}) implies that tests based on $SL_{T}(\bm{\beta}_{1})$, such as the test given by (\ref{eq: SLn critical region}), are distribution-free and robust against heteroskedasticity of unknown form. On the other hand, under the alternative hypothesis, the power function of the test depends on the form of the distribution of $\varepsilon_t$. A special case is where $\varepsilon_1,\cdots,\varepsilon_T$ are independently distributed according to $N(0,1)$, which leads to the optimal test statistic assuming the following form
\[
SL_{T}(\bm{\beta} _{1})=\sum\limits_{t=2}^{T}\sum\limits_{l=t-1}^{2}\ln c_{\sigma_{lt}t,\mid \sigma_{1t},\cdots,\sigma_{t-1,t}}+\sum\limits_{t=2}^{T}\ln c_{\sigma_{1t}t}+\sum\limits_{t=1}^{T}s(y_t)a_t(\bm{\beta}_1)>c_1(\bm{\beta}_1),
\]
where 
\[
a_t(\bm{\beta}_1)=\ln\left\{\frac{\Phi(\bm{\beta}' \bm{x}_{t-1})}{1-\Phi(\bm{\beta}' \bm{x}_{t-1})}\right\},
\]
where $\Phi(.)$ is the standard normal distribution function. The distribution of $SL_{T}(\bm{\beta}_{1})$ can be simulated under the null hypothesis with sufficient number of replications, and the critical values can be obtained to any degree of precision.
\subsection{Testing general full coefficient hypothesis in nonlinear predictive regressions \label{nonlinear}}
{\hskip 1.5em}We now consider the nonlinear predictive regression model%
\begin{equation}
y_{t}=f(\bm{x}_{t-1},\bm{\beta} )+\varepsilon_{t},\text{ }t=1,\cdots,\,T\text{,}
\label{eq: modelnl}
\end{equation}%
where $\bm{x}_{t-1}$ is a $(k+1)\times 1$ vector of stochastic
explanatory variables, such that $\bm{x}_{t-1}=[1,x_{1,t-1},\cdots,x_{k,t-1}]'$, $f(\,\cdot \,)$ is a scalar function, $\bm{\beta} \in 
\mathbb{R}^{(k+1)}$ is an unknown vector of parameters and
\begin{equation*}
\varepsilon_{t}\mid X \sim F_{t}(.\mid X)
\end{equation*}%
where as before $F_{t}(.\mid X)$ is a distribution function and $X=[\bm{x}_0',\cdots,\bm{x}_{T-1}']$ is an $T\times (k+1)$ matrix. Suppose that the error process $\{\varepsilon_t,t=1,2,\cdots\}$ is a strict conditional mediangale, such that

\begin{equation}\label{eq: mediane1}
P[\varepsilon_{t}> 0\mid \bm{{\varepsilon}}_{t-1},X]=P[\varepsilon_{t}<0\mid \bm{\varepsilon}_{t-1},X]=\frac{1}{2},
\end{equation}%
with
\[
\bm{\varepsilon}_{0}=\{\emptyset\},\quad\bm{\varepsilon}_{t-1}=\{\varepsilon_1,\cdots,\varepsilon_{t-1}\},\quad\text{for}\quad t\geq2
\]
and where (\ref{eq: mediane1}) entails that $\varepsilon_t\mid X$ has no mass at zero, \emph{i.e.} $P[\varepsilon_t=0\mid X]$=0 for all $t$. We do not require that the
parameter vector $\bm{\beta}$ be identified. 

We consider the problem of testing the null hypothesis
\begin{equation}\label{eq: nullnl}
H(\bm{\beta}_0):\bm{\beta}=\bm{\beta}_0,
\end{equation}
against the alternative hypothesis
\begin{equation}\label{eq: altnl}
H(\bm{\beta}_1):\bm{\beta}=\bm{\beta}_1,
\end{equation}
We construct a test statistic for testing $H(\beta_0)$ against $H(\beta_1)$ in a similar manner to Section \ref{Linear}, by transforming model (\ref{eq: modelnl}) to 
\[
\tilde{y}_t=g(\bm{x}_{t-1},\bm{\beta},\bm{\beta}_0)+\varepsilon_t,\quad t=1,\cdots,T
\]
where $\tilde{y}_t=y_t-f(\bm{x}_{t-1},\bm{\beta}_0)$ and $g(\bm{x}_{t-1},\bm{\beta},\bm{\beta}_0)=f(\bm{x}_{t-1},\bm{\beta})-f(\bm{x}_{t-1},\bm{\beta}_0)$. Notice that testing $H(\bm{\beta}_0)$ against $H(\bm{\beta}_1)$ is equivalent to testing
\[
\bar{H}_0:g(\bm{x}_{t-1}, \bm{\beta},\bm{\beta}_0)=\bm{0},\quad\text{for}\quad t=1,\cdots,T
\]
where $\bm{0}$ is a $(k+1)\times 1$ zero vector, against the alternative
\[
\bar{H}_A: g(\bm{x}_{t-1},\bm{\beta},\bm{\beta}_0)=f(\bm{x}_{t-1},\bm{\beta}_1)-f(x_{t-1},\bm{\beta}_0),\quad\text{for}\quad t=1,\cdots,T.
\]
For $\tilde{U}(T)=(s(\tilde{y}_1),\cdots,s(\tilde{y}_T))'$, where for $1\leq t\leq T$
\begin{equation*}
s(\tilde{y}_{t})=\left\{ 
\begin{tabular}{l}
$1,$ $if$ $\tilde{y}_{t}\geq 0$ \\ 
$0,$ $if$ $\tilde{y}_{t}<0$%
\end{tabular}%
\ \right. \text{.}
\end{equation*}
As before, the conditional joint p.m.f of the process of signs is expressed as
\begin{equation}\label{eq: IDKANYMORE}
P_1[s(\tilde{y}_1)=\tilde{s}_1\mid X]\times\prod\limits_{t=2}^{T}P_{t\mid 1:{t-1}}[s(\tilde{y}_t)=\tilde{s}_t\mid s(\tilde{y}_1)=\tilde{s}_1,\cdots,s(\tilde{y}_{t-1})=\tilde{s}_{t-1},X].
\end{equation}
Furthermore, the D-vine-array $\tilde{A}=(\tilde{\sigma}_{lt})_{1\leq l\leq t\leq T}$, is such that $l=2,\cdots,T-1$ is the row with tree $\mathcal{T}_l$, and column $t$ has the permutation $\tilde{\underline{\sigma}}_{t-1}=(\tilde{\sigma}_{1t},\cdots,\tilde{\sigma}_{t-1,t})$ of the previously added variables. Then
\begin{equation}\label{eq: stober2}
P_{t\mid 1:{t-1}}[s(\tilde{y}_t)=\tilde{s}_t\mid s(\tilde{y}_1)=\tilde{s}_1,\cdots,s(\tilde{y}_{t-1})=\tilde{s}_{t-1},X]=\left\{\prod\limits_{l=t-1}^{2} c_{\tilde{\sigma}_{lt}t,\mid \tilde{\sigma}_{1t},\cdots,\tilde{\sigma}_{t-1,t}}\right\}\times c_{\tilde{\sigma}_{lt}t}\times P_t[s(\tilde{y}_t)=\tilde{s}_t\mid X],
\end{equation}
which leads to the following corollary.
\begin{corollary}\label{coroll2}
 Let $\tilde{A}=(\tilde{\sigma}_{lt})_{1\leq l\leq t\leq T}$ be a D-vine array for the signs $s(\tilde{y}_1),\cdots,s(\tilde{y}_T)$. Under assumptions (\ref{eq: modelnl}) and (\ref{eq: median}%
), let $H(\bm{\beta}_{0})$ and $H(\bm{\beta}_{1})$ be defined by (\ref{eq: nullnl}) - (\ref%
{eq: altnl})$,$ 
\begingroup
\allowdisplaybreaks
\begin{align*}
SN_{T}(\bm{\beta}_0\mid\bm{\beta}_{1})=\sum\limits_{t=2}^{T}\sum\limits_{l=t-1}^{2}\ln c_{\tilde{\sigma}_{lt}t,\mid \tilde{\sigma}_{1t},\cdots,\tilde{\sigma}_{t-1,t}}+\sum\limits_{t=2}^{T}\ln c_{\tilde{\sigma}_{1t}t}+\sum\limits_{t=1}^{T}s(y_t-f(\bm{x}_{t-1},\bm{\beta}_0))\tilde{a}_t(\bm{\beta}_0\mid\bm{\beta}_1)>c_1(\bm{\beta}_0,\bm{\beta}_1),
\end{align*}%
\endgroup
where 
\[
\tilde{a}_t(\bm{\beta}_0\mid\bm{\beta}_1)=\ln\left\{\frac{1-p_t[\bm{x}_{t-1},\bm{\beta}_0,\bm{\beta}_1\mid X]}{p_t[\bm{x}_{t-1},\bm{\beta}_0,\bm{\beta}_1\mid X]}\right\},\quad p_t[\bm{x}_{t-1},\bm{\beta}_0,\bm{\beta}_1\mid X]=P_t[\varepsilon_t\leq f(\bm{x}_{t-1},\bm{\beta}_0)-f(\bm{x}_{t-1},\bm{\beta}_1)\mid X]
\]
and suppose the constant $c_{1}(\bm{\beta}_0,\bm{\beta}_{1})$ satisfies the constraint $P\left[SN_T(\bm{\beta}_0\mid\bm{\beta}_1)>c_{1}(\bm{\beta}_0,\bm{\beta}_1)\right] =\alpha $
under $H(\bm{\beta}_{0}),$ with $0<\alpha <1.$ Then the test that rejects $H(\bm{\beta}_{0})$ when 
\begin{equation}
SN_{T}(\bm{\beta}_0\mid\bm{\beta}_{1})>c_{1}(\bm{\beta}_0,\bm{\beta}_1) \label{eq: SLn critical region}
\end{equation}%
is most powerful for testing $H(\bm{\beta}_{0})$ against $H(\bm{\beta}_{1})$ among level-$\alpha $
tests based on the signs $\big(s(\tilde{y}_{1}),\cdots,s(\tilde{y}_{T})\big)%
'.$

\end{corollary}

Consider a linear function $f(\bm{x}_{t-1},\bm{\beta})=\bm{\beta}'\bm{x}_{t-1}$, and assume that under the alternative hypothesis $\varepsilon_t$ for $t=1,\cdots,T$ follows a standard normal distribution (i.e. $\varepsilon_t\sim N(0,1))$. Then the statistic for testing $H(\bm{\beta}_0)$ against the alternative $H(\bm{\beta}_1)$ is given by
\begingroup
\allowdisplaybreaks
\begin{align*}
SN_{T}(\bm{\beta}_0\mid\bm{\beta} _{1})=\sum\limits_{t=2}^{T}\sum\limits_{l=t-1}^{2}\ln c_{\tilde{\delta}_{lt}t,\mid \tilde{\delta}_{1t},\cdots,\tilde{\delta}_{t-1,t}}+\sum\limits_{t=2}^{T}\ln c_{\tilde{\delta}_{1t}t}+\sum\limits_{t=1}^{T}s(y_t-\bm{\beta}_0'\bm{x}_{t-1})\tilde{\delta}_t(\bm{\beta}_0\mid\bm{\beta}_1)>c_1(\bm{\beta}_0,\bm{\beta}_1),
\end{align*}%
\endgroup
where 
\[
\tilde{a}_t(\bm{\beta}_0\mid\bm{\beta}_1)=\ln\left\{\frac{\Phi((\bm{\beta}_1-\bm{\beta}_0)'\bm{x}_{t-1})}{1-\Phi((\bm{\beta}_1-\bm{\beta}_0)'\bm{x}_{t-1})}\right\},
\]
such that $\Phi(.)$ is the standard normal distribution function. As in Section \ref{Point-optimal sign test based on PCC}, the distribution of $SN_{T}(\bm{\beta}_0\mid\bm{\beta} _{1})$ can be simulated under the null hypothesis with sufficient number of replications and the relevant critical values can be obtained to any degree of precision.
\section{Estimation \label{EstimationC3}}
{\hskip 1.5em}In this Section, we first consider the issue of estimating the bivariate copulas in the D-vine decomposition and suggest a sequential estimation strategy for the parameters of the copulas. We then turn our attention to the problem of selecting a class of parametric bivariate copulas. The choice of the latter has an important implication on introducing dependence to the vector of signs.  
\subsection{Sequential estimation of the D-vine}

{\hskip 1.5em}The calculation of the test statistics in Section \ref{Point-optimal sign test based on PCC} requires four bivariate copula evaluations at $T(T-1)/2$ distinct points, leading to a total of $2T(T-1)$ copula evaluations. The estimation of the D-vine is often facilitated with the maximum likelihood (MLE hereafter). However, since the latter requires optimization with respect to at least $2T(T-1)$ copula parameters, sequential estimation procedures are favored for faster computation times, with the caveat that the increased speed comes at the cost of efficiency. Furthermore, the sequential estimates may be provided as starting points for the simultaneous numerical optimization using MLE [see \citet{czado2012maximum}, \citet{haff2012comparison} and \citet {dissmann2013selecting} among others]. We assume that the copulas are specified parametrically, given by an appropriate parameter (vector). More specifically, let $\pmb{\theta}_{l}=(\theta_{1,k}',\cdots,\theta_{T-l,k}')'$ be the set of all the parameters to be estimated for tree $\mathcal{T}_l$, $l=1,\cdots,T-1$ of the D-vine, with $k=l-1$ conditioning variables. Therefore, $\pmb{\theta}=(\pmb{\theta}_1',\cdots,\pmb{\theta}_{T-1}')'$ is the entire set of the parameters that need to be estimated for the D-vine decomposition. To estimate the parameter vector $\pmb{\theta}$, we follow a sequential estimation strategy proposed by \citet{czado2012maximum}, whereby first, the parameters of the unconditional bivariate copulas are estimated. These parameters are then utilized as means of estimating the parameters of bivariate copulas with a single conditioning variable. The latter are then used to estimate the pair-copulas with two conditioning variables, and so on. This bivariate copula estimation approach is continued sequentially until all parameters are estimated. 

In the first step, the marginals are obtained by computing the conditional Bernoulli CDFs (\ref{eq: BernoulliCDF2}) using an arbitrary distribution, such as the standard normal distribution considered in Section \ref{Point-optimal sign test based on PCC}. The second step of the process involves estimating the parameters of the unconditional copula, by fixing the marginals with their aforementioned estimates and maximizing the bivariate likelihood corresponding to each copula in each tree $\mathcal{T}_l$ to obtain $\pmb{\hat{\theta}}_l=(\hat{\theta}_{1,k}',\cdots,\hat{\theta}_{T-l,k}')$ for $l=1,\cdots,T-1$ and $k=l-1$ . As all the variables are discrete, the log-likelihood function, say, for the unconditional copula $C_{t,t+1}$ for $t=1,\cdots,T-1$ for the signs $\left(s(y_{i,t}),s(y_{i,t+1})\right)$, $i=1,\cdots,n-1$ is expressed as 
\[
L(\theta_{t,0})=\sum\limits_{i=1}^{T-1}\log\left\{\sum\limits_{\{a_1,a_2\}\in \{-,+\}^{2}}(-1)^{a_{j}}C_{t,t+1}\left(F_t(s_{i,t}^{a_1}\mid X;\mathbf{\hat{\bm{\beta}}_1}),F_{t+1}(s_{i,t+1}^{a_2}\mid X;\mathbf{\hat{\bm{\beta}}_1});\theta_{t,0}\right)\right\}.
\]
The estimate of the copula parameter, $\hat{\theta}_{t,0}$ for $t=1,\cdots,T-1$, is then obtained as follows
\[
\hat{\theta}_{t,0}=\arg \max_{\theta_{t,0}} L(\mathbf{\theta}_{t,0}),
\] 
which under regularity conditions solves
\[
\frac{\partial L(\mathbf{\theta}_{t,0})}{\partial \theta_{t,0}}=0.
\]

Let us illustrate this process with an example: once the marginals are obtained, the next step involves estimating the parameters $\theta_{t,0}$ for $t=1,\cdots,T-1$ of the unconditional copulas. Next, we are interested in estimating $\theta_{t,1}$ for $t=1,\cdots,T-2$. Define
\[
\hat{u}_{t\mid t+1}=F_{t\mid t+1}\left(s_t\mid s_{t+1},X;\hat{\theta}_{t,0} \right),
\]  
and
\[
\hat{v}_{t+2\mid t+1}=F_{t+2\mid t+1}\left(s_{t+2}\mid s_{t+1},X;\hat{\theta}_{t+1,0} \right),
\]
for $t=1,\cdots,T-2$. The data $\hat{u}_{t\mid t+1}$ and $\hat{v}_{t+2\mid t+1}$ is then used to estimate the parameters $\theta_{t,1}$ for $t=1,\cdots,T-2$, denoted by $\hat{\theta}_{t,1}$. This procedure is repeated sequentially until all parameters are estimated. \citet{haff2010simplified} show that under regularity conditions, the sequential estimates are asymptotically normal; however, as noted earlier their asymptotic covariance is \textquotedblleft intractable\textquotedblright{ }and the faster computation time comes at the cost of efficiency. Therefore, the sequential estimates can be utilized as the starting values of the high-dimensional MLE.

Another approach for estimating one-parameter pair-copulas in the sequential estimation procedure for copula families with a known relationship to Kendall's $\tau$ consists of inverting the empirical Kendall's $\tau$ based on, say, $\hat{u}_{t}$ and $\hat{u}_{t+1}$ for $t=1,\cdots,T-1$ for the edges of the first tree. However, we provide a caveat that the Kendall's $\tau$ of discrete data does not correspond to the Kendall's $\tau$ of the bivariate copulas [see \citet{denuit2005constraints}]. \citet{denuit2005constraints} show that by continuous extension of the discrete variables with a perturbation with values in $[0,1]$, the continuous features of Kendall's $\tau$ are adaptable to discrete data. In other words
\[ 
\tau(s^*(y_{t}),s^*(y_{t+1}))=4\int\int_{[0,1]^2}C_{t,t+1}^*\left(\hat{u}_{t},\hat{u}_{t+1}\right)dC_{t,t+1}^*\left(\hat{u}_{t},\hat{u}_{t+1}\right)-1
\]
for $t=1,\cdots,T-1$, such that $\hat{u}_{t}=F\left(s^*_t\mid X;\hat{\bm{\beta}}_1\right)$,
\[
s^*(y_t)=s(y_t)+U-1
\]
where $U$ is a continuous random variable in $[0,1]$. A natural choice for $U$ is the uniform distribution.
\subsection{Selection of the copula family \label{Selection of the copula family}}
{\hskip 1.5em}Many different classes of parametric bivariate copulas have extensively been studied and reviewed by \citet{joe2014dependence}. These include the Archimedean, elliptical, extreme value or max-id copula families that can specify the dependence structure of the vector of signs. As the dependence is introduced by the copula family, the type and the degree of the dependence between the signs depends on the choice of the copula. The literature surrounding the goodness-of-fit of copulas is extensive and has been analyzed by \citet{genest2006goodness}, \citet{genest2009goodness}, and \citet{berg2009copula}, among many others. \citet{genest2009goodness} categorize goodness-of-fit tests into three broad categories: procedures for testing particular dependence structures such as Gaussian or Clayton family; procedures that may be used for any classes of copulas, but require a strategic choice for their implementation; and finally, the so-called \textquotedblleft blanket tests\textquotedblright that apply to all classes of copulas and require no strategic choice for their use. A simple procedure proposed by \citet{joe1997multivariate} involves specifying the Akaike information critetion (AIC) to different copulas and using it as a copula selection criterion, which is particularly attractive as it allows for the automation of the copula selection process [see. \citet{czado2012maximum}]. The AIC specified to the copulas of, say, the first tree of the D-vine, can be expressed as follows
\[
AIC=-2\sum\limits_{i=1}^{t}\log c_{t,t+1}(\hat{u}_{i,t},\hat{u}_{i,t+1};\hat{\theta}_{1,k})+2l
\] 
for $t=1,.\cdots,T-1$ and $k=1,\cdots,T-1$, and where $l$ is the number of parameters $\theta_{1,k}$. \citet{panagiotelis2012pair} suggest that while dependence structures such as tail dependence are weak in discrete data, the choice of the copulas could still have a significant effect on the joint pmf of the signs. They considered Gaussian, Clayton and Gumbel copulas in constructing the D-vines for Bernoulli margins by keeping the marginal probabilities and dependence constant, and have found that in the case where the probabilities of zero marginals and joint probabilities in the data is high, preference goes to the use of the Gumbel copula over the other two alternatives. 

Within the context of our work, the mediangale assumption (\ref{eq: median}) implies that the signs $s(y_1),\cdots,s(y_T)$  conditional on $X$, and in turn $\hat{u}_1,\cdots,\hat{u}_T$ exhibit serial nonlinear dependence. Thus, the issue with specifying a copula-based model for the signs is that the distribution of the signs must imply an identity correlation matrix, where independence is only sufficient for uncorrelatedness. The literature generally deals with this issue by capturing the serial nonlinear dependence by imposing the identity correlation matrix on a multivariate Student's $t$ distribution. Henceforth, we consider the \textquotedblleft jointly symmetric\textquotedblright{ }copulas proposed by \citet{oh2016high}, where the latter can be constructed with any given (possibly asymmetric) copula family. In addition, when they are combined with symmetric marginals, they ensure an identity correlation matrix. A \textquotedblleft jointly symmetric\textquotedblright{ }copula is defined as follows
\begin{definition}
The $n$ dimensional copula $C^{JS}$, is jointly symmetric:
\[
C^{JS}\left(u_1,\cdots,u_n\right)=\frac{1}{2^n}\sum\limits_{k_1=0}^{2}\cdots\sum\limits_{k_n=0}^{2}\left(-1\right)^R C(\tilde{u}_1,\cdots,\tilde{u}_i,\cdots,\tilde{u}_n),
\] 
\[
\text{where}\quad R=\sum\limits_{i=1}^n\bm{1}\{k_i=2\},\quad\text{and}\quad\tilde{u}_i=
\begin{cases}
1,& k_i=0\\
u_i,&k_i=1\\
1-u_i,& k_i=2
\end{cases}.
\]
 \end{definition}
 The general idea is that the average of mirror image rotations of a possibly asymmetric copula along each axis generates a jointly symmetric copula [see \citet{oh2016high}]. For instance, the marginals can be assumed to possess standard normal distributions, while the nonlinear dependency is modeled using jointly symmetric copulas.  

Therefore, using the jointly symmetric copula family, the bivariate copulas can be evaluated by considering

\[
C^{JS}\left(\phi_t(s_t\mid X;\hat{\bm{\beta_1}}),\phi_{t+1}(s_{t+1}\mid X;\hat{\bm{\beta_1}})\right)=\frac{1}{4}\sum\limits_{k_1=0}^{2}\sum\limits_{k_2=0}^{2}\left(-1\right)^R C(\phi_t(s_t\mid X;\hat{\bm{\beta_1}}),\phi_{t+1}(s_{t+1}\mid X;\hat{\bm{\beta_1}}))
\]
\[
\text{where}\quad R=\sum\limits_{i=1}^n\bm{1}\{k_i=2\},\quad\text{and}\quad\tilde{u}_i=
\begin{cases}
1,& k_i=0\\
\phi_i(.),&k_i=1\\
1-\phi_i(.),& k_i=2
\end{cases},\quad \forall i\in\{t,t+1\}.
\]
for $t=1,\cdots,T-1$.

\subsection{Truncated D-vines \label{Truncated D-vines}}

{\hskip 1.5em}Following \citet{joe2014dependence}, we refer to a D-vine as a $p$-truncated D-vine, if the copulas in the trees $\mathcal{T}_{p+1},\cdots,\mathcal{T}_{T}$ are $C^{\perp}$, where by definition
\[
C^{\perp}\left(u_1,\cdots,u_n\right)=\prod\limits_{t=1}^{T}u_t,\quad\text{with}\quad (U_1,\cdots,U_T)\sim U(0,1), 
\]
implying $U_1\perp U_2 \perp \cdots \perp U_T$. The POS-based tests constructed using the Markov assumption of order one can be regarded as a special case of the PCC-POS based tests, whereby the former can be constructed by a $1$-truncated D-vine, which only depends on $C_{12},C_{23},\cdots,C_{T-1,T}$, or rather $C_{12},C_{\sigma_{13}3},\cdots,,C_{\sigma_{1T}T}$ using a vine array representation, given the following array
\begin{equation*}
\begin{bmatrix}
- & 12 & 23 &34 &\cdots & T-1,T\\
 &-&13\mid2^{\textcolor{red}{\perp}}&24\mid 3^{\textcolor{red}{\perp}}&\cdots&T-2,T\mid T-1^{\textcolor{red}{\perp}}\\
& & \ddots&\cdots&\cdots&\vdots\\
&&&-&1,T-1\mid \text{\b{S}}_{\backslash j-}^{\textcolor{red}{\perp}}&2,T \mid \text{\b{S}}_{\backslash j+}^{\textcolor{red}{\perp}}\\
&&&&-&1,T\mid \text{\b{S}}_{\backslash j}^{\textcolor{red}{\perp}}\\
&&&&&-
\end{bmatrix},\quad
A=
\begin{bmatrix}
1 & 1 & 2 &3 &\cdots & T-1\\
 &2&1^{\textcolor{red}{\perp}}&2^{\textcolor{red}{\perp}}&\cdots&T-2^{\textcolor{red}{\perp}}\\
& & \ddots&\cdots&\cdots&\vdots\\
&&&T-2&1^{\textcolor{red}{\perp}}&2^{\textcolor{red}{\perp}}\\
&&&&T-1&1^{\textcolor{red}{\perp}}\\
&&&&&T
\end{bmatrix}
\end{equation*}
Similarly, a $2$-truncated D-vine, depends on the copulas $C_{12},C_{23},\cdots,C_{T-1,T}$ and $C_{13\mid 2},C_{24\mid 3},\cdots,C_{T-2,T\mid T-1}$ or $C_{\sigma_{1t}}t$ for $t=2,\cdots,T$ and $C_{\sigma_{2t}t\mid\sigma_{1t}}$ for $t=3,\cdots,T$ using the vine array representation. Therefore, for a $p$-truncated D-vine, (\ref{eq: stober2}) is modified to 
\begingroup
\allowdisplaybreaks
\begin{align}\label{eq: stober3}
\begin{split}
P_{t\mid 1:{t-1}}[s(\tilde{y}_t)=\tilde{s}_t\mid s(\tilde{y}_1)=\tilde{s}_1,\cdots,s(\tilde{y}_{t-1})=\tilde{s}_{t-1},X]&=P_{t\mid 1:{p}}[s(\tilde{y}_t)=\tilde{s}_t\mid s(\tilde{y}_1)=\tilde{s}_1,\cdots,s(\tilde{y}_{p})=\tilde{s}_{p},X]\\
&=\left\{\prod\limits_{l=p\land (t-1)}^{2} c_{\tilde{\sigma}_{lt}t,\mid \tilde{\sigma}_{1t},\cdots,\tilde{\sigma}_{t-1,t}}\right\}\times c_{\tilde{\sigma}_{lt}t}\times P_t[s(\tilde{y}_t)=\tilde{s}_t\mid X],
\end{split}
\end{align}
\endgroup
for $t-1\geq p$.
\begin{corollary}\label{coroll3}
 Let $\tilde{A}=(\tilde{\sigma}_{lt})_{1\leq l\leq t\leq T}$ be a D-vine array for the signs $s(\tilde{y}_1),\cdots,s(\tilde{y}_T)$, where the signs $\left\{s(\tilde{y}_t)\right\}_{t=0}^{\infty}$ follow a Markov process of order $p$. Under assumptions (\ref{eq: modelnl}) and (\ref{eq: median}%
), let $H(\bm{\beta}_{0})$ and $H(\bm{\beta}_{1})$ be defined by (\ref{eq: nullnl}) - (\ref%
{eq: altnl})$,$ 
\begingroup
\allowdisplaybreaks
\begin{align*}
SN_{T}(\bm{\beta}_0\mid\bm{\beta}_{1})=\sum\limits_{t=2}^{T}\sum\limits_{l=p\land( t-1)}^{2}\ln c_{\tilde{\sigma}_{lt}t,\mid \tilde{\sigma}_{1t},\cdots,\tilde{\sigma}_{t-1,t}}+\sum\limits_{t=2}^{T}\ln c_{\tilde{\sigma}_{1t}t}+\sum\limits_{t=1}^{T}s(y_t-f(\bm{x}_{t-1},\bm{\beta}_0))\tilde{a}_t(\bm{\beta}_0\mid\bm{\beta}_1)>c_1(\bm{\beta}_0,\bm{\beta}_1),
\end{align*}%
\endgroup
where 
\[
\tilde{a}_t(\bm{\beta}_0\mid\bm{\beta}_1)=\ln\left\{\frac{1-p_t[\bm{x}_{t-1},\bm{\beta}_0,\bm{\beta}_1\mid X]}{p_t[\bm{x}_{t-1},\bm{\beta}_0,\bm{\beta}_1\mid X]}\right\},\quad p_t[\bm{x}_{t-1},\bm{\beta}_0,\bm{\beta}_1\mid X]=P_t[\varepsilon_t\leq f(\bm{x}_{t-1},\bm{\beta}_0)-f(\bm{x}_{t-1},\bm{\beta}_1)\mid X]
\]
and suppose the constant $c_{1}(\bm{\beta}_0,\bm{\beta}_{1})$ satisfies the constraint $P\left[SN_T(\bm{\beta}_0\mid\bm{\beta}_1)>c_{1}(\bm{\beta}_0,\bm{\beta}_1)\right] =\alpha $
under $H(\bm{\beta}_{0}),$ with $0<\alpha <1.$ Then the test that rejects $H(\bm{\beta}_{0})$ when 
\begin{equation}
SN_{T}(\bm{\beta}_0\mid\bm{\beta}_{1})>c_{1}(\bm{\beta}_0,\bm{\beta}_1) \label{eq: SLn critical region}
\end{equation}%
is most powerful for testing $H(\bm{\beta}_{0})$ against $H(\bm{\beta}_{1})$ among level-$\alpha $
tests based on the signs $\big(s(\tilde{y}_{1}),\cdots,s(\tilde{y}_{T})\big)%
'.$
\end{corollary}

\section{Choice of the optimal alternative hypothesis \label{optimal
alternative hypothesis}}

In this Section, we follow \citet{dufour2010exact} by first showing the analytical derivation of the power envelope function of the PCC-POS-based tests. We then suggest using simulations as means of approximating the said function, by showing the difficulty of inverting the latter to find the optimal alternative. Thereafter, we propose an adaptive approach based on the split-sample technique to choose an alternative which has a power function close to that of the power envelope.

\subsection{Power envelope of PCC-POS tests \label{Power envelope of PCC-POS tests}}
{\hskip 1.5em}Point-optimal tests trace out the power envelope (i.e. the maximum attainable power) for any given testing problem [see \citet{king1987towards}]. However, in practice the alternative hypothesis $\bm{\beta}_1$ is unknown and a problem consists of finding an approximation for it, such that the power function is maximized and is close to that of the power envelope. Following \citet{dufour2010exact} and \citet{dufour2001finite}, we propose an adaptive approach based on the split-sample technique to choose an alternative $\bm{\beta}_1$ that yields the greatest power function and makes size control easier [see \citet{dufour2001finite} and \citet{dufour2010exact} for an overview]. We follow \citet{dufour2010exact} by presenting the analytical derivation of the power envelope of the PCC-POS tests for predictive regressions, which can be purposed as a benchmark for comparing the power functions of the PCC-POS tests for different sample splits. 

We have shown in Section \ref{nonlinear} that the PCC-POS tests are a function of $\bm{\beta}_1$. In other words,
\[
SN_{T}(\bm{\beta}_0\mid\bm{\beta}_{1})=\sum\limits_{t=2}^{T}\sum\limits_{l=t-1}^{2}\ln c_{\tilde{\sigma}_{lt}t\mid \tilde{\sigma}_{1t},\cdots,\tilde{\sigma}_{t-1,t}}+\sum\limits_{t=2}^{T}\ln c_{\tilde{\sigma}_{1t}t}+\sum\limits_{t=1}^{T}\ln\left\{\frac{1-p_t[\bm{x}_{t-1},\bm{\beta}_0,\bm{\beta}_1\mid X]}{p_t[\bm{x}_{t-1},\bm{\beta}_0,\bm{\beta}_1\mid X]}\right\}s(y_t-f(\bm{x}_{t-1},\bm{\beta}_0)).
\]
which in turn implies that its power function, say $\Pi(\bm{\beta}_0,\bm{\beta}_1)$, is also a function of $\bm{\beta}_1$
\[
\Pi(\bm{\beta}_0,\bm{\beta}_1)=P[SN_{T}(\bm{\beta}_0\mid\bm{\beta} _{1})>c_1(\bm{\beta}_0,\bm{\beta}_1)\mid H(\bm{\beta}_1)]
\]
where $c_1(\bm{\beta}_0,\bm{\beta}_1)$ is the smallest constant that satisfies $P[SN_T(\bm{\beta}_0\mid\bm{\beta}_1)>c_1(\bm{\beta}_0,\bm{\beta}_1)\mid H(\bm{\beta}_0)]\leq \alpha$, and where $\alpha$ is an arbitrary significance level.  Theorem \ref{Theorem1} provides the theoretical results for the power function of the PCC-POS tests.
\begin{theorem}\label{Theorem1} 
Under assumption (\ref{eq: median}) and and given model (\ref{eq: modelnl}), and further under the condition that $s(\tilde{y}_1),\cdots,s(\tilde{y}_T)$ conditional on $X$ follow a Regularity Markov Type process (RMT hereafter), the power function of $SN_{T}(\bm{\beta}_0\mid\bm{\beta}_{1})$ is given by
\[
\Pi(\bm{\beta}_0,\bm{\beta}_1)=P[SN_T(\bm{\beta}_0\mid \bm{\beta}_1)> c_1(\bm{\beta}_0,\bm{\beta}_1)\mid X]=\frac{1}{2}+\frac{1}{\pi}\int_{0}^{\infty}\frac{\Im\{\exp(iuc_1(\bm{\beta}_0,\bm{\beta}_1))\phi_{SN_T}(u)\}}{u}du
\] 
$\forall u \in \mathbb{R}$, $i=\sqrt{-1}$, and with $\Im{z}$ denoting the imaginary part of the complex number $z$. $\phi_{SN_T}(u)$ is given by
\[
\phi_{SN_T}(u)=\prod\limits_{t=1}^{T}\left(\E_X\left[\exp\left(iu\left\{R_{t,t-1}+\ln\left\{\frac{1-p_t[\bm{x}_{t-1},\bm{\beta}_0,\bm{\beta}_1\mid X]}{p_t[\bm{x}_{t-1},\bm{\beta}_0,\bm{\beta}_1\mid X]}\right\}s(\tilde{y}_t)\right\}\right)\right]+\rho_t(u)\right),
\]
where $R_{1,0}=0$, $R_{t,t-1}=\sum\limits_{l=t-1}^{2}\ln c_{\tilde{\sigma}_{lt}t\mid \tilde{\sigma}_{1t},\cdots,\tilde{\sigma}_{t-1,t}}+\ln c_{\tilde{\sigma}_{1t}t}$ for $t=2,\cdots,T$, such that for D-vine-array $\tilde{A}=(\tilde{\sigma}_{lt})_{1\leq l\leq t\leq T}$, $l=2,\cdots,T-1$ is the row with tree $\mathcal{T}_l$, and column $t$ has the permutation $\tilde{\underline{\sigma}}_{t-1}=(\tilde{\sigma}_{1t},\cdots,\tilde{\sigma}_{t-1,t})$ of the previously added variables, $p_t[\bm{x}_{t-1},\bm{\beta}_0,\bm{\beta}_1\mid X]=P_t[\varepsilon_t\leq f(\bm{x}_{t-1},\bm{\beta}_0)-f(\bm{x}_{t-1},\bm{\beta}_1)\mid X]$, $\tilde{\text{\b{S}}}_{t-1}=s(y_{t-1}-f(\bm{x}_{t-2},\bm{\beta}_0)),\cdots,s(y_{1}-f(\bm{x}_{0},\bm{\beta}_0))$.
Finally, $c_1(\bm{\beta}_0,\bm{\beta}_1)$ is the smallest constant that satisfies $P[SN_T(\bm{\beta}_0\mid\bm{\beta}_1)>c_1(\bm{\beta}_0,\bm{\beta}_1)\mid H(\bm{\beta}_0)]\leq \alpha$, where $\alpha$ is an arbitrary significance level. 
\end{theorem}

Under the assumption that the signs follow an RMT-process, $\rho_t(u)$ can be estimated using the results from Theorem 2 of \citet{heinrich1982factorization}.  Given that point-optimal tests are optimal at a specific point in the alternative parameter space, the power envelope of the PCC-POS tests, say $\bar{\Pi}(\bm{\beta}_1)$, is obtained for values of $\bm{\beta}$, such that $\{\bm{\beta}: \bm{\beta}=\bm{\beta}_1, \forall \bm{\beta}_1\in \mathbb{R}^{(k+1)}\}$. Finding values of $\bm{\beta}_1$ for a PCC-POS test at level $\alpha$, with a power function that is close to the power envelope can be achieved by inverting the power envelope function. However, in a much simpler case of POS tests for i.n.i.d data, \citet{dufour2010exact} show that the inversion of the power function is not a straightforward task and obtaining an exact solution is not feasible. Therefore, simulations are used as means of approximating the power envelope function and finding the optimal alternative for the PCC-POS test.
\subsection{Split-sample technique for choosing the optimal alternative \label{Choice of the optimal alternative hypothesis}}
{\hskip 1.5em}As noted earlier, the power function of the PCC-POS test statistic depends on the alternative $\bm{\beta}_1$, which in practice is unknown and needs to be approximated. To make size control easier and to choose an approximation to $\bm{\beta}_1$ such that the power function of the test statistic is close to that of the power envelope, we follow \citet{dufour2010exact} by proposing an adaptive approach based on the split-sample technique for choosing the alternative. For an extensive review of adaptive statistical methods, we refer the reader to \citet{o2004applied}. Furthermore, the application of the split-sample technique in parametric settings can be studied by consulting \citet{dufour2003point} and \citet{dufour2008finite}.  
\begin{figure}[tbph]
\caption{Power comparisons: different split-samples. Normal error distributions with
different values of $\rho $ in (\ref{eq: errorsim}) and $\theta =0.9$ in (\ref{eq: theta})}
\begin{center}
\subfigure{\includegraphics[scale=0.58]{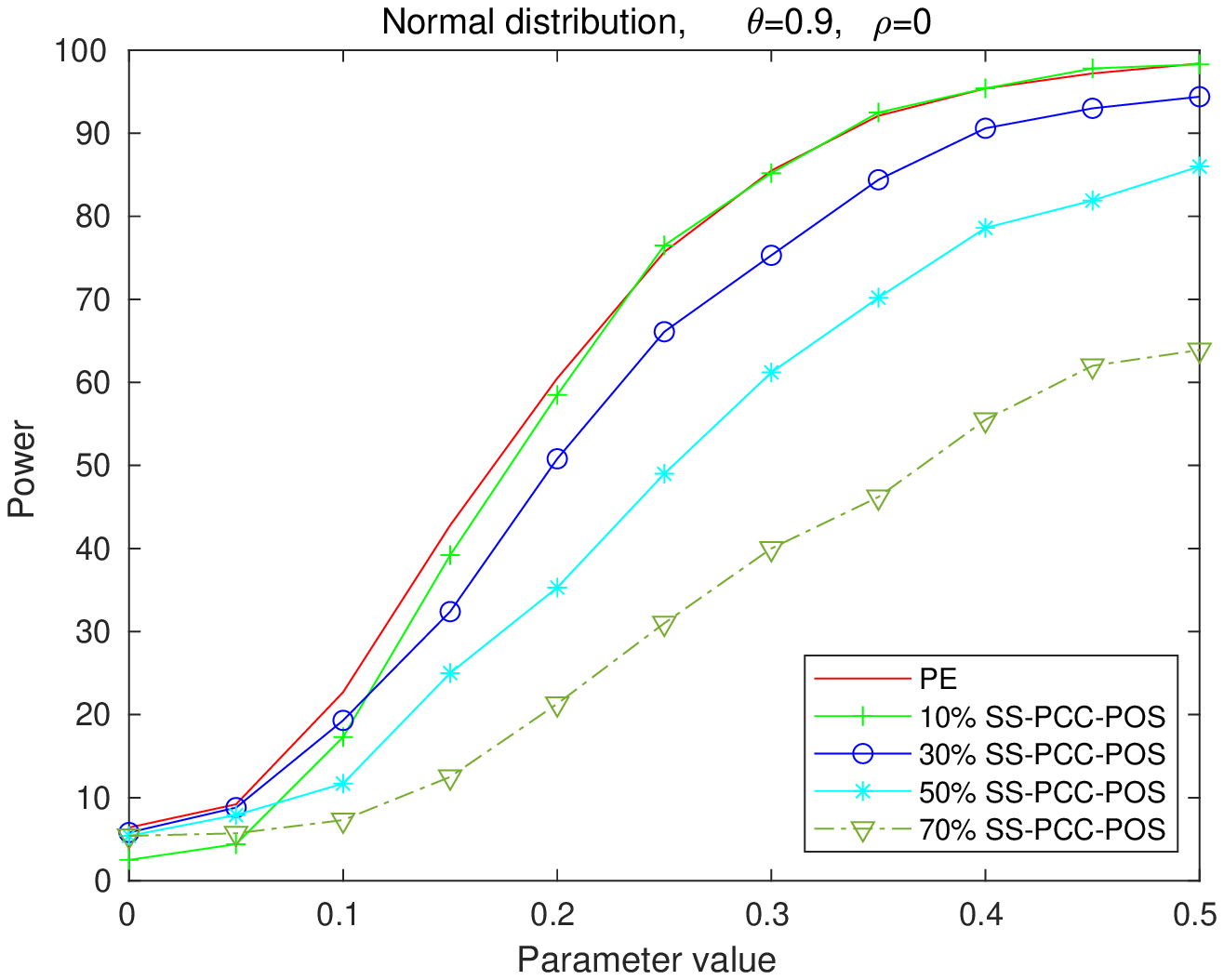}} %
\subfigure{\includegraphics[scale=0.58]{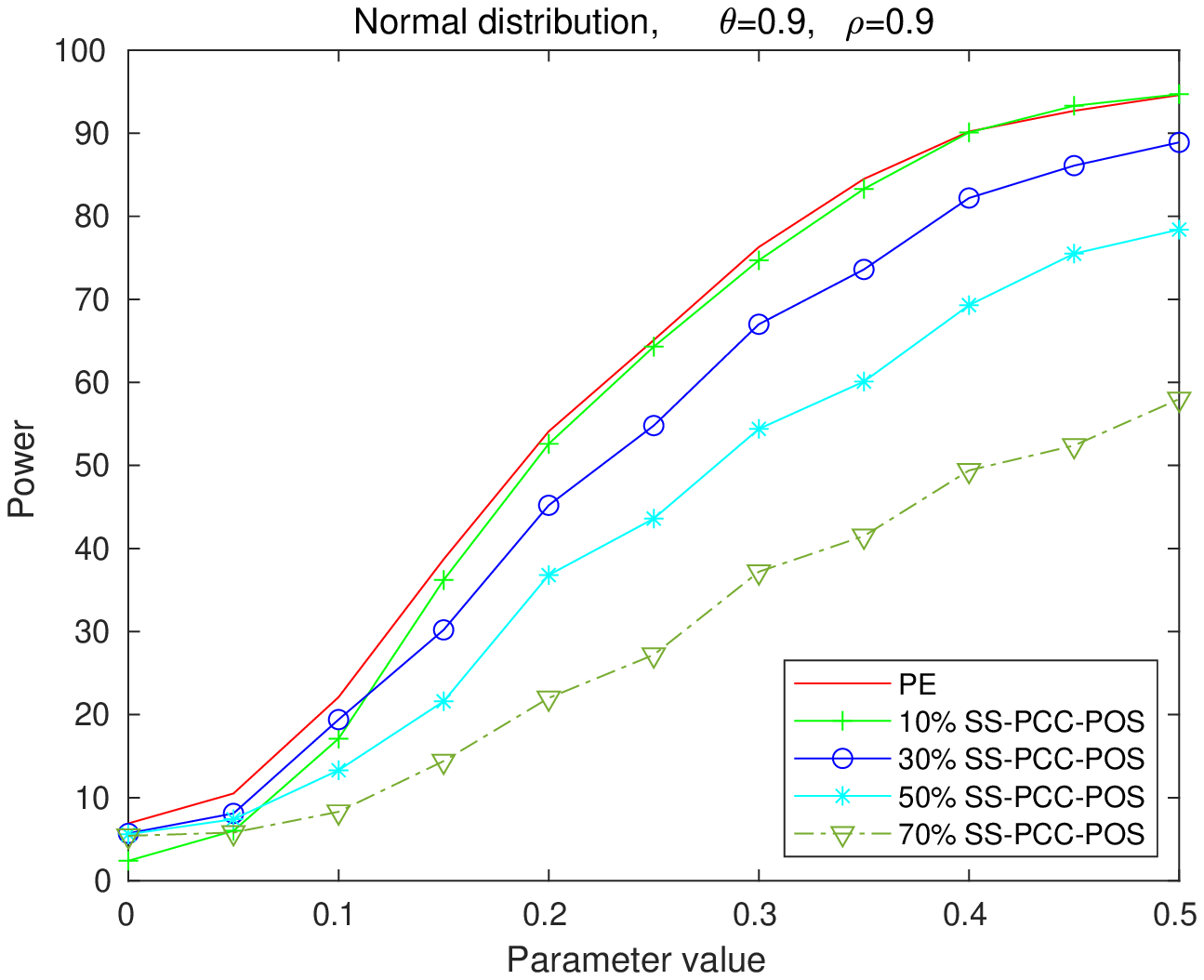}} \\[0pt]
\vspace{1pt}
\end{center}
\doublespacing
Note: These figures compare the power
envelope the PCC-POS test statistic using different split-samples: 10\%, 30\%, 50\%, 70\%. \textquotedblleft PE\textquotedblright refers to the power envelope of the PCC-POS test.
\label{fig: SS17}
\end{figure}

The split-sample technique involves splitting a sample of size $T$ into two independent subsamples, say $T_1$ and $T_2$, such that $T=T_1+T_2$. The first subsample is then used to estimate the alternative $\bm{\beta}_1$, while the other is purposed for computing the PCC-POS test statistic. Assuming that $f(\bm{x}_{t-1},\bm{\beta})=\bm{x}_{t-1}'\bm{\beta}$, the alternative $\bm{\beta}_1$ can be estimated using OLS
\[
\hat{\bm{\beta}}_{(1)}=(X_{(1)}'X_{(1)})^{-1}X_{(1)}'y_{(1)}.
\]
We provide a caveat that the OLS estimator is sensitive to extreme outliers, which motivates the use of robust estimators [see. \citet{maronna2019robust} for a review of robust estimators]. Using $\hat{\bm{\beta}}_{(1)}$ and the observations in the second independent subsample, we compute the test-statistic as follows
\begingroup
\allowdisplaybreaks
\begin{align*}
SN_{T}(\bm{\beta}_0\mid\bm{\beta}_{(1)})&=\sum\limits_{t=T_1+2}^{T}\sum\limits_{l=t-1}^{2}\ln c_{\tilde{\sigma}_{lt}t\mid \tilde{\sigma}_{(T_1+1)t},\cdots,\tilde{\sigma}_{t-1,t}}+\sum\limits_{t=T_1+2}^{T}\ln c_{\tilde{\sigma}_{(T_1+1)t}t}+\\
&\textcolor{white}{=}\sum\limits_{t=T_1+1}^{T}\ln\left\{\frac{1-p_t[\bm{x}_{t-1},\bm{\beta}_0,\bm{\beta}_{(1)}\mid X]}{p_t[\bm{x}_{t-1},\bm{\beta}_0,\bm{\beta}_{(1)}\mid X]}\right\}s(y_t-\bm{x}_{t-1}'\bm{\beta}_0).
\end{align*}
\endgroup
where for $t=T_1+2,\cdots,T$ and D-vine-array $\tilde{A}=(\tilde{\sigma}_{lt})_{1\leq l\leq t\leq n}$, $l=T_1+1,\cdots,T-1$ is the row with tree $\mathcal{T}_l$, and column $t$ has the permutation $\tilde{\underline{\sigma}}_{t-1}=(\tilde{\sigma}_{(T_1+1)t},\cdots,\tilde{\sigma}_{t-1,t})$ of the previously added variables, $p_t[\bm{x}_{t-1},\bm{\beta}_0,\bm{\beta}_{(1)}\mid X]=P_t[\varepsilon_t\leq \bm{x}_{t-1}'(\bm{\beta}_0-\bm{\beta}_1)\mid X]$, and $\tilde{\text{\b{S}}}_{t-1}=s(y_{t-1}-\bm{x}_{t-2}'\bm{\beta}_0),\cdots,s(y_{T_1+2}-\bm{x}_{T_1+1}'\bm{\beta}_0)$.
\begin{figure}[tbph]
\caption{Power comparisons: different split-samples. Cauchy error distributions with
different values of $\rho $ in (\ref{eq: errorsim}) and $\theta =0.9$ in (\ref{eq: theta})}
\begin{center}
\subfigure{\includegraphics[scale=0.57]{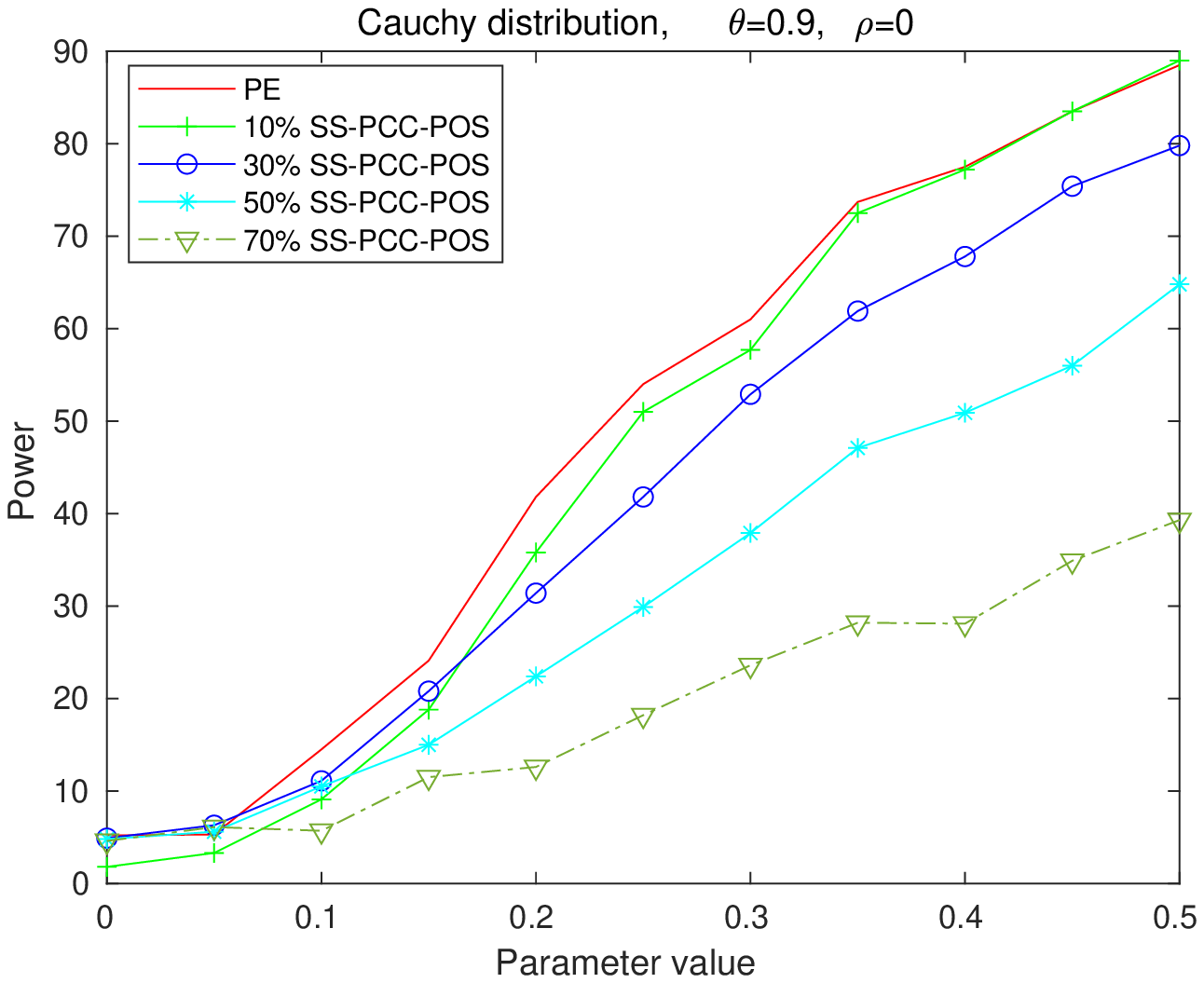}} %
\subfigure{\includegraphics[scale=0.57]{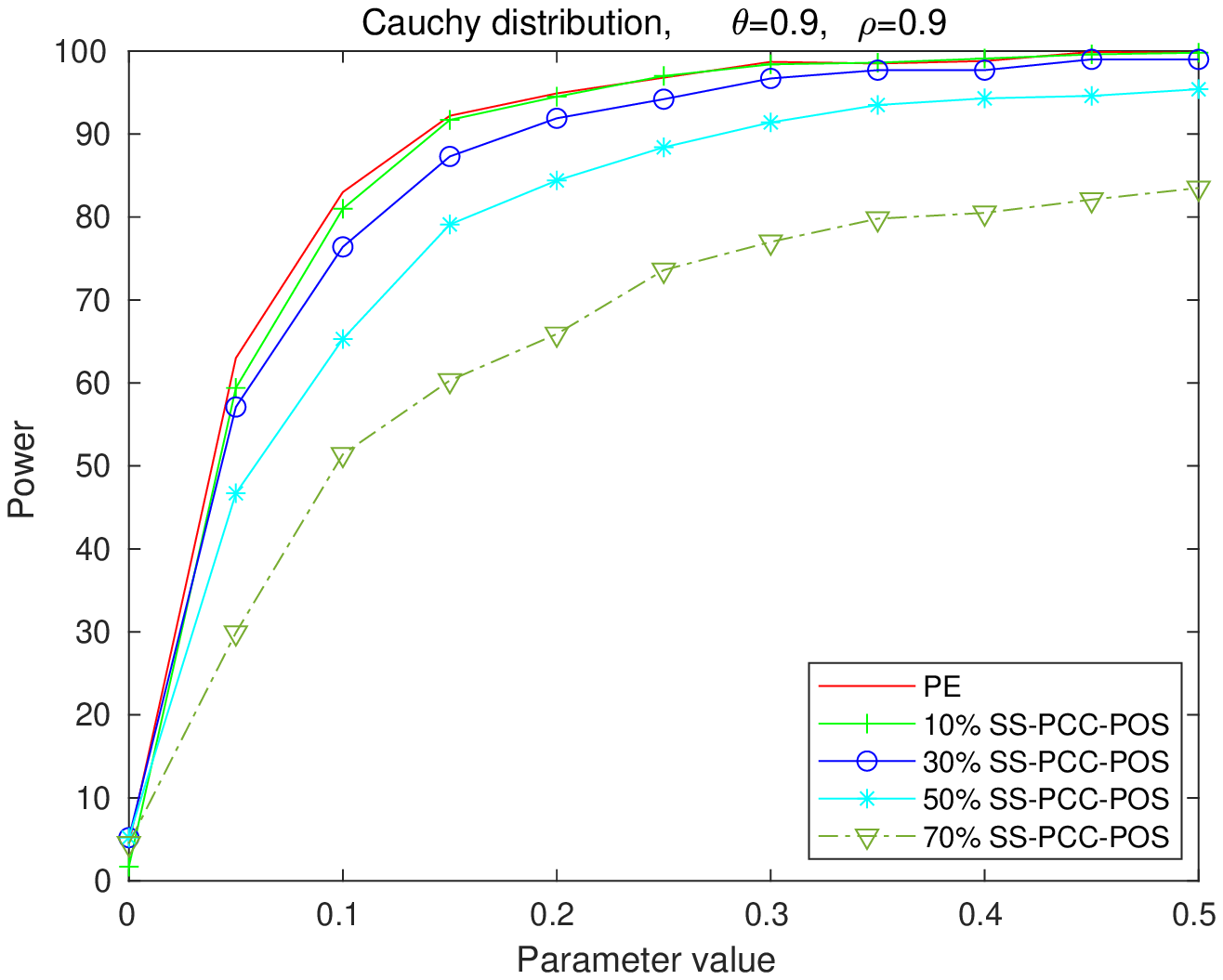}} \\[0pt]
\end{center}
\doublespacing
Note: These figures compare the power
envelope the PCC-POS test statistic using different split-samples: 10\%, 30\%, 50\%, 70\%. \textquotedblleft PE\textquotedblright refers to the power envelope of the PCC-POS test.
\label{fig: SS28}
\end{figure}
The choices for the subsamples $T_1$ and $T_2$ can be arbitrary. However, our simulations show that the proportion of the observations retained for estimating the alternative and in turn for computing the PCC-POS test statistic has an impact on the power of the test. We find that the power function of the split-sample PCC-POS test (SS-PCC-POS test hereafter) is closest to that of the power envelope, when a relatively small number of observations is retained for estimating the alternative, with the rest used for computing the test statistic - findings that are in line with \citet{dufour2010exact}. Specifically, by considering all the DGPs in our simulations study, we find that the subsamples $T_1$ and $T_2$ must in turn contain roughly $10\%$ and $90\%$ of the observations in the entire sample respectively.  
\begin{figure}[tbph]
\caption{Power comparisons: different split-samples. Student's $t$ error distributions with 2 degrees of freedom [i.e $t(2)$] with
different values of $\rho $ in (\ref{eq: errorsim}) and $\theta =0.9$ in (\ref{eq: theta})}
\begin{center}
\subfigure{\includegraphics[scale=0.58]{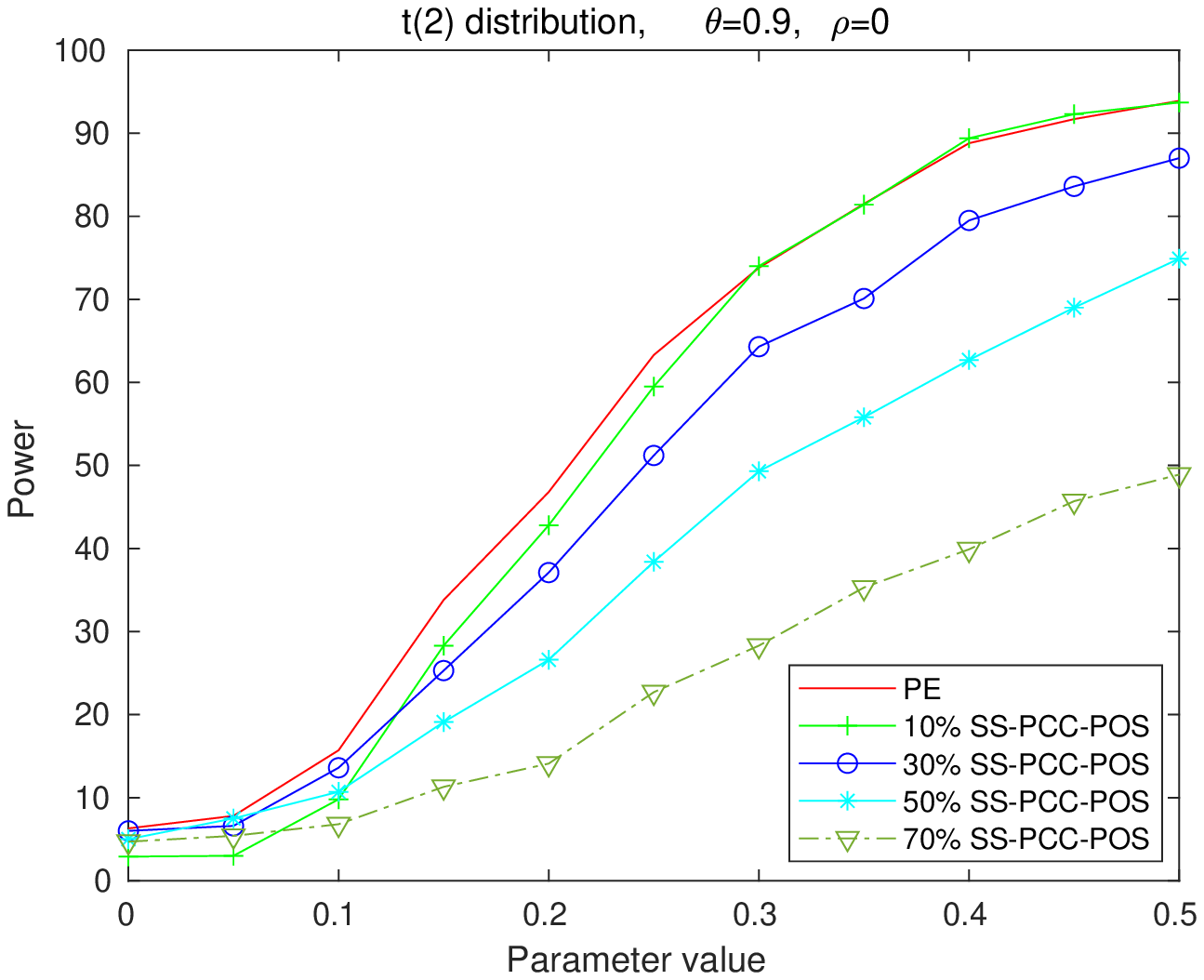}} %
\subfigure{\includegraphics[scale=0.58]{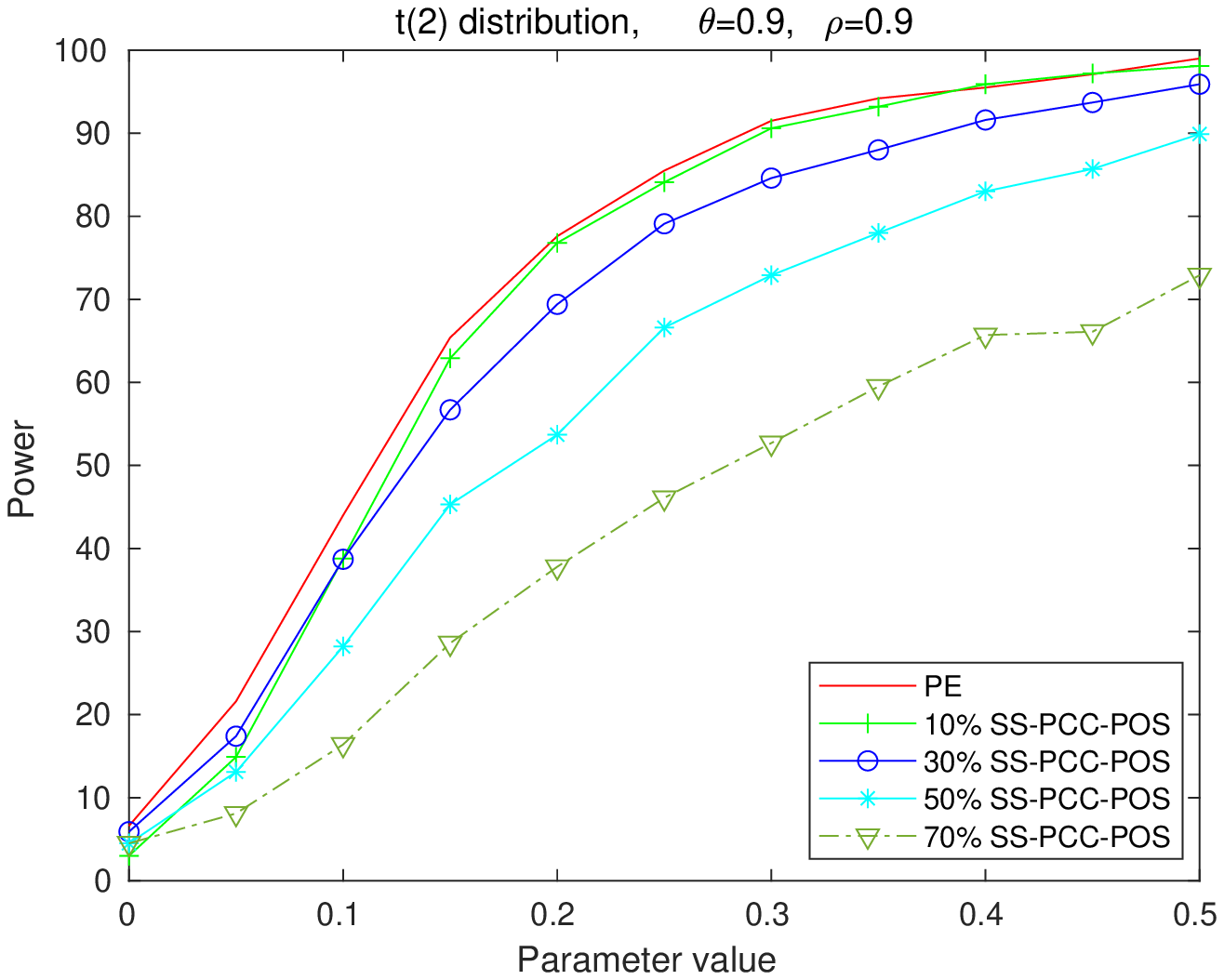}} \\[0pt]
\end{center}
\doublespacing
Note: These figures compare the power
envelope the PCC-POS test statistic using different split-samples: 10\%, 30\%, 50\%, 70\%. \textquotedblleft PE\textquotedblright refers to the power envelope of the PCC-POS test.
\label{fig: Power comparison using different tests Normal}
\end{figure}

\begin{figure}[tbph]
\caption{Power comparisons: different split-samples. Normal error distributions with break in variance, with
different values of $\rho $ in (\ref{eq: errorsim}) and $\theta =0.9$ in (\ref{eq: theta})}
\begin{center}
\subfigure{\includegraphics[scale=0.58]{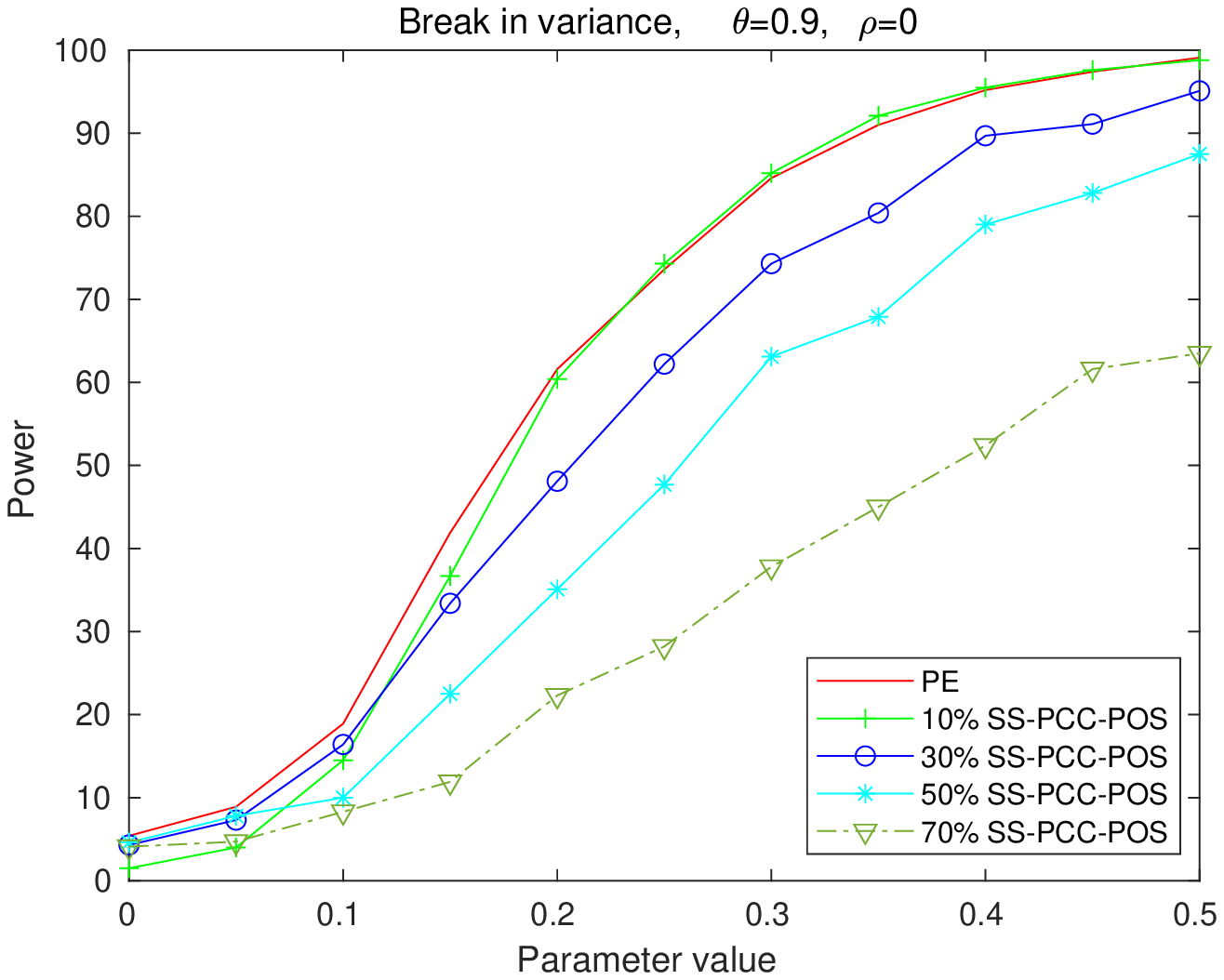}} %
\subfigure{\includegraphics[scale=0.58]{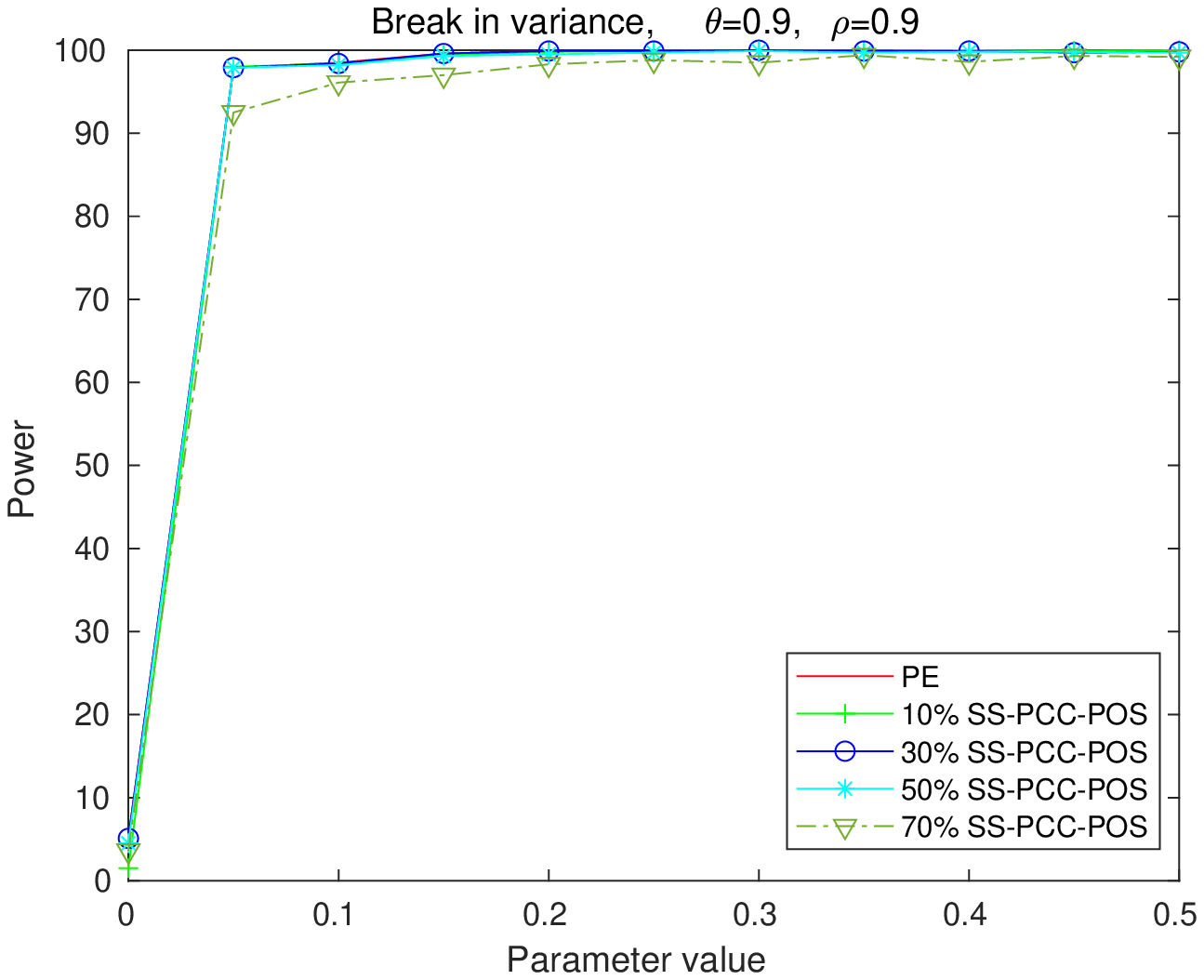}} \\[0pt]
\end{center}
\doublespacing
Note: These figures compare the power
envelope the PCC-POS test statistic using different split-samples: 10\%, 30\%, 50\%, 70\%. \textquotedblleft PE\textquotedblright refers to the power envelope of the PCC-POS test.
\label{fig: SS410}
\end{figure}
\section{PCC-POS confidence regions \label{projectiontechniqueC3}}

{\hskip 1.5em}In this Section, we lay out the theoretical framework for building confidence regions for a vector
(sub-vector) of the unknown parameters $\bm{\beta}$, say $C_{\bm{\beta}}(\bm{\alpha})$, at a given
significance level $\alpha $, using the proposed PCC-POS tests. Consider model (\ref{eq: modelnl}) such that $f(\bm{x}_{t-1},\bm{\beta})=\bm{x}_{t-1}'\bm{\beta}$.
Suppose we wish to test the null hypothesis (\ref{eq: nullnl})
against the alternative hypothesis (\ref{eq: altnl}). Formally, this implies finding all the values of $\bm{\beta}_{0}\in \mathbb{R}^{k}$ such that
\begin{equation}
SN_{T}(\bm{\beta}_0\mid\bm{\beta}_{1})=\sum\limits_{t=2}^{T}\sum\limits_{l=t-1}^{2}\ln c_{\tilde{\sigma}_{lt}t,\mid \tilde{\sigma}_{1t},\cdots,\tilde{\sigma}_{t-1,t}}+\sum\limits_{t=2}^{T}\ln c_{\tilde{\sigma}_{1t}t}+\sum\limits_{t=1}^{T}s(y_t-\bm{\beta}_0'-\bm{x}_{t-1})\tilde{a}_t(\bm{\beta}_0\mid\bm{\beta}_1)<c_1(\bm{\beta}_0,\bm{\beta}_1).
\end{equation}%
where the critical value is given by the smallest constant $c(\bm{\beta}
_{0},\bm{\beta}_{1})$ such that
\begin{equation*}
P\left[SN_{T}(\left. \bm{\beta}_{0}\right\vert \bm{\beta}
_{1})>c(\bm{\beta}
_{0},\bm{\beta}_{1})\mid\bm{\beta} =\bm{\beta}_{0}\right]\leq\alpha 
\end{equation*}%
Thus, the confidence region of the vector of parameters $\bm{\beta}$ can be
defined as follows:
\begin{equation*}
C_{\bm{\beta}}(\alpha )=\left\{ \bm{\beta} _{0}:SN_{T}(\left. \bm{\beta}_{0}\right\vert \bm{\beta}
_{1})<c(\bm{\beta} _{0},\bm{\beta}_{1})|P[SN_{T}(\left. \bm{\beta}_{0}\right\vert \bm{\beta}
_{1})>c(\bm{\beta}_{0},\bm{\beta}
_{1})\mid\bm{\beta} =\bm{\beta}_{0}]\leq \alpha \right\} .
\end{equation*}%
Given the confidence region $C_{\bm{\beta} }(\alpha )$, confidence intervals for the components and sub-vectors of vector $\bm{\beta}$ can be derived using the projection techniques [see \citet{dufour2010exact} and \citet{coudin2009finite}]. Confidence sets in the form of transformations $T$ of $\bm{\beta}\in%
\mathbb{R}^{m}$, for $m\leq (k+1)$, say $T(C_{\bm{\beta}}(\alpha ))$, can easily be found using these techniques. Since, for any set $C_{\bm{\beta} }(\alpha )$
\begin{equation}
\bm{\beta} \in C_{\bm{\beta} }(\alpha )\implies T(\bm{\beta} )\in T(C_{\bm{\beta} }(\alpha )),
\label{cr1}
\end{equation}%
we have
\begin{equation}
P[\bm{\beta} \in C_{\bm{\beta} }(\alpha )]\geq 1-\alpha \implies P%
[T(\bm{\beta} )\in T(C_{\bm{\beta}}(\alpha ))]\geq 1-\alpha   \label{cr2}
\end{equation}%
where
\begin{equation*}
T(C_{\bm{\beta}}(\alpha ))=\{\delta \in \mathbb{R}^{m}:\exists \bm{\beta} \in
C_{\bm{\beta} }(\alpha ),T(\bm{\beta} )=\delta \}.
\end{equation*}%
\newline
From (\ref{cr1}) and (\ref{cr2}) , the set $T(C_{\bm{\beta} }(\alpha ))$ is a
conservative confidence set for $T(\bm{\beta} )$ with level $1-\alpha $. If $%
T(\bm{\beta} )$ is a scalar, then we have
\begin{equation*}
P[\inf \{T(\bm{\beta} _{0}),\quad\text{for}\quad\bm{\beta}_{0}\in C_{\bm{\beta}
}(\alpha )\}\leq T(\bm{\beta} )\leq \sup \{T(\bm{\beta} _{0}),\quad\text{for}\quad\bm{\beta}
_{0}\in C_{\bm{\beta}}(\alpha )\}]>1-\alpha .
\end{equation*}

\section{Monte Carlo study \label{sec: Monte Carlo study}}
{\hskip 1.5em}In this Section, we assess the performance of the proposed 10\% SS-PCC-POS tests (in terms of size control and power) by comparing it to other tests that are intended to be robust against non-standard distributions and heteroskedasticity of unknown form. We consider DGPs under different distributional assumptions and heteroskedasticities. For each DGP, we further consider different correlation coefficients between the errors of the predictive regression and the disturbances of the regressors. In the first subsection, the DGPs are formally introduced, and in the following subsection, the performance of the proposed 10\% SS-PCC-POS tests are compared to that of the other tests considered in our study.  
\subsection{Simulation setup \label{Simulation setup}}
{\hskip 1.5em}We assess the performance of the proposed 10\% SS-PCC-POS tests in terms of size and power, by considering various DGPs with different symmetric and asymmetric distributions and forms of heteroskedasticity. The DGPs under consideration are supposed to mimic different scenarios that are often encountered in practical settings. The performance of the 10\% SS-PCC-POS tests is compared to that of a few other tests, by considering the following linear predictive regression model
\begin{equation}
 y_t=\beta x_{t-1}+\varepsilon_t
\end{equation}
where $\beta$ is an unknown parameter and
\begin{equation}\label{eq: theta}
x_{t}=\theta x_{t-1}+u_t
\end{equation}
where $\theta=0.9$ and
\begin{equation}\label{eq: errorsim}
u_t=\rho \varepsilon_t+w_t\sqrt{1-\rho^2}
\end{equation}
for $\rho=0,0.1,0.5,0.9$ and $\varepsilon_t$ and $w_t$ are assumed to be independent. The initial value of $x$ is given by: $x_0=\frac{w_0}{\sqrt{1-\theta^2}}$ and $w_t$ are generated from $N(0,1)$. 
The errors $\varepsilon_t$ are i.n.i.d and are categorized by two groups in our simulation study. In the first group, we consider DGPs where the errors $\varepsilon_t$ possess symmetric and asymmetric distributions:
\begin{enumerate}
\item[\textbf{1.}] normal distribution: $\varepsilon_t\sim N(0,1)$;
\item[\textbf{2.}] Cauchy distribution: $\varepsilon_t\sim Cauchy$;
\item[\textbf{3.}] Student's $t$ distribution with two degrees of freedom: $\varepsilon_t\sim t(2)$;
\item[\textbf{4.}] Mixture of Cauchy and normal distributions: $\varepsilon_t \sim \mid\varepsilon_t^C\mid-(1-s_t)\mid\varepsilon_t^N\mid$, where $\varepsilon_t^C$ follows Cauchy distribution, $\varepsilon_t^N$ follows $N(0,1)$ distribution, and 
\[
P(s_t=1)=P(s_t=0)=\frac{1}{2}
\] 
\end{enumerate} 
The second group of DGPs represents different forms of heteroskedasticity:
\begin{enumerate}
\item[\textbf{5.}] break in variance: 
\begin{equation*}
\varepsilon _{t}\sim \left\{ 
\begin{array}{cc}
N(0,1) & \text{for}\ t\neq 25 \\ 
\sqrt{1000}N(0,1) & \text{for}\ t=25%
\end{array}%
\,;\right.
\end{equation*}%
\item[\textbf{6.}] GARCH$(1,\,1)$ plus jump variance:%
\begin{equation*}
\sigma _{\varepsilon }^{2}(t)=0.00037+0.0888\varepsilon
_{t-1}^{2}+0.9024\sigma _{\varepsilon }^{2}(t-1)\,,
\end{equation*}%
\begin{equation*}
\varepsilon _{t}\sim \left\{ 
\begin{array}{cc}
N(0,\sigma _{\varepsilon }^{2}(t)) & \text{for}\ t\neq 25 \\ 
50N(0,\sigma _{\varepsilon }^{2}(t)) & \text{for}\ t=25%
\end{array}%
\right. \,;
\end{equation*}%
\end{enumerate}
We consider the problem of testing the null hypothesis $H_0: \beta=0$. Our Monte Carlo simulations compare the size and power of the 10\%-PCC-POS test to those of \textit{t}-test, \textit{t}-test based on \citet{white1980heteroskedasticity} variance-correction (WT-test hereafter), and the sign-based test proposed by \citet{dufour1995exact}. Due to computational constraints, we perform only $M_1=999$ simulations to evaluate the probability distribution of the 10\% SS-PCC-POS test statistic and $M_2=1,000$ iterations for approximating the power functions of the proposed PCC-POS test and other tests. In all simulations, we consider a sample size of $T=50$. As the sign-based statistic of \citet{dufour1995exact} has a discrete distribution, it is not possible to obtain test with a precise size of $5\%$; therefore, the size of the test is $5.95\%$ for $T=50$. 
\subsection{Simulation results \label{Simulation results}}

{\hskip 1.5em}The results of the Monte Carlo study corresponding to DGPs described in Section \ref{Simulation setup} are presented in figures \ref{fig: Sim17}-\ref{fig: Sim612}. These figures compare the performance of the 10\% SS-PCC-POS test in terms of size and power, to those of the \textit{t}-test, \textit{t}-test based on White's (1980) variance-correction, as well as the sign-based procedure proposed by \citet{dufour2010exact}. The results are described in detail below.

First, figure \ref{fig: Sim17} considers the case where the error terms $\varepsilon_t$ are normally distributed. At first glance, we note that all tests control size. Evidently, our test is outperformed by the \textit{t}-test, as well as the \textit{t}-test based on White's (1980) variance-correction. The former is expected, since for normally distributed error terms, the \textit{t}-test is the most powerful test. However, the 10\% SS-PCC-POS test outperforms the sign-based procedure proposed by \citet{dufour1995exact} [CD (1995) hereafter]. Furthermore, changing the correlation coefficient $\rho$ does not seem to lead to visually significant differences in the performance of the tests.
 
Second, figure \ref{fig: Sim28} presents the results of the performance of the aforementioned tests, when the errors $\varepsilon_t$ follows Cauchy distribution. It is evident that the 10\% SS-PCC-POS test outperforms all other tests. Moreover, the \textit{t}-test and WT-test do not possess much power for low correlation coefficient (0 and 0.1) values, $\rho$. However, as the correlation between $u_t$ and $w_t$ increases, the gap between the power functions narrows significantly.  

Third, in figures \ref{fig: Sim39} and \ref{fig: Sim410}, we have considered the cases where the errors in turn follow $t(2)$ and mixture distributions. In the former case of $t(2)$ distributed errors, the 10\% SS-PCC-POS test outperforms the rest; however, for almost all correlation coefficients $\rho$, the gap between the power functions is rather small, albeit it is narrowest for $\rho=0.9$. In the case of errors with mixture distribution, our 10\% SS-PCC-POS test is still the most powerful test. On other hand, it is evident that the \textit{t}-test and WT-test do not possess much power for small values (0 and 0.1) of correlation coefficient $\rho$. However, the power functions increase and converge to those of the other tests, as the correlation increases.

An interesting observation is the stark contrast between the power of the 10\% SS-PCC-POS test and the \textit{t}-test, when the errors follow the Cauchy, $t(2)$ and normal distributions respectively. The Cauchy and $t(2)$ distributions possess heavy tails, in the presence of which the standard error of the regression coefficients is inflated, which in turn leads to low power when the mean is used as a measure of central tendency. For instance, the Cauchy distribution has the heaviest tails among the considered  DGPs, as a result of which the \textit{t}-test and WT-test have very low power. By noting that a Student's $t$ distribution with $\nu$ degrees of freedom converges to the Cauchy distribution for $\nu=1$ and to the normal distribution as $\nu\rightarrow \infty$, it would be interesting to see at which degree of freedom the 10\% SS-PCC-POS test is outperformed by the \textit{t}-test and WT-tests.  Figures \ref{fig: c21}-\ref{fig: c24} suggest that for different values of $\rho$ in (\ref{eq: errorsim}) the \textit{t}-test and WT test outperform the 10\% SS-PCC-POS test for $\nu=4$. Interestingly, figure \ref{fig: c25} shows that the tails of the $t(2)$ distribution are substantially heavier than that of the $t(4)$, which may explain the transition.

Finally, in figures \ref{fig: Sim39} and \ref{fig: Sim410}, the errors are normally distributed with different forms of heteroskedasticity. In the first case [see figure \ref{fig: Sim39}], there is a break in variance, in the presence of which our test outperforms the other tests. Furthermore, the \textit{t}-test and WT-test do not possess any power for low correlation (0 and 0.1) values of $\rho$. However, with increasing values of the correlation coefficient the power curves of all test appear to converge. In the other case [see figure \ref{fig: Sim410}], the variance follows a GARCH(1,1) model with a jump in variance. In this case, our test is only outperformed by the CD (1995) test, which has the greatest power. Nevertheless, the 10\% SS-PCC-POS test still outperforms the \textit{t}-test and WT-test.  
\begin{figure}[tbph]
\caption{Power comparisons: different tests. Normal error distributions with
different values of $\protect\rho $ in (\protect\ref{eq: errorsim}) and $\protect%
\theta =0.9$ in (\protect\ref{eq: theta})}
\begin{center}
\subfigure{\includegraphics[scale=0.58]{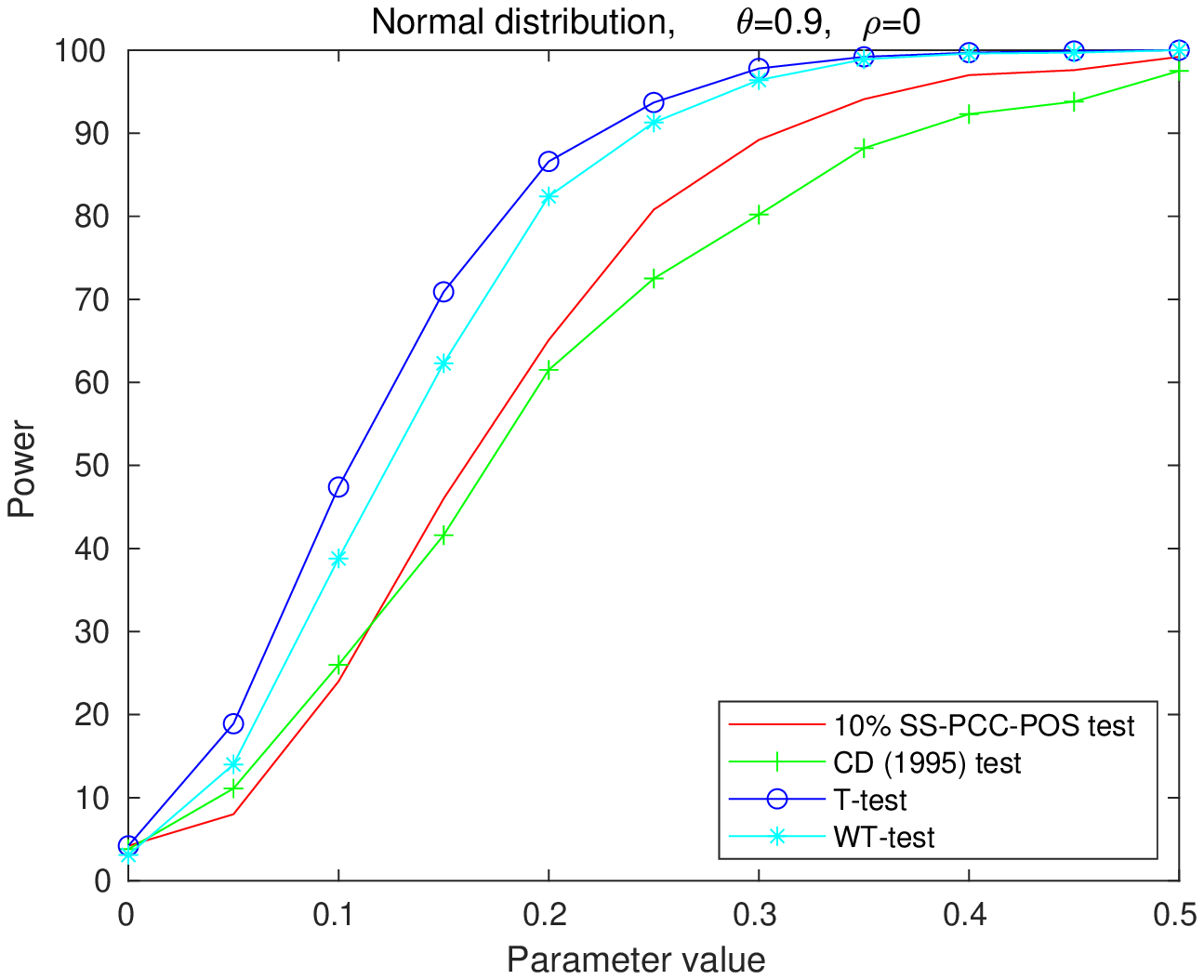}} %
\subfigure{\includegraphics[scale=0.58]{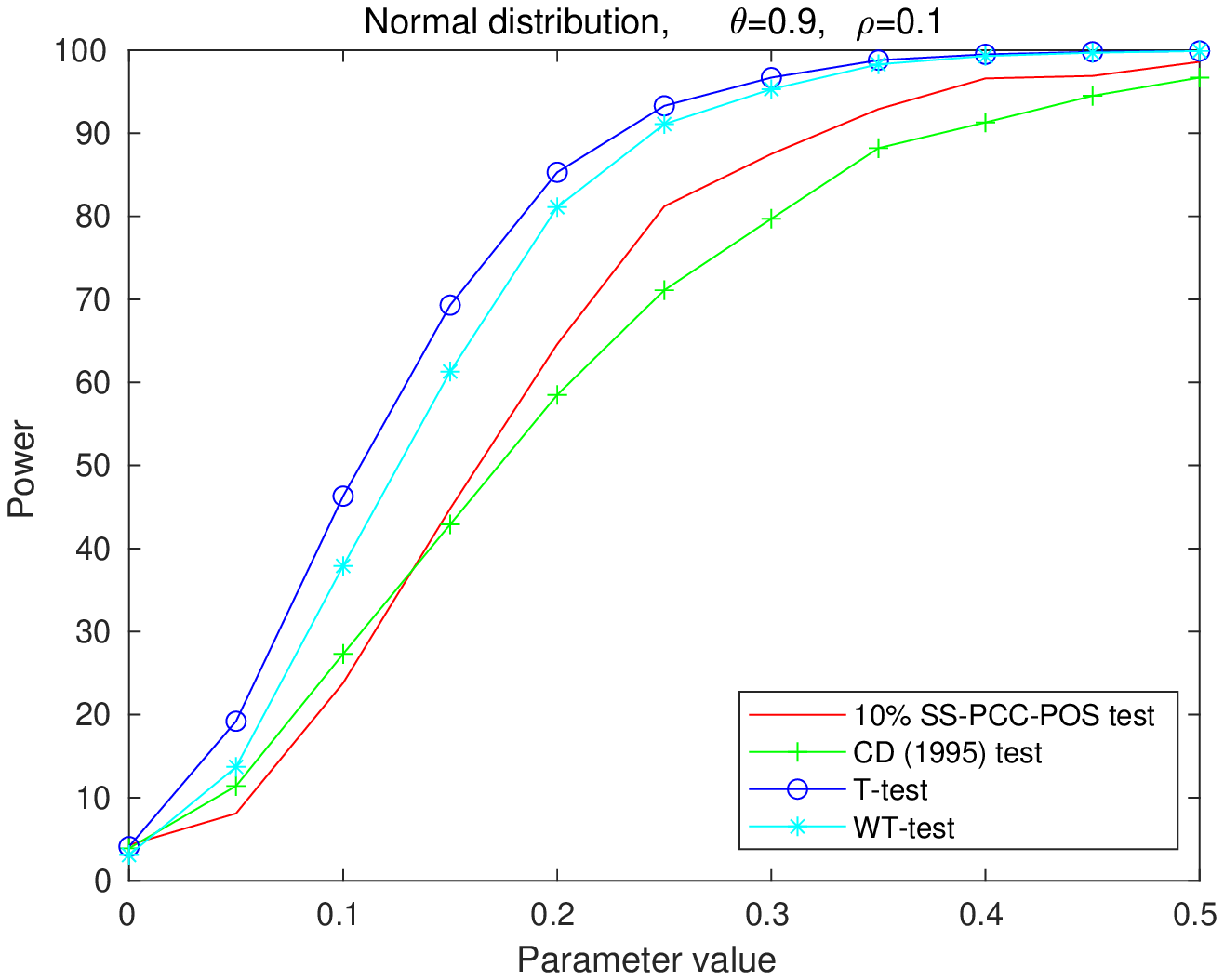}} \\[0pt]
\subfigure{\includegraphics[scale=0.58]{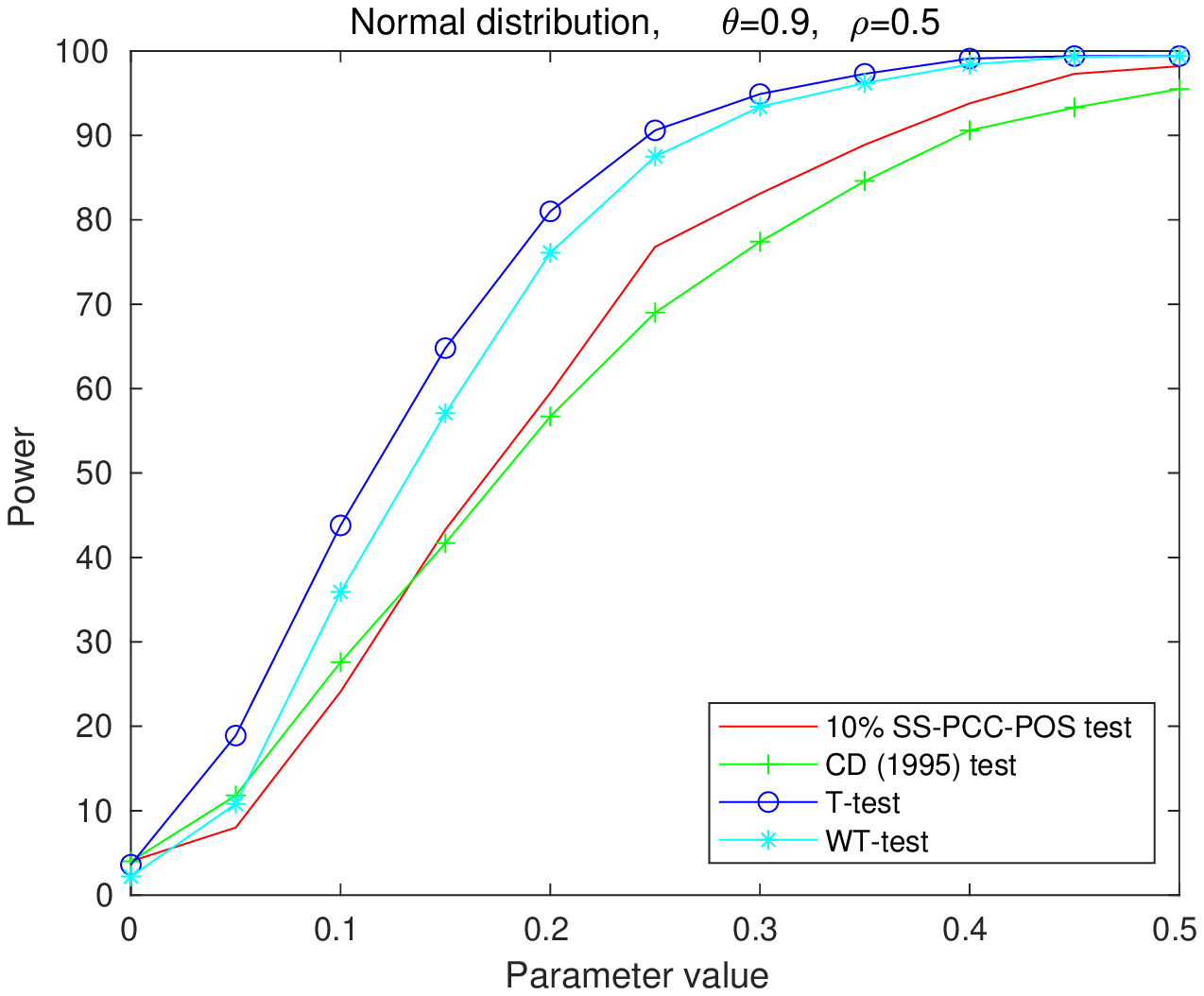}} %
\subfigure{\includegraphics[scale=0.58]{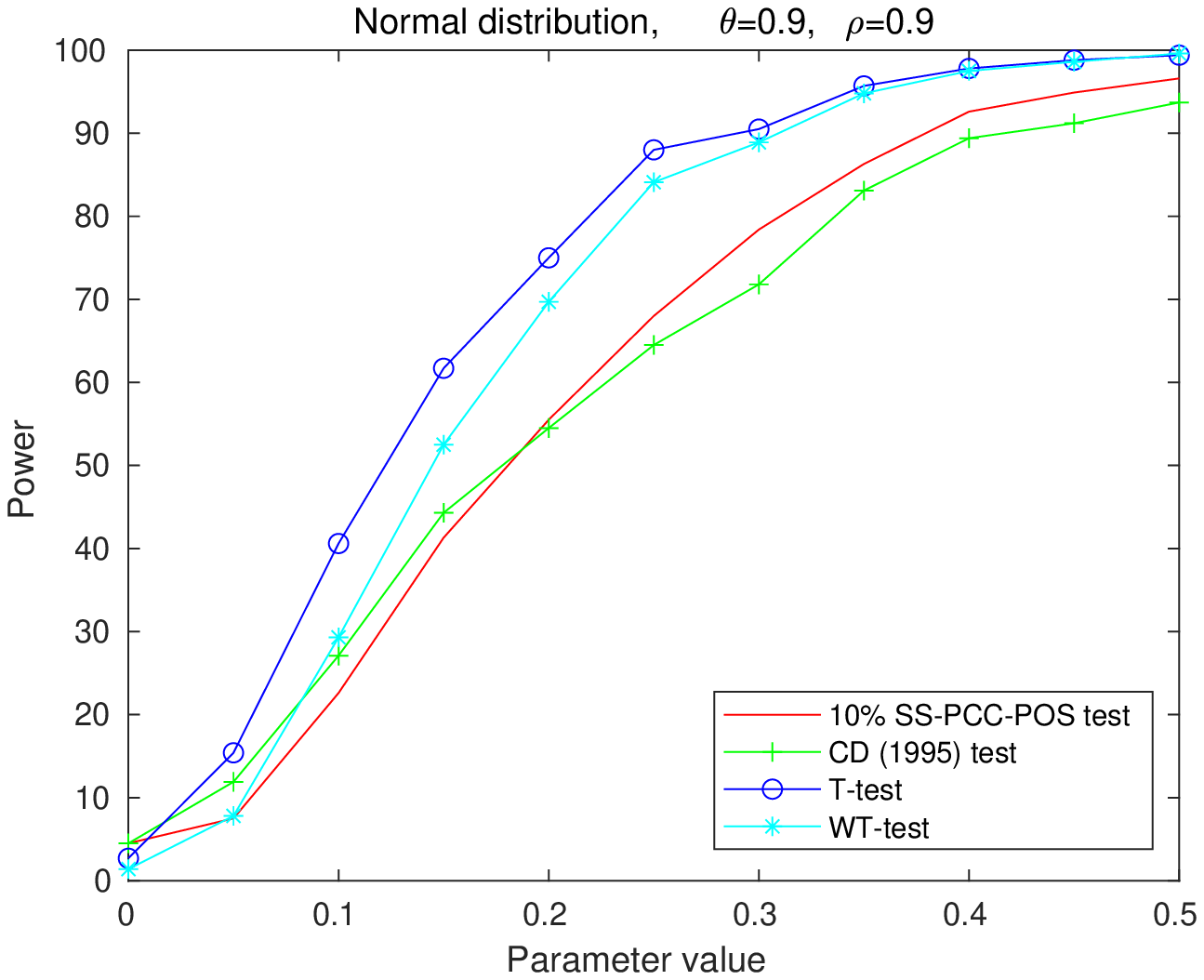}}\\[0pt]
\end{center}
\doublespacing
Note: These figures compare the power curves of the 10\% split-sample PCC-POS test
[10\% SS-PCC-POS test] with: (1) the \textit{t}-test; (2) the sign-based test
proposed by Campbell and Dufour (1995) [CD (1995) test]; and (3) the \textit{t}-test based
on White's (1980) variance correction [WT-test]. 
\label{fig: Sim17}
\end{figure}

\begin{figure}[tbph]
\caption{Power comparisons: different tests. Cauchy error distributions with
different values of $\protect\rho $ in (\protect\ref{eq: errorsim}) and $\protect%
\theta =0.9$ in (\protect\ref{eq: theta})}
\begin{center}
\subfigure{\includegraphics[scale=0.57]{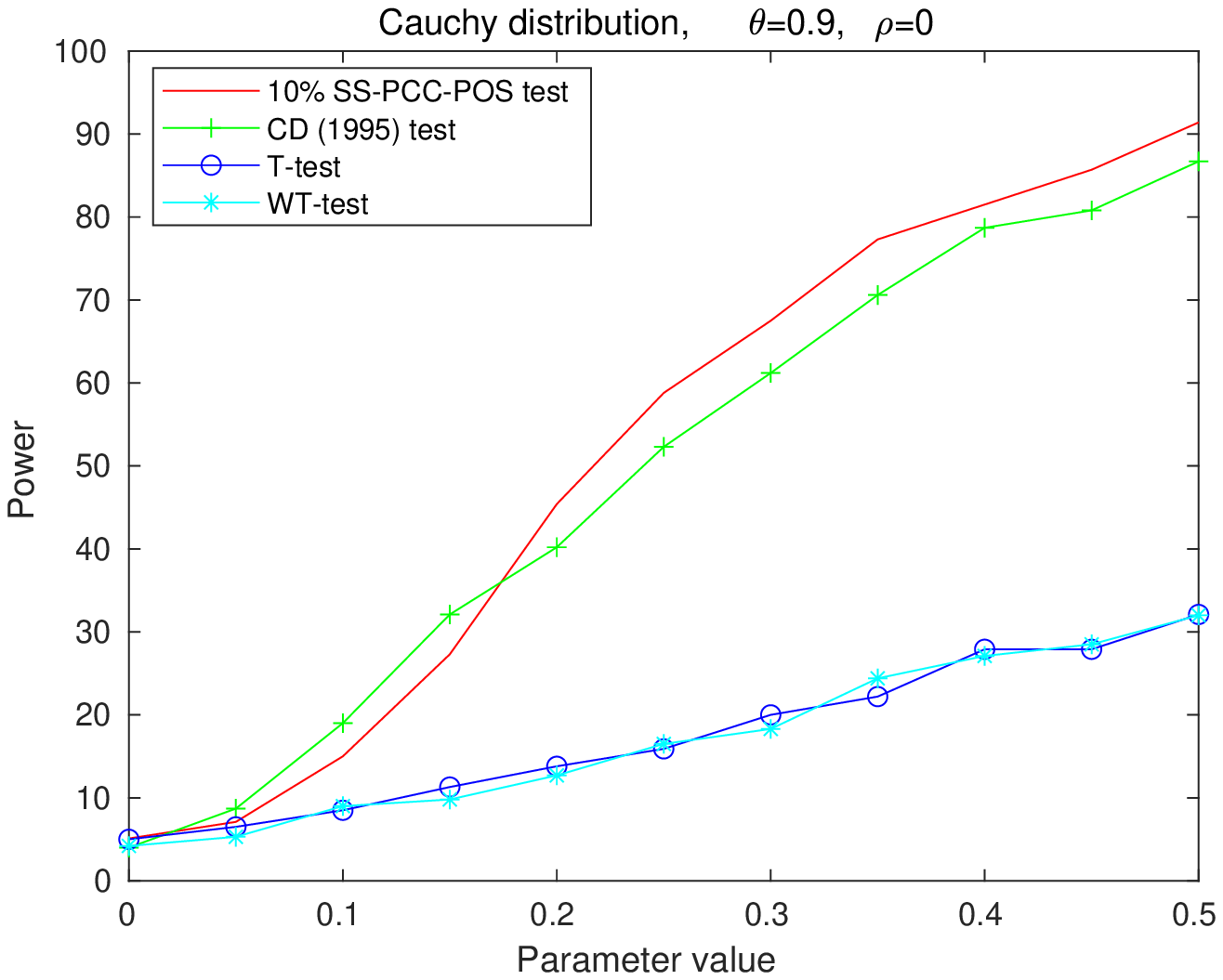}} %
\subfigure{\includegraphics[scale=0.57]{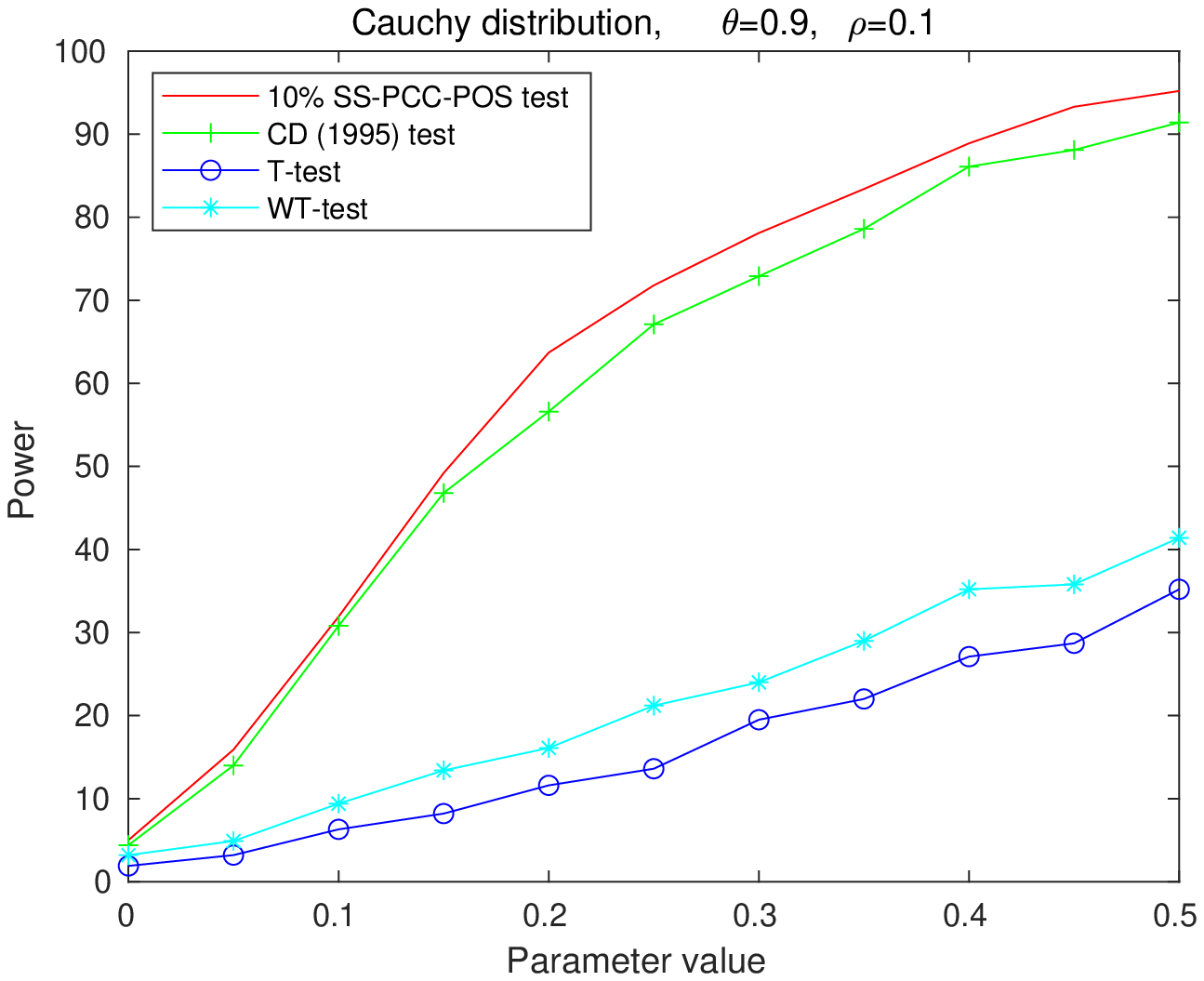}} \\[0pt]
\subfigure{\includegraphics[scale=0.57]{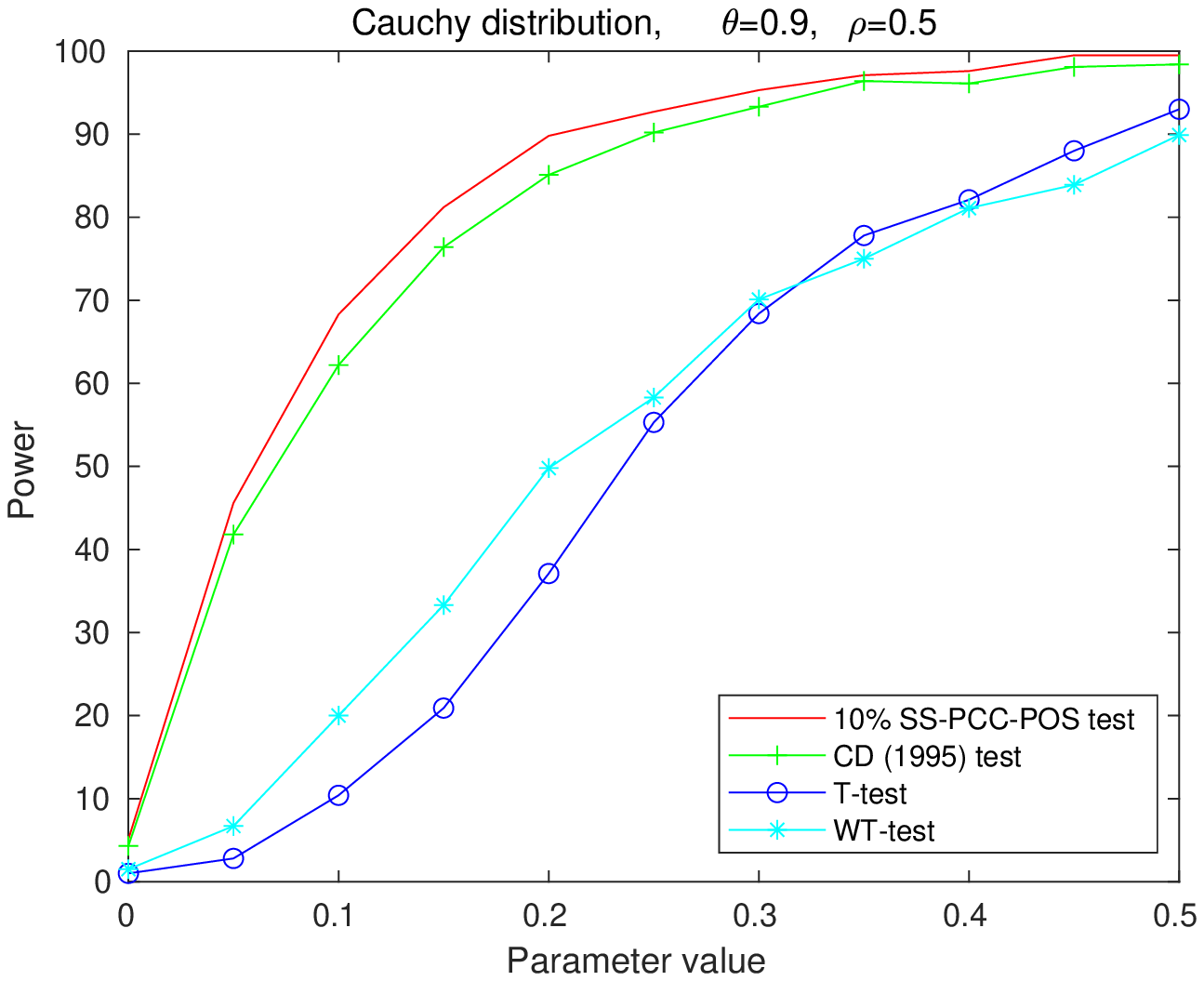}} %
\subfigure{\includegraphics[scale=0.57]{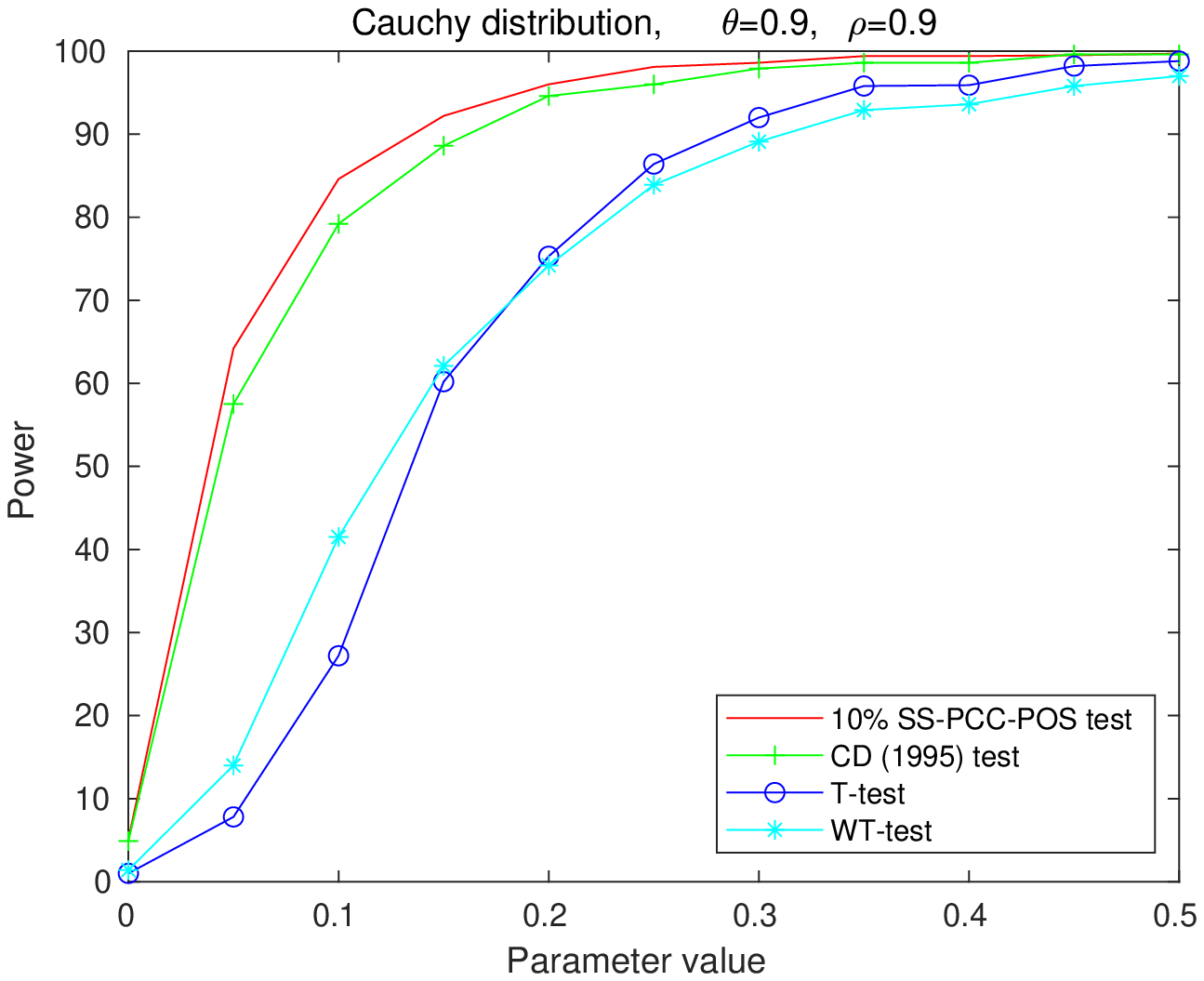}}\\[0pt]
\end{center}
\doublespacing
 Note: These figures compare the power curves of the 10\% split-sample PCC-POS test
[10\% SS-PCC-POS test] with: (1) the \textit{t}-test; (2) the sign-based test
proposed by Campbell and Dufour (1995) [CD (1995) test]; and (3) the \textit{t}-test based
on White's (1980) variance correction [WT-test]. 
\label{fig: Sim28}
\end{figure}

\begin{figure}[tbph]
\caption{Power comparisons: different tests. Student's $t$ error distributions with 2 degrees of freedom [i.e $t(2)$], with
different values of $\protect\rho $ in (\protect\ref{eq: errorsim}) and $\protect%
\theta =0.9$ in (\protect\ref{eq: theta})}
\begin{center}
\subfigure{\includegraphics[scale=0.58]{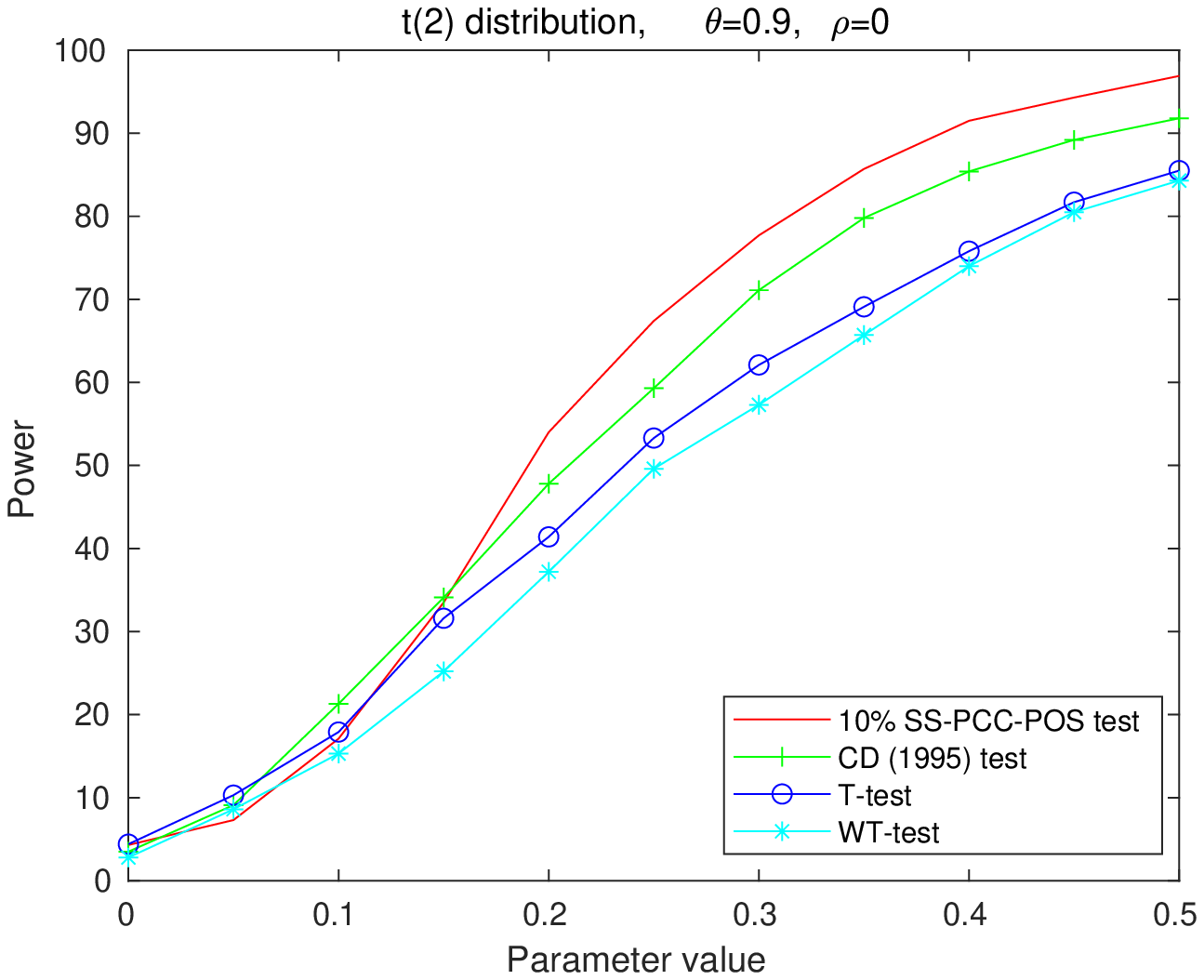}} %
\subfigure{\includegraphics[scale=0.58]{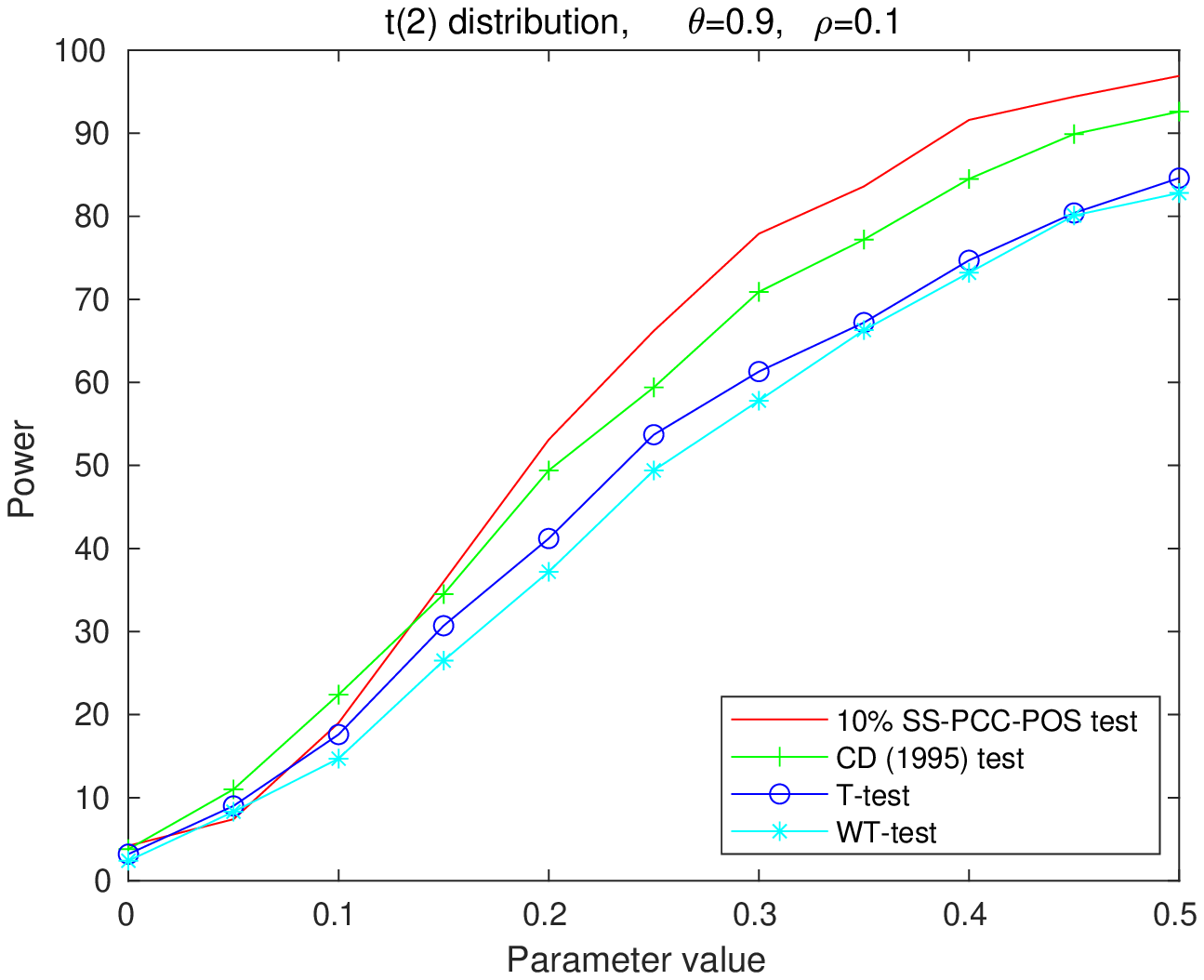}} \\[0pt]
\subfigure{\includegraphics[scale=0.58]{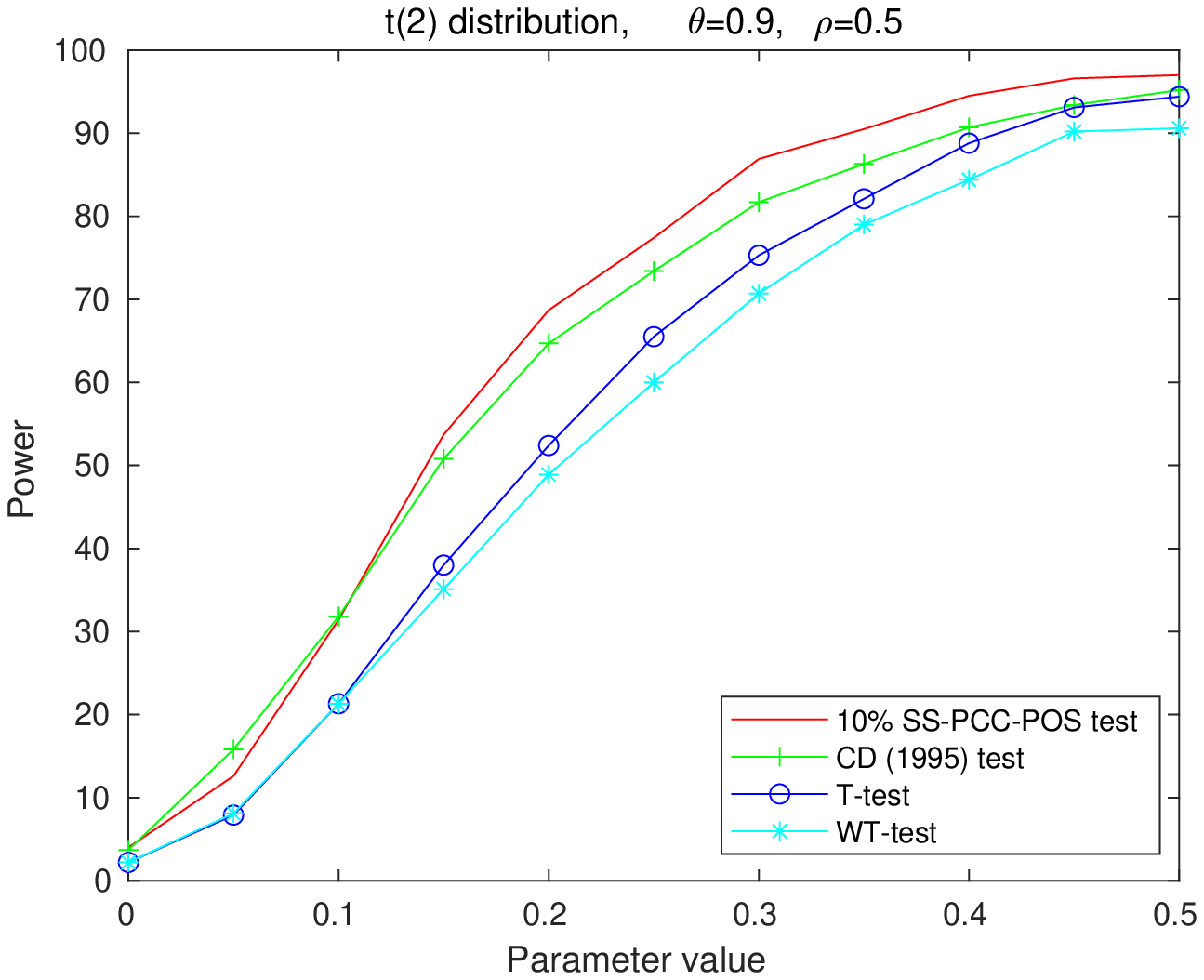}} %
\subfigure{\includegraphics[scale=0.58]{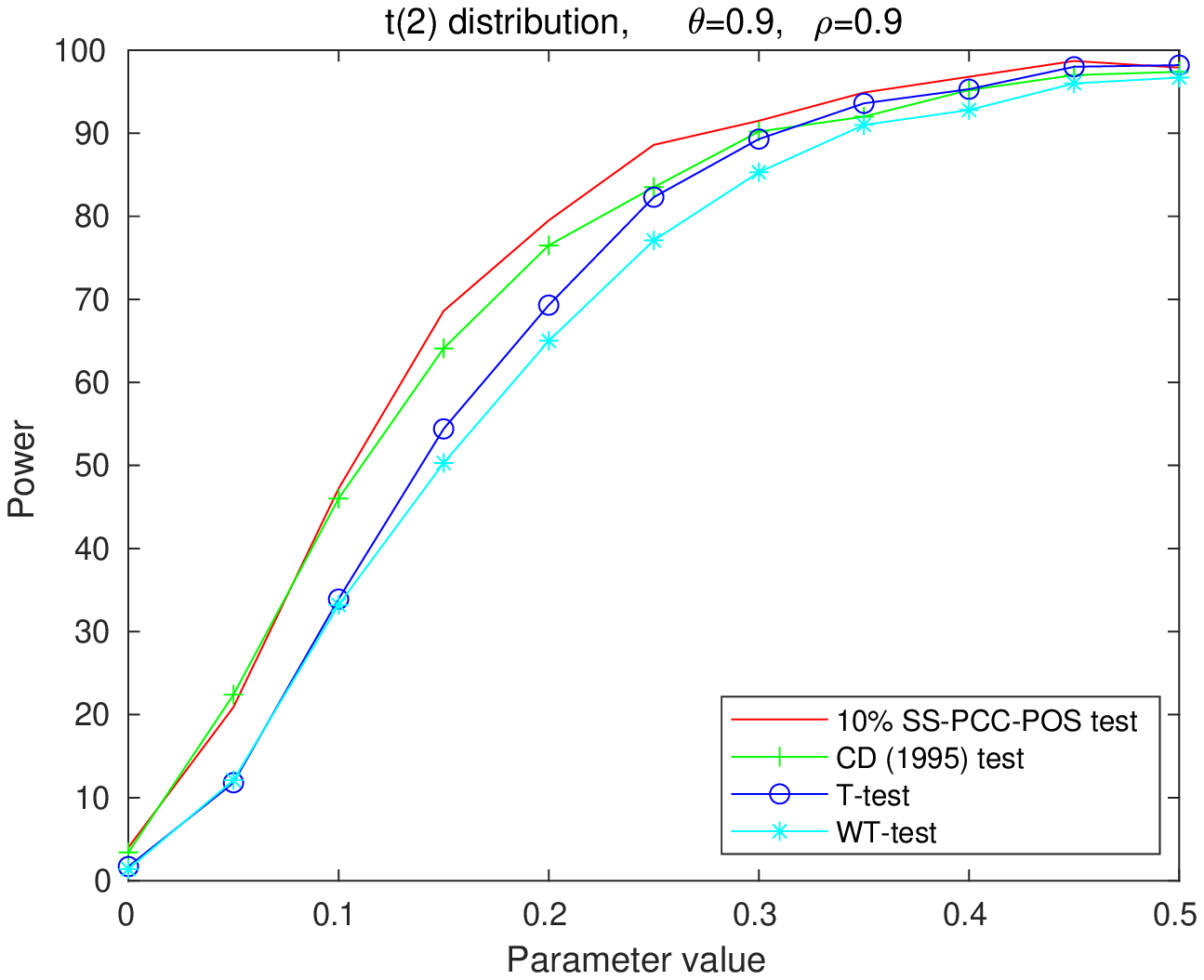}}\\[0pt]
\end{center}
\doublespacing
Note: These figures compare the power curves of the 10\% split-sample PCC-POS test
[10\% SS-PCC-POS test] with: (1) the \textit{t}-test; (2) the sign-based test
proposed by Campbell and Dufour (1995) [CD (1995) test]; and (3) the \textit{t}-test based
on White's (1980) variance correction [WT-test]. 
\label{fig: Sim39}
\end{figure}

\begin{figure}[tbph]
\caption{Power comparisons: different tests. Mixture error distributions with
different values of $\protect\rho $ in (\protect\ref{eq: errorsim}) and $\protect%
\theta =0.9$ in (\protect\ref{eq: theta})}
\begin{center}
\subfigure{\includegraphics[scale=0.58]{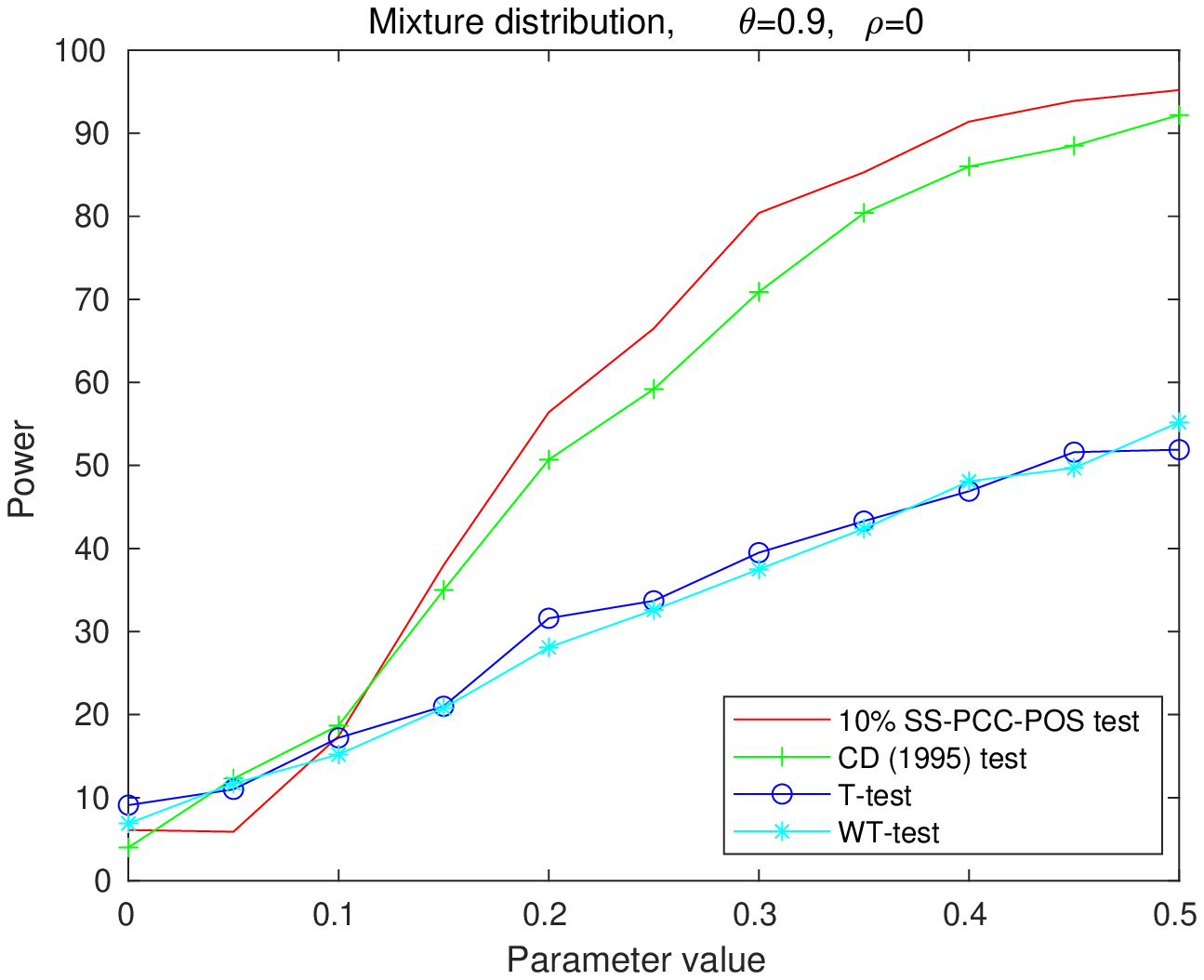}} %
\subfigure{\includegraphics[scale=0.58]{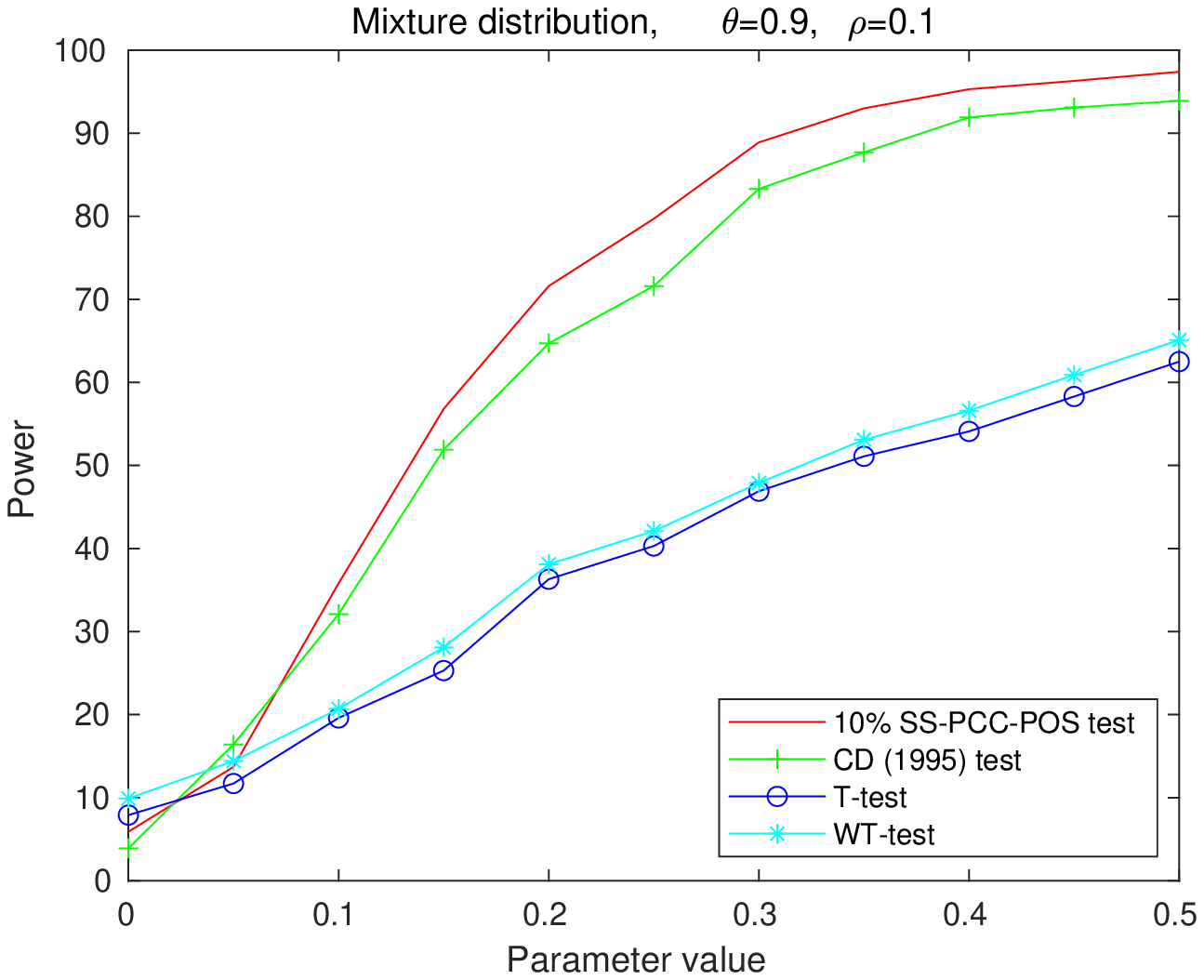}} \\[0pt]
\subfigure{\includegraphics[scale=0.58]{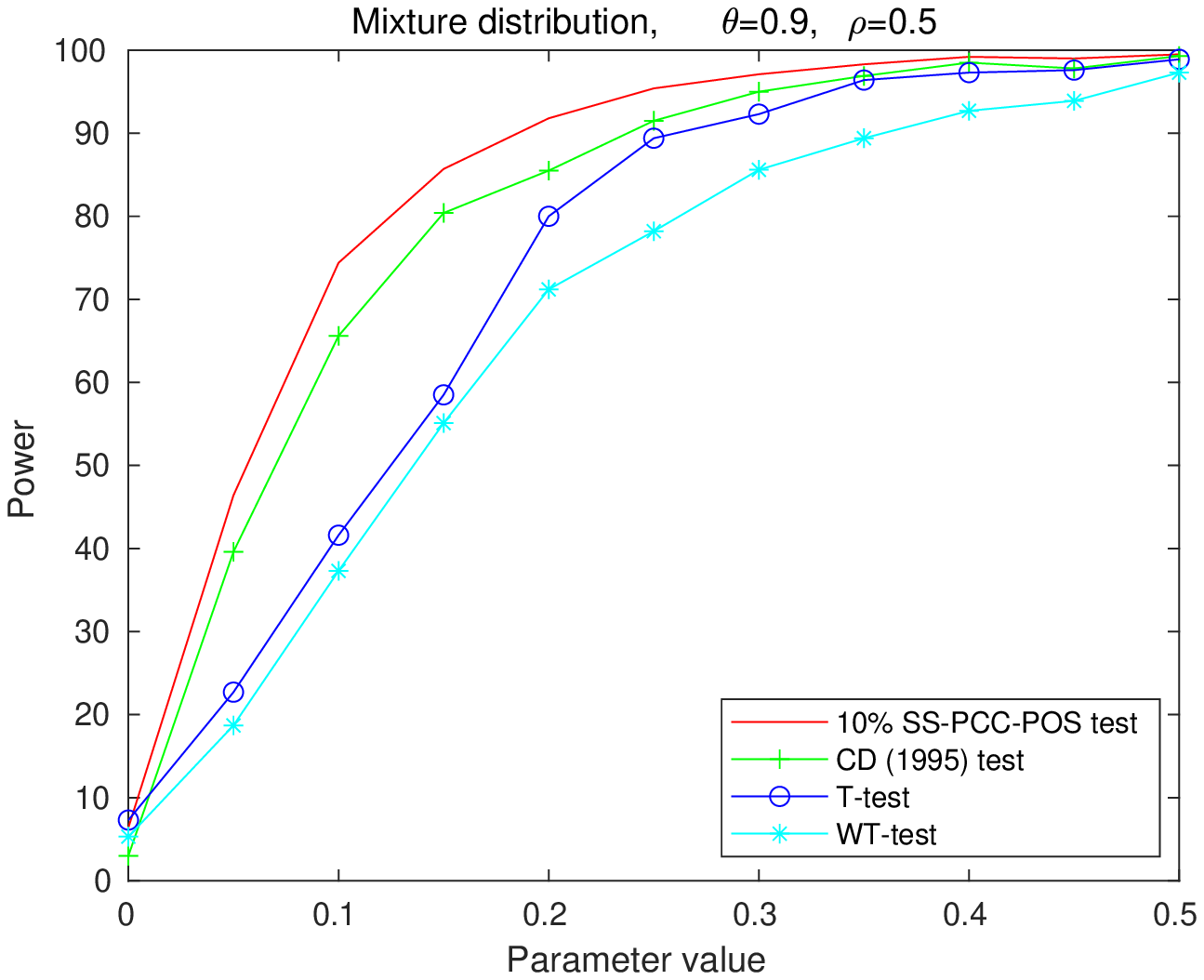}} %
\subfigure{\includegraphics[scale=0.58]{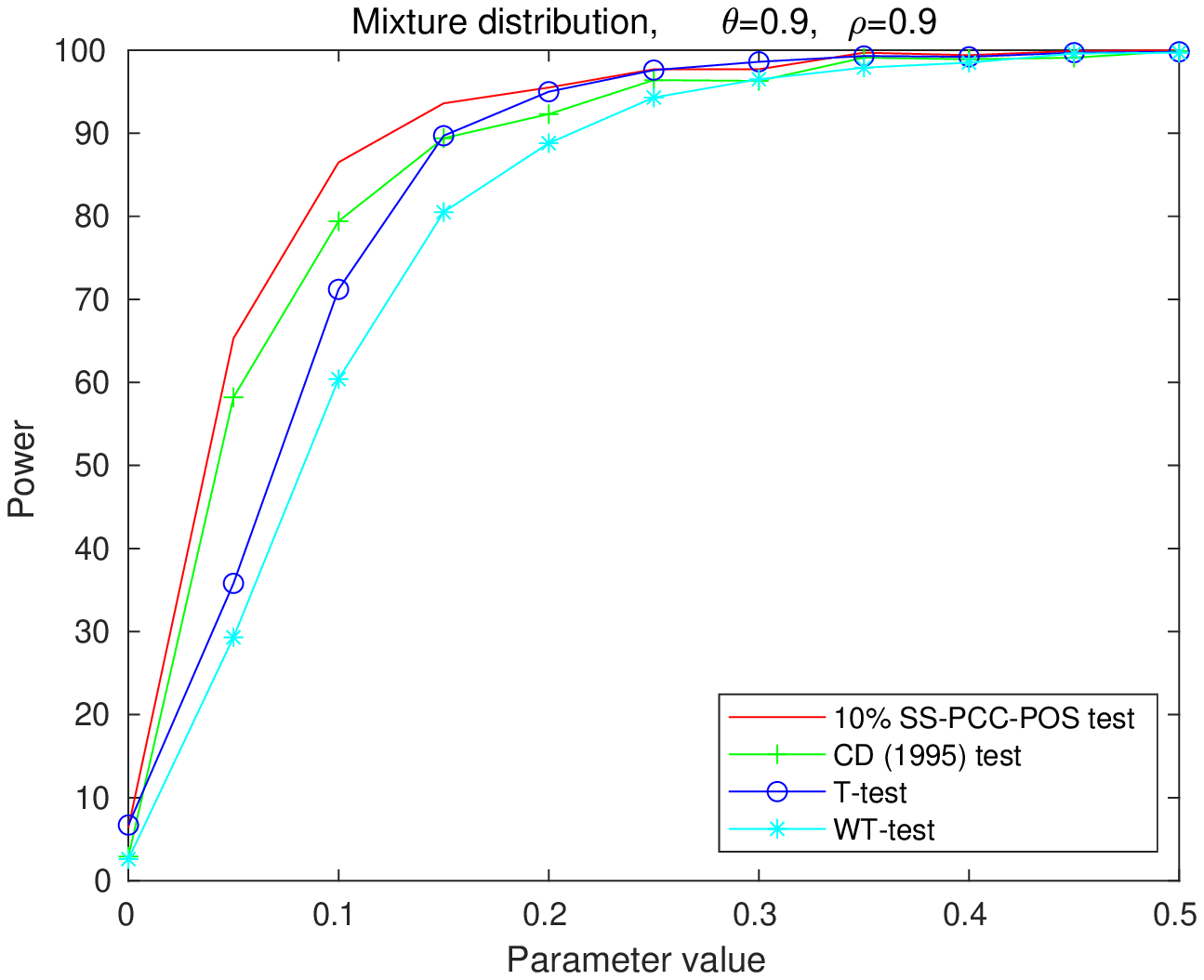}}\\[0pt]
\end{center}
\doublespacing
Note: These figures compare the power curves of the 10\% split-sample PCC-POS test
[10\% SS-PCC-POS test] with: (1) the \textit{t}-test; (2) the sign-based test
proposed by Campbell and Dufour (1995) [CD (1995) test]; and (3) the \textit{t}-test based
on White's (1980) variance correction [WT-test].
\label{fig: Sim410}
\end{figure}

\begin{figure}[tbph]
\caption{Power comparisons: different tests. Normal error distributions with break in variance,
with different values of $\protect\rho $ in (\protect\ref{eq: errorsim}) and $\protect%
\theta =0.9$ in (\protect\ref{eq: theta})}
\begin{center}
\subfigure{\includegraphics[scale=0.58]{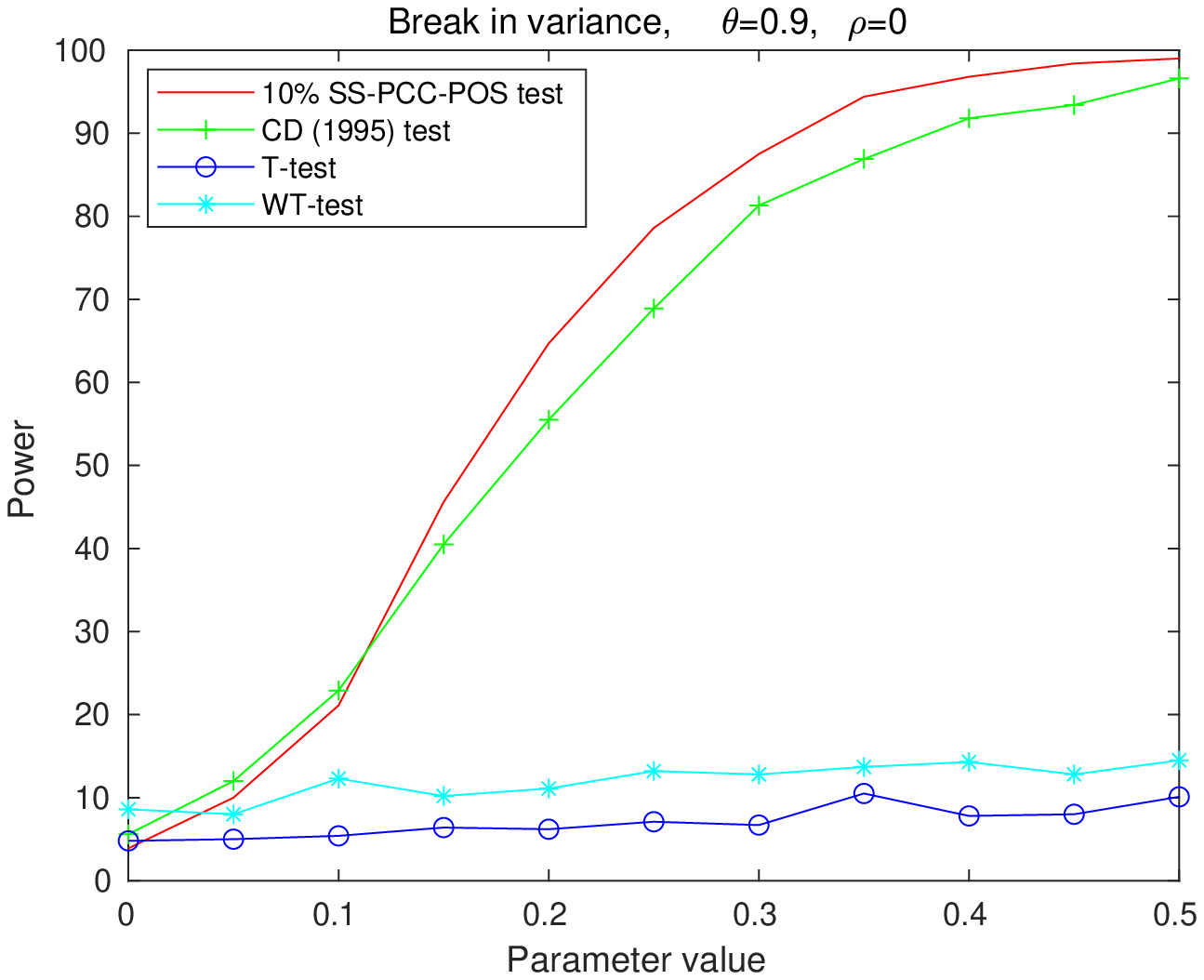}} %
\subfigure{\includegraphics[scale=0.58]{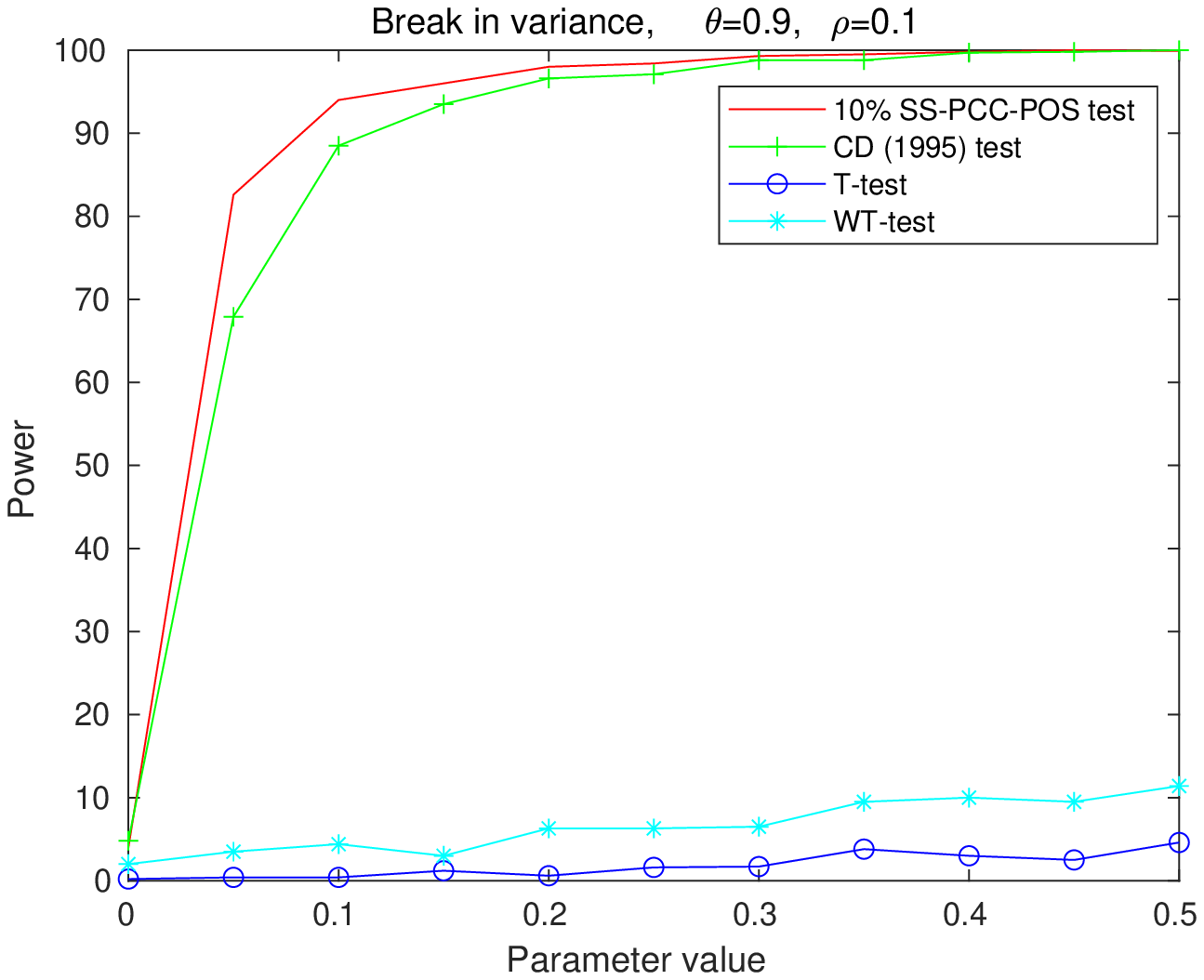}} \\[0pt]
\subfigure{\includegraphics[scale=0.58]{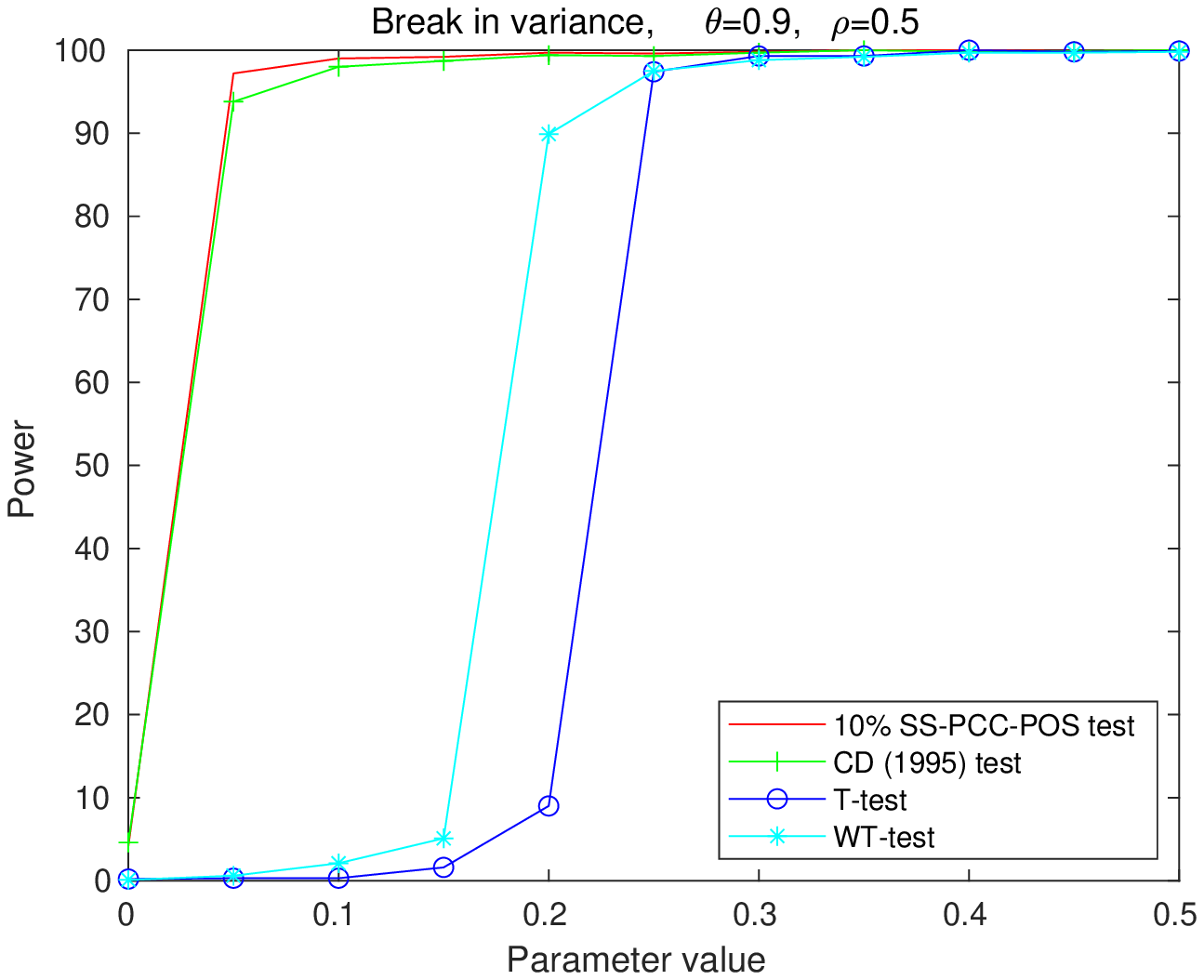}} %
\subfigure{\includegraphics[scale=0.58]{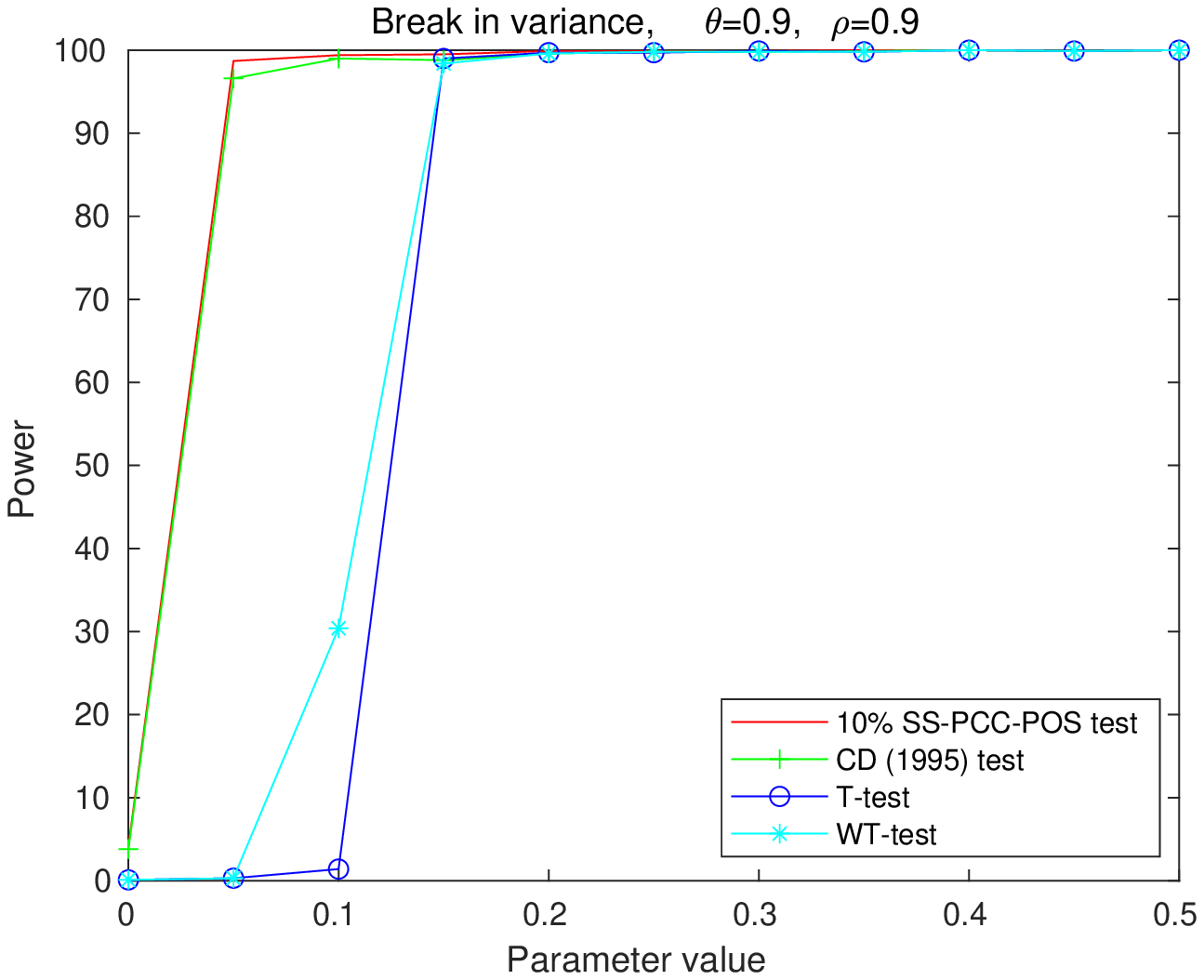}}\\[0pt]
\end{center}
\doublespacing
Note: These figures compare the power curves of the 10\% split-sample PCC-POS test
[10\% SS-PCC-POS test] with: (1) the \textit{t}-test; (2) the sign-based test
proposed by Campbell and Dufour (1995) [CD (1995) test]; and (3) the \textit{t}-test based
on White's (1980) variance correction [WT-test].
\label{fig: Sim511}
\end{figure}

\begin{figure}[tbph]
\caption{Power comparisons: different tests. Normal error distributions GARCH(1,1) plus jump invariance, with
different values of $\protect\rho $ in (\protect\ref{eq: errorsim}) and $\protect%
\theta =0.9$ in (\protect\ref{eq: theta})}
\begin{center}
\subfigure{\includegraphics[scale=0.58]{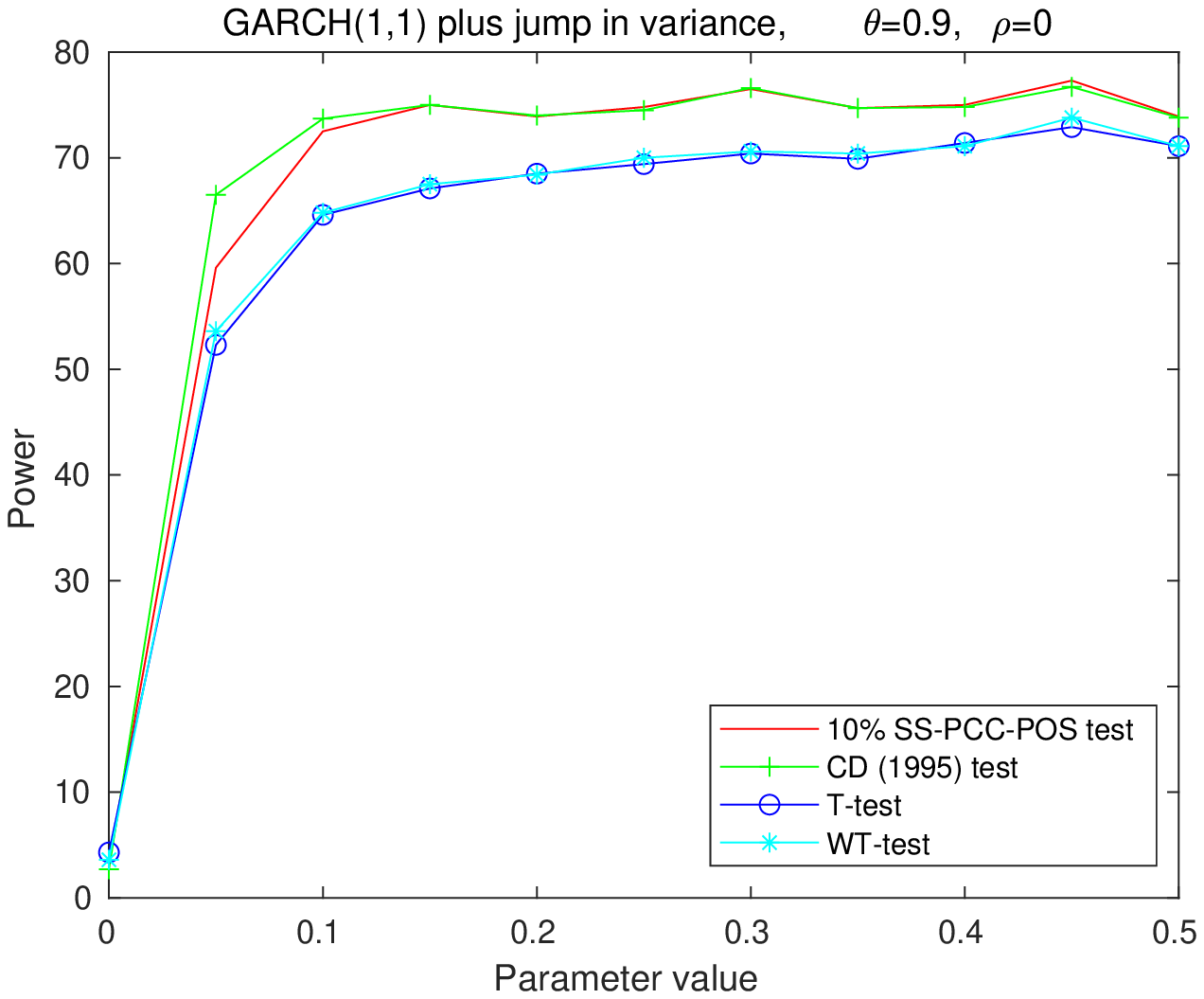}} %
\subfigure{\includegraphics[scale=0.58]{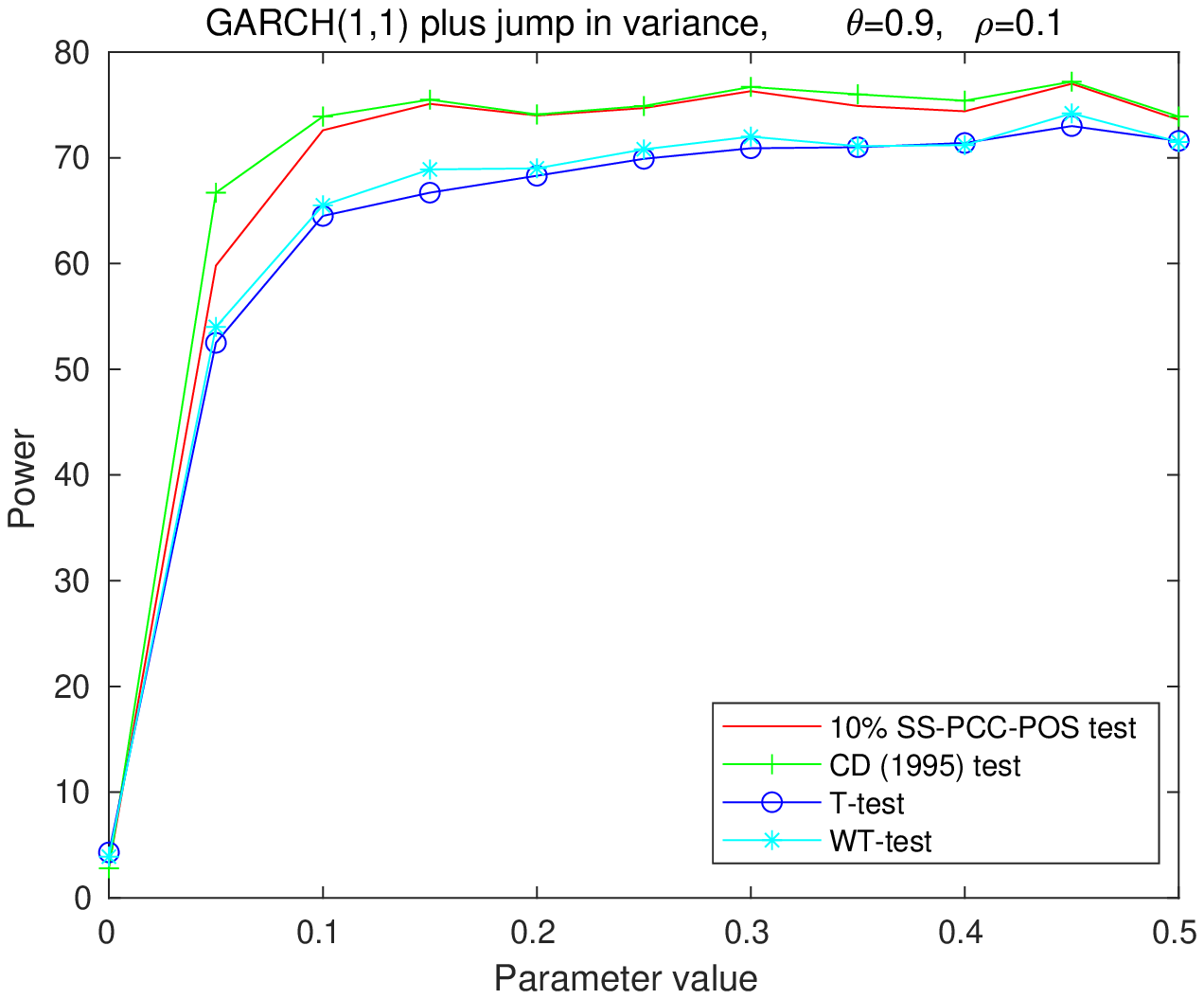}} \\[0pt]
\subfigure{\includegraphics[scale=0.58]{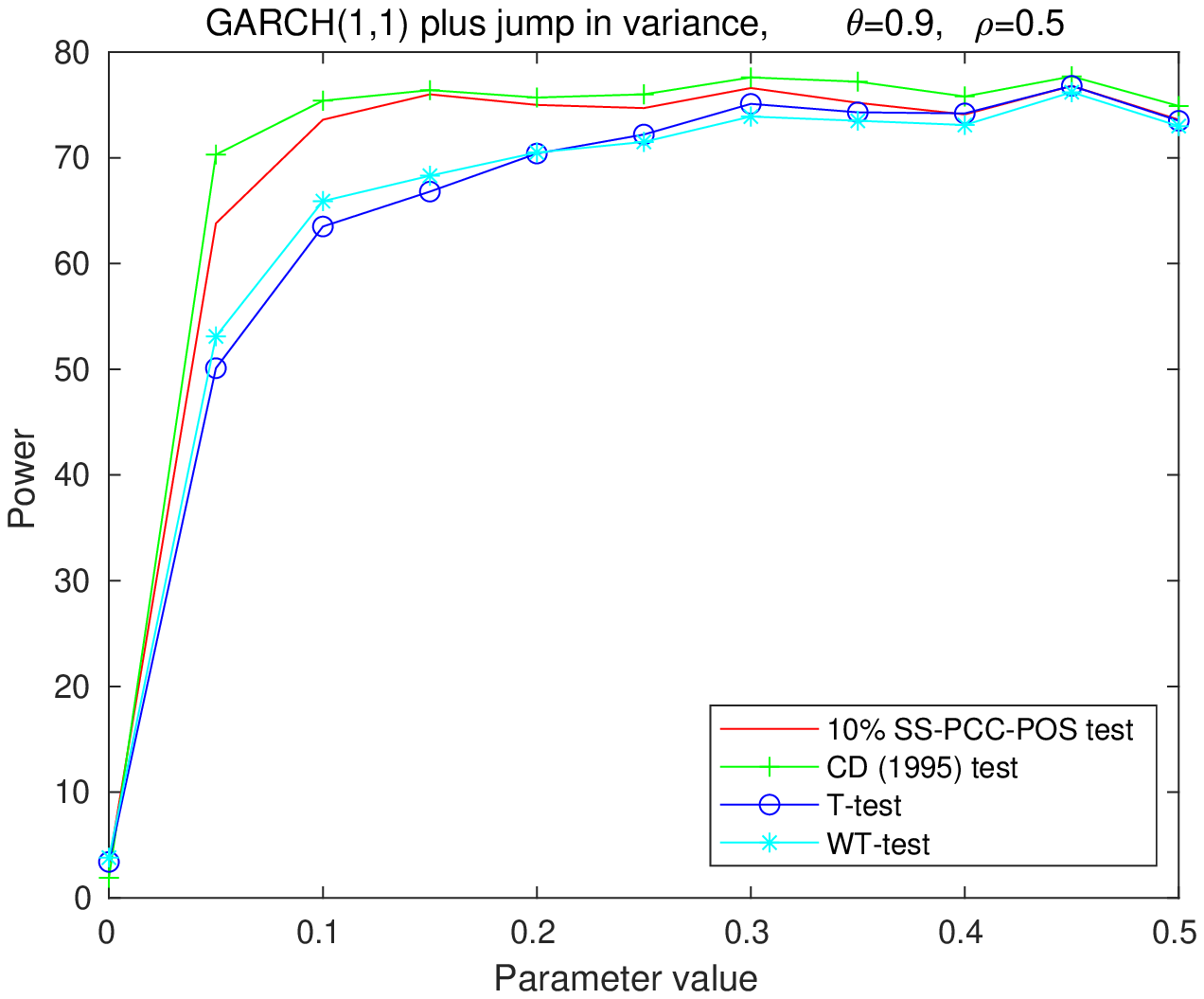}} %
\subfigure{\includegraphics[scale=0.58]{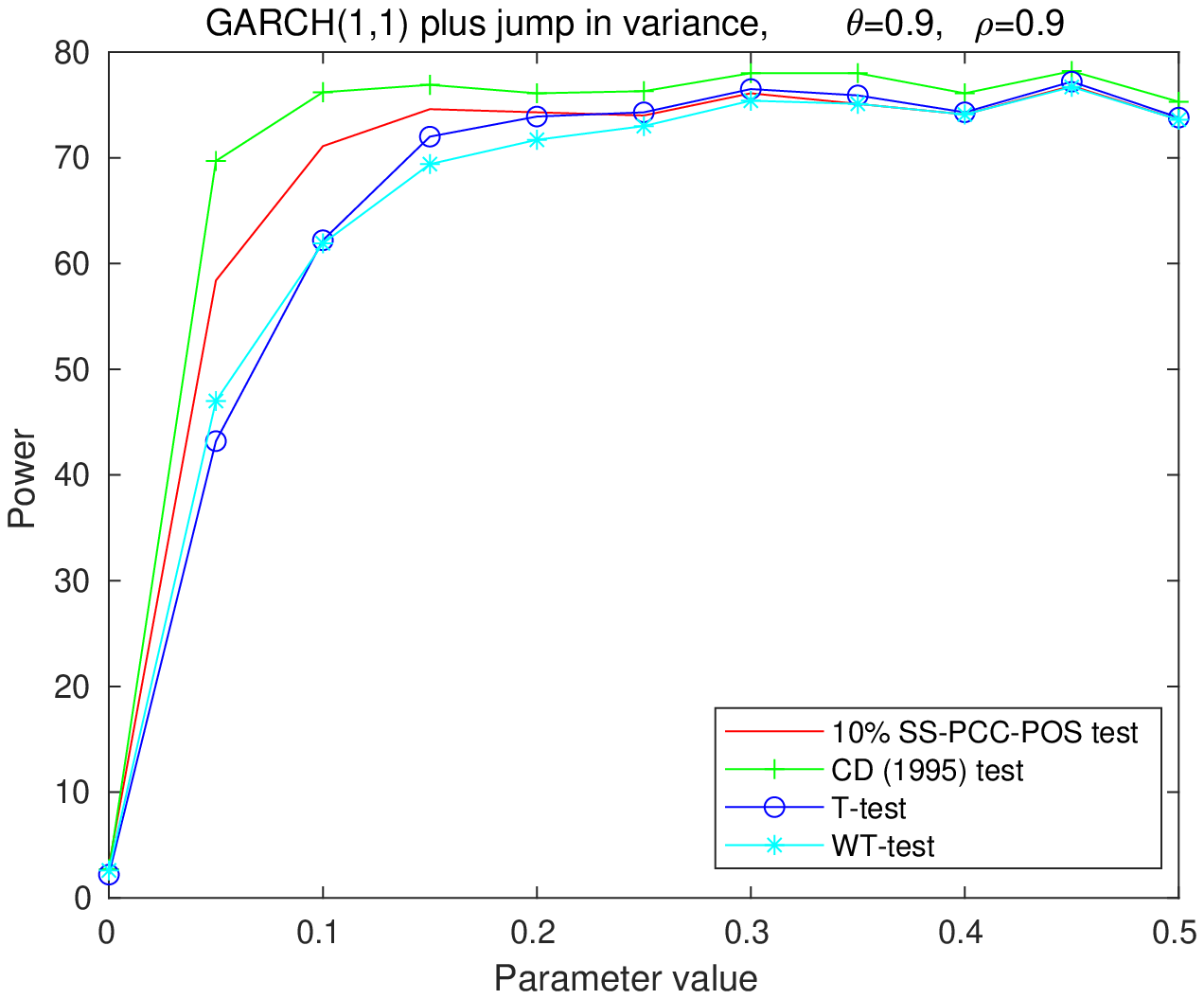}}\\[0pt]
\end{center}
\doublespacing
Note: These figures compare the power curves of the 10\% split-sample PCC-POS test
[10\% SS-PCC-POS test] with: (1) the \textit{t}-test; (2) the sign-based test
proposed by Campbell and Dufour (1995) [CD (1995) test]; and (3) the \textit{t}-test based
on White's (1980) variance correction [WT-test]. 
\label{fig: Sim612}
\end{figure}
\FloatBarrier
\section{Conclusion \label{ConclusionC3}}
{\hskip 1.5em}In this paper, we extend the exact point-optimal sign-based procedures proposed by \citet{dufour2010exact} to a predictive regression framework. We showed that by implementing the procedures for pair copula constructions of discrete data, we can derive exact and distribution-free sign-based statistics for dependent data in the context of linear and nonlinear predictive  regressions, without imposing any potentially restrictive assumptions. The proposed tests are valid, distribution-free and robust against heteroskedasticity of unknown form. Furthermore, they may be inverted to produce a confidence region for the vector (sub-vector) of parameters of the regression model. 

We further suggest a sequential estimation strategy for the D-vine PCC and discuss the choice of the copula family. As the proposed sign statistics depend on the alternative hypothesis, another problem consists of finding an alternative that controls size and maximizes the power. In line with \citet{dufour2010exact}, we find that when 10\% of sample is used to estimate the alternative and the rest to compute the test-statistic, our procedures have the optimal power and are closest to the power envelope.

Finally, we present a Monte Carlo study to assess the performance of the proposed tests in terms of size control and power, by comparing them to some other tests that are intended to be robust against heteroskedasticity. We consider a variety of different DGPs and we show that the 10\% split-sample point-optimal sign-test based on pair copula constructions is superior to the \textit{t}-test, \citet{dufour1995exact} sign-based test, and the \textit{t}-test based on \citet{white1980heteroskedasticity} variance correction in most cases.
\newpage

\section{Appendix}

\begin{proof}[Derivation of the Neyman-Pearson type sign-based statistic for testing the unpredictability hypothesis for $T\leq 3$]

The likelihood function of sample in terms of signs $s(y_{1}),\cdots,s(y_{T})$ conditional on $X$ is
\begin{equation*}
L\left( U(T),\bm{\beta},X \right) =P\left[
s(y_{1})=s_{1},\cdots,s(y_{T})=s_{T}\mid X\right] =\prod\limits_{t=1}^{T}%
P\left[ s(y_{t})=s_{t}\mid \text{\b{S}}_{t-1}=\text{\b{s}}_{t-1},X\right] ,
\end{equation*}%
for 
\begin{equation*}
\text{\b{S}}_{0}=\left\{ \emptyset \right\} ,\text{ \ \b{S}}_{t-1}=\left\{
s(y_{1}),\cdots,s(y_{t-1})\right\} ,\text{ for }t\geq 2,
\end{equation*}%
and%
\begin{equation*}
P\left[ s(y_{1})=s_{1}\mid \text{\b{S}}_{0}=\text{\b{s}}_{0},X\right] =P%
\left[s(y_{1})=s_{1}\mid X\right] ,
\end{equation*}%
where each $s_{t}$, for $1\leq t\leq T$, takes two possible values $0$ and $%
1 $. Given model (\ref{eq: DGP}) and assumption (\ref{eq: median}), under the null hypothesis of unpredictability, the signs $s(\varepsilon_{t})$, for $1\leq t\leq T$,\ are i.i.d conditional on $X$ according to $Bi(1,0.5)$. Then, the signs $s(y_{t}),$ for $1\leq t\leq T$, will also be i.i.d conditional on $X$
with%
\begin{equation*}
P\left[ s(y_{t})=1\mid X\right] =P\left[ s(y_{t})=0\mid X\right] =%
\frac{1}{2},\quad\text{for}\quad t=1,\cdots,T.
\end{equation*}%
Consequently, under $H_{0}$%
\begin{equation*}
L_{0}\left( U(T),\bm{0},X\right) =\prod\limits_{t=1}^{T}
P\left[ s(y_{t})=s_{t}\mid X\right] =\left( \frac{1}{2}\right) ^{T}
\end{equation*}%
and under $H_1$ we have%
\begin{equation*}
L_1\left( U(T),\bm{\beta} _{1},X\right) =\prod\limits_{t=1}^{T}%
P\left[ s(y_{t})=s_{t}\mid \text{\b{S}}_{t-1}=\text{\b{s}}_{t-1}, X\right]
\end{equation*}%
where now, for $t=1,\cdots,T,$%
\begin{equation*}
y_{t}=\bm{\beta}_1'\bm{x}_{t-1}+\varepsilon_{t}
\end{equation*}%
The log-likelihood ratio is given by%
\begin{equation*}
\ln \left\{ \frac{L_1\left( U(T),\bm{\beta}_{1},X\right) }{L
_{0}\left( U(T),\bm{0},X\right) }\right\} =\sum\limits_{t=1}^{T}\ln \left\{
P\left[ s(y_{t})=s_{t}\mid \text{\b{S}}_{t-1}=\text{\b{s}}_{t-1},X\right] \right\} -\text{T}\ln
\left\{ \frac{1}{2}\right\} .
\end{equation*}%
According to Neyman-Pearson lemma [see e.g. Lehmann (1959), page 65], the
best test to test $H_{0}$ against $H_1,$ based on $s(y_{1}),\cdots,s(y_{T}),$
rejects $H_{0}$ when%
\begin{equation*}
SL_T(\bm{\beta}_1)=\ln \left\{ \frac{L_1\left( U(T),\bm{\beta}_{1},X\right) }{L_{0}\left( U(T),\bm{0},X\right) }\right\} \geq c
\end{equation*}%
or when%
\begin{equation*}
\sum\limits_{t=1}^{T}\ln \left\{ P\left[ s(y_{t})=s_{t}\mid \text{%
\b{S}}_{t-1}=\text{%
\b{s}}_{t-1},X\right] \right\} \geq c_1\equiv c+T\ln\left(\frac{1}{2}\right),
\end{equation*}%
The critical value, say $c_1$ is given by the smallest constant $c_1$ such that%
\begin{equation*}
P\left(\sum\limits_{t=1}^{T}\ln \left\{ P\left[ s(y_{t})=s_{t}\mid \text{%
\b{S}}_{t-1}=\text{%
\b{s}}_{t-1},X\right] \right\} >c_1\mid
H_{0}\right) \leq \alpha .
\end{equation*}
We have
\begingroup
\allowdisplaybreaks
\begin{align*}
\ln\left\{P[s(y_1)=s_1\mid \text{\b{S}}_{0}=\text{\b{s}}_{0},X]\right\}&=\ln\left\{P[s(y_1)=s_1\mid X]\right\}\\
&= s(y_1)\ln P[y_1\geq0\mid X]+(1-s(y_1))\ln P[y_1<0\mid X]\\
&=s(y_1)\ln \left\{\frac{P[y_1\geq 0\mid X]}{P[y_1< 0\mid X]}\right\}+ \ln P[y_1< 0\mid X]\\
&=s(y_1)\ln \left\{\frac{P[\varepsilon_1\geq -\bm{\beta}_1' \bm{x}_{t-1}\mid X]}{P[\varepsilon_1< -\bm{\beta}_1'\bm{x}_{t-1}\mid X]}\right\}+ \ln P[\varepsilon_1< -\bm{\beta}_1'\bm{x}_{t-1}\mid X],
\end{align*}
\endgroup
and for $t=2,\cdots,T$, with $T\leq3$ we have
\begingroup
\allowdisplaybreaks
\begin{align*}
 \sum\limits_{t=2}^{T} \ln P\left[s(y_{t})=s_{t}\mid \text{\b{S}}_{t-1}=\text{\b{s}}_{t-1},X\right] &=\sum\limits_{t=2}^{T}\ln\left(\frac{P[s(y_t)=s_t,s(y_{t-1})=s_{t-1}\mid \text{\b{S}}%
_{t-2}=\text{\b{s}}%
_{t-2},X]}{P[s(y_{t-1})=s_{t-1}\mid \text{\b{S}}%
_{t-2}=\text{\b{s}}%
_{t-2},X]}\right)\\
&=\sum\limits_{t=2}^{T}\ln\Bigg(\sum\limits_{k_t=0,1}\sum\limits_{k_{t-1}=0,1}(-1)^{k_t+k_{t-1}}\\
&\textcolor{white}{=}\times\left\{P[s(y_t)\leq s_t-k_t,s(y_{t-1})\leq s_{t-1}-k_{t-1}\mid \text{\b{S}}%
_{t-2}=\text{\b{s}}_{t-2},X]\right\}\\
&\textcolor{white}{=}/P[s(y_{t-1})=s_{t-1}\mid \text{\b{S}}_{t-2}=\text{\b{s}}_{t-2},X]\Bigg)\\
\textcolor{white}{\left\{ P\left[s(y_{t})=s_{t}\mid \text{\b{S}}_{t-1}=\text{\b{s}}_{t-1}\right] \right\} }&=\sum\limits_{t=2}^{T}\ln\Bigg(\sum\limits_{k_t=0,1}\sum\limits_{k_{t-1}=0,1}(-1)^{k_t+k_{t-1}}\\
& \textcolor{white}{=}\times\Big\{C_{s(y_t),s(y_{t-1})\mid \text{\b{S}}%
_{t-2}}\left(F_{s(y_t)\mid \text{\b{S}}%
_{t-2} }(s_t-k_t\mid \text{\b{s}}%
_{t-2},X),\right.\\
& \textcolor{white}{=}\hspace{14em}\left.F_{s(y_{t-1})\mid \text{\b{S}}_{t-2}}(s_{t-1}-k_{t-1}\mid \text{\b{s}}_{t-2},X)\right)\Big\}\\
&\textcolor{white}{=}/P[s(y_{t-1})=s_{t-1}\mid \text{\b{S}}_{t-2}=\text{\b{s}}_{t-2},X]\Bigg)\\
&=\sum\limits_{t=2}^{T}\ln\Bigg\{\sum\limits_{k_t=0,1}\sum\limits_{k_{t-1}=0,1}(-1)^{k_t+k_{t-1}} \\
&\textcolor{white}{=} \times\Big\{C_{s(y_t),s(y_{t-1})\mid \text{\b{S}}%
_{t-2}}\left(F_{s(y_t)\mid \text{\b{S}}%
_{t-2} }(s_t-k_t\mid \text{\b{s}}%
_{t-2},X),\right.\\
& \textcolor{white}{=}\hspace{14em}\left.F_{s(y_{t-1})\mid \text{\b{S}}%
_{t-2}}(s_{t-1}-k_{t-1}\mid \text{\b{s}}%
_{t-2},X)\right)\Big\}\Bigg\}\\
&\textcolor{white}{=}-\sum\limits_{t=2}^{T}\ln\left\{P[s(y_{t-1})=s_{t-1}\mid \text{\b{S}}_{t-2}=\text{\b{s}}_{t-2},X]\right\}
\end{align*}%
\endgroup
Each argument $F_{s(y_t)\mid \text{\b{S}}%
_{t-2}}(s_t-k_t\mid \text{\b{s}}%
_{t-2},X)$ and $F_{s(y_{t-1})\mid \text{\b{S}}%
_{t-2}}(s_{t-1}-k_{t-1}\mid \text{\b{s}}%
_{t-2},X)$ in the copula expression above can be evaluated as follows
\begingroup
\allowdisplaybreaks
\begin{align*}
&F_{s(y_t)\mid\text{\b{S}}%
_{t-2}}(s_t-k_t\mid\text{\b{s}}%
_{t-2},X)=\\
&\left\{C_{s(y_t),s(y_{t-2})\mid\text{\b{S}}%
_{t-3}}(F(s_t-k_t\mid\text{\b{s}}%
_{t-3},X),F(s_{t-2}\mid\text{\b{s}}%
_{t-3},X))\right.\\
&\left.-C_{s(y_t),s(y_{t-2})\mid\text{\b{S}}%
_{t-3}}(F(s_t-k_t\mid\text{\b{s}}%
_{t-3},X),F(s_{t-2}-1\mid\text{\b{s}}%
_{t-3},X))\right\}/P[s(y_{t-2})=s_{t-2}\mid\text{\b{S}}%
_{t-3}=\text{\b{s}}%
_{t-3},X]
\end{align*}
 \endgroup
and similarly
\begingroup
\allowdisplaybreaks
\begin{align*}
&F_{s(y_{t-1})\mid\text{\b{S}}%
_{t-2}}(s_{t-1}-k_{t-1}\mid\text{\b{s}}%
_{t-2},X)=\\
&\left\{C_{s(y_{t-2}),s(y_{t-1})\mid\text{\b{S}}%
_{t-3}}(F(s_{t-2}\mid\text{\b{s}}%
_{t-3},X),F(s_{t-1}-k_{t-1}\mid\text{\b{s}}%
_{t-3},X))\right.\\
&\left.-C_{s(y_{t-2}),s(y_{t-1})\mid\text{\b{S}}%
_{t-3}}(F(s_{t-2}-1\mid\text{\b{s}}%
_{t-3},X),F(s_{t-1}-k_{t-1}\mid\text{\b{s}}%
_{t-3},X))\right\}/P[s(y_{t-2})=s_{t-2}\mid\text{\b{S}}%
_{t-3}=\text{\b{s}}%
_{t-3},X]
\end{align*}
\endgroup
Thus, for $T\leq 3$ the Neyman-Pearson type test statistic based on $%
s(y_{1}),\cdots,s(y_{T}),$ can be expressed as%
\begingroup
\allowdisplaybreaks
\begin{align*}
SL_T(\bm{\beta}_1)&=s(y_1)\ln \left\{\frac{P[\varepsilon_1\geq -\bm{\beta}_1' \bm{x}_{t-1}\mid X]}{P[\varepsilon_1< -\bm{\beta}_1'\bm{x}_{t-1}\mid X]}\right\}+ \ln P[\varepsilon_1< -\bm{\beta}_1'\bm{x}_{t-1}\mid X]+\sum\limits_{t=2}^{T}\ln\Bigg\{\sum\limits_{k_t=0,1}\sum\limits_{k_{t-1}=0,1}(-1)^{k_t+k_{t-1}} \\
&\textcolor{white}{=}\times \left(C_{s(y_t),s(y_{t-1})\mid \text{\b{S}}%
_{t-2}}\left(F_{s(y_t)\mid \text{\b{S}}%
_{t-2} }(s_t-k_t\mid \text{\b{s}}%
_{t-2},X),F_{s(y_{t-1})\mid \text{\b{S}}%
_{t-2}}(s_{t-1}-k_{t-1}\mid \text{\b{s}}%
_{t-2},X)\right)\right)\Bigg\}\\
&\textcolor{white}{=}-\sum\limits_{t=2}^{T}\ln\left\{P[s(y_{t-1})=s_{t-1}\mid \text{\b{S}}_{t-2}=\text{\b{s}}_{t-2},X]\right\}-n\ln \left\{ \frac{1}{2}\right\}
\end{align*}%
\endgroup
\end{proof}

\begin{proof}[Vine decomposition in the continuous case]

In Section \ref{EstimationC3}, it is shown that the signs $s(y_1),\cdots,s(y_T)$ may have a continuous extension with a perturbation in $[0,1]$ [see \citet{denuit2005constraints}]. This can be achieved by employing a transformation of the form $s^*(y_t)=s(y_t)+U-1$ for $t=1,\cdots,T$, where a natural choice for $U$ is the uniform distribution. Thus, for $\{s^*(y_t)\in \mathbb{R},t=1,\cdots,T\}$ consider the continuous equivalent of the conditional probability mass function  (\ref{eq: bayes}) - i.e. the conditional density function. Further, by letting $\text{\b{S}}_{t-1}^{*}$ be the continuous extension of $\text{\b{S}}_{t-1}$, the conditional density function may be expressed as
\begin{equation}\label{eq: continuous}
f_{s^*(y_t)\mid \text{\b{S}}_{t-1}^{*\backslash j}\cup s^*(y_j)}=\frac{f_{s^*(y_t),s^*(y_j)\mid \text{\b{S}}_{t-1}^{*\backslash j}}}{f_{s^*(y_j)\mid \text{\b{S}}_{t-1}^{*\backslash j}}}.
\end{equation}
From the Theorem of \citet{sklar1959fonctions}, we know that
\begingroup
\allowdisplaybreaks
\begin{align}\label{eq: Sklar}
\begin{split}
f_{s^*(y_t),s^*(y_j)\mid \text{\b{S}}_{t-1}^{*\backslash j}}(s_t^*,s_j^*\mid \text{\b{s}}_{t-1}^{*\backslash j},X)&=c_{s^*(y_t),s^*(y_j)\mid \text{\b{S}}_{t-1}^{*\backslash j}}\left(F_{s^*(y_t)\mid\text{\b{s}}_{t-1}^{*\backslash j}}(s_t^*\mid \text{\b{s}}_{t-1}^{*\backslash j},X),F_{s^*(y_j)\mid\text{\b{S}}_{t-1}^{*\backslash j}}(s_j^*\mid \text{\b{s}}_{t-1}^{*\backslash j},X) \right)\\
&\times\textcolor{white}{=}f_{s^*(y_t)\mid \text{\b{S}}_{t-1}^{*\backslash j}}f_{s^*(y_j)\mid \text{\b{S}}_{t-1}^{*\backslash j}},
\end{split}
\end{align}
\endgroup
where $c()$ is the copula density function. Thus,
\begin{equation}
f_{s^*(y_t)\mid \text{\b{S}}_{t-1}^{*\backslash j}\cup s^*(y_j)}=c_{s^*(y_t),s^*(y_j)\mid \text{\b{S}}_{t-1}^{*\backslash j}}\left(F_{s^*(y_t)\mid\text{\b{s}}_{t-1}^{*\backslash j}}(s_t^*\mid \text{\b{s}}_{t-1}^{*\backslash j},X),F_{s^*(y_j)\mid\text{\b{S}}_{t-1}^{*\backslash j}}(s_j^*\mid \text{\b{s}}_{t-1}^{*\backslash j},X) \right)f_{s^*(y_t)\mid \text{\b{S}}_{t-1}^{*\backslash j}},
\end{equation}
with
\begingroup
\allowdisplaybreaks
\begin{align}
\begin{split}
\begin{array}{l}
c_{s^*(y_t),s^*(y_j)\mid \text{\b{S}}_{t-1}^{*\backslash j}}\left(F_{s^*(y_t)\mid\text{\b{S}}_{t-1}^{*\backslash j}}(s_t^*\mid \text{\b{s}}_{t-1}^{*\backslash j},X),F_{s^*(y_j)\mid\text{\b{S}}_{t-1}^{*\backslash j}}(s_j^*\mid \text{\b{s}}_{t-1}^{*\backslash j},X) \right)=\\
\textcolor{white}{c_{s^*(y_t),s^*(y_j)\mid \text{\b{S}}_{t-1}^{*\backslash j}}(F_{s^*(y_t)\mid\text{\b{S}}_{t-1}^{*\backslash j}}(s_t^*\mid \text{\b{s}}_{t-1}^{*\backslash j},X)}\frac{\partial^2C_{s^*(y_t),s^*(y_j)\mid \text{\b{S}}_{t-1}^{*\backslash j}}\left(F_{s^*(y_t)\mid\text{\b{S}}_{t-1}^{*\backslash j}}(s_t^*\mid \text{\b{s}}_{t-1}^{*\backslash j},X),F_{s^*(y_j)\mid\text{\b{S}}_{t-1}^{*\backslash j}}(s_j^*\mid \text{\b{s}}_{t-1}^{*\backslash j},X) \right)}{\partial F_{s^*(y_t)\mid\text{\b{S}}_{t-1}^{*\backslash j}}\left(s_t^*\mid \text{\b{s}}_{t-1}^{*\backslash j},X\right)\partial F_{s^*(y_j)\mid\text{\b{S}}_{t-1}^{*\backslash j}}\left(s_j^*\mid \text{\b{s}}_{t-1}^{*\backslash j},X\right)}
\end{array}
\end{split}
\end{align}
\endgroup
 can express (\ref{eq: continuous}), and the arguments of the copulas, say, $F_{s^*(y_t)\mid\text{\b{S}}_{t-1}^{*\backslash j}}(s_t^*\mid \text{\b{s}}_{t-1}^{*\backslash j},X)$ are obtained using the expression by \citet{joe1996families}, such that
\begin{equation}
F_{s^*(y_t)\mid\text{\b{S}}_{t-1}^{*\backslash j}}(s_t^*\mid \text{\b{s}}_{t-1}^{*\backslash j},X)=\frac{\partial C_{s^*(y_t),s^*(y_i)\mid \text{\b{S}}_{t-1}^{*\backslash j,i}}\left(F_{s^*(y_t)\mid\text{\b{S}}_{t-1}^{*\backslash j,i}}(s_t^*\mid \text{\b{s}}_{t-1}^{*\backslash j,i},X),F_{s^*(y_i)\mid\text{\b{S}}_{t-1}^{*\backslash j,i}}(s_i^*\mid \text{\b{s}}_{t-1}^{*\backslash j,i},X)\right)}{\partial F_{s^*(y_i)\mid\text{\b{S}}_{t-1}^{*\backslash j,i}}(s_i^*\mid \text{\b{s}}_{t-1}^{*\backslash j,i},X)}.
\end{equation}
Therefore, when the data is continuous, the marginals in the copula expressions of, say, the third tree, $F_{t\mid t+1,t+2}$ for $t=1,\cdots,T-2$ and $F_{t+3\mid t+1,t+2}$ for $t=1,\cdots,T-3$ are obtained by 
\begin{equation}
F_{t\mid t+1,t+2}=\frac{\partial C_{t,t+1\mid t+2}(F_{t\mid t+2}(s_t^*\mid s_{t+2}^*,X),F_{t+1\mid t+2}(s_{t+1}^*\mid s_{t+2}^*,X))}{\partial F_{t+1\mid t+2}(s_{t+1}^*\mid s_{t+2}^*,X)},
\end{equation}
where $F_{t+3\mid t+1,t+2}$ is obtained in a similar way.
\end{proof}

\begin{proof}[Proof of Proposition \protect\ref{coroll1}]

The likelihood function of the sample in terms of signs $s(y_{1}),\cdots,s(y_{T})$ conditional on $X$ is given by%
\begin{equation*}
L\left( U(T),\bm{\beta},X \right) =P\left[s(y_{1})=s_{1},\cdots,s(y_{T})=s_{T}\mid X\right] 
\end{equation*}%
where each $s_{t}$, for $1\leq t\leq T$, takes two possible values $0$ and $%
1 $. Given model (\ref{eq: DGP}) and assumption (\ref{eq: median}), under the null hypothesis the signs $s(\varepsilon_{t})$, for $1\leq t\leq T$,\ are i.i.d conditional on $X$ according to $Bi(1,0.5)$. Then, the signs $s(y_{t}),$ for $1\leq t\leq T$, will also be i.i.d conditional on $X$
\begin{equation*}
P\left[ s(y_{t})=1\mid X\right] =P\left[ s(y_{t})=0\mid X\right] =%
\frac{1}{2},\text{ for }t=1,\cdots,T.
\end{equation*}%
Consequently, under $H_0$ we have
\begin{equation*}
L_{0}\left( U(T),\bm{0},X\right) =\prod\limits_{t=1}^{T}
P\left[s(y_{t})=s_{t}\mid X\right] =\left( \frac{1}{2}\right) ^{T}
\end{equation*}%
and under $H_1$ the likelihood function conditional on $X$ can be expressed as 
\[
L_1\left( U(T),\bm{\beta}_1,X \right) =P_1[s(y_1)=s_1\mid X]\times\prod\limits_{t=2}^{T}P_{t\mid 1:{t-1}}[s(y_t)=s_t\mid s(y_1)=s_1:s(y_{t-1})=s_{t-1},X].
\]
which can further be decomposed using the D-vine array $A=(\sigma_{lt})_{1\leq l\leq t\leq T}$ to obtain
\[
L_1\left( U(T),\bm{\beta}_1,X \right) =P_1[s(y_1)=s_1\mid X]\times\prod\limits_{t=2}^{T}\prod\limits_{l=t-1}^{2} c_{\sigma_{lt}t,\mid \sigma_{1t},\cdots,\sigma_{t-1,t}}\times c_{\sigma_{1t}t}\times P_t[s(y_t)=s_t\mid X]
\]
where now for $t=1,\cdots,T,$%
\begin{equation*}
y_{t}=\bm{\beta}_1'\bm{x}_{t-1}+\varepsilon_{t}
\end{equation*}%
Under assumption (\ref{eq: DGP}) and (\ref{eq: median}), the likelihood function conditional on $X$, under the alternative hypothesis can be expressed as 
\begingroup
\allowdisplaybreaks
\begin{align*}
L_1\left( U(T),\bm{\beta}_1,X \right)& =\left(1-P_1[\varepsilon_1< -\bm{\beta}_1' \bm{x}_0\mid X]\right)^{s(y_1)}\times P_1[\varepsilon_1 <- \bm{\beta}_1' \bm{x_0}\mid X]^{1-s(y_1)}\\
&\textcolor{white}{=}\times\prod\limits_{t=2}^{T}\prod\limits_{l=t-1}^{2} c_{\sigma_{lt}j,\mid \sigma_{1t},\cdots,\sigma_{t-1,t}}\times c_{\sigma_{1t}t}
\times \left(1-P_t[\varepsilon_t< -\bm{\beta}_1'\bm{x}_{t-1}\mid X]\right)^{s(y_t)}\\
&\textcolor{white}{=}\times P_t[\varepsilon_t < -\bm{\beta}_1' \bm{x}_{t-1}\mid X]^{1-s(y_t)}
\end{align*}
\endgroup
The log-likelihood ratio is given by%
\begingroup
\allowdisplaybreaks
\begin{align*}
\ln \left\{ \frac{L_1\left( U(T),\bm{\beta} _{1},X\right) }{L
_{0}\left( U(T),\bm{0},X\right) }\right\}&=s(y_1)\ln\left\{\frac{1-P_1[\varepsilon_1<-\bm{\beta}_1' \bm{x_0}\mid X ]}{P_1[\varepsilon_1 < -\bm{\beta}_1'\bm{x_0}\mid X]}\right\}+\ln \left(1-P_1[\varepsilon_1< -\bm{\beta}_1' \bm{x_0}\mid X]\right)\\
&\textcolor{white}{=}+\sum\limits_{t=2}^{T}\sum\limits_{l=t-1}^{2}\ln c_{\sigma_{lt}t,\mid \sigma_{1t},\cdots,\sigma_{t-1,t}}+\sum\limits_{t=2}^{T}\ln c_{a_{1t}t}+\sum\limits_{t=2}^{T}s(y_t)\ln\left\{\frac{1-P_t[\varepsilon_t<- \bm{\beta}_1' \bm{x}_{t-1}\mid X]}{P_t[\varepsilon_t <-\bm{\beta}_1' \bm{x}_{t-1}\mid X]}\right\}\\
&\textcolor{white}{=}+\sum\limits_{t=2}^{T}\ln\left(1- P_t[\varepsilon_t< -\bm{\beta}_1'\bm{x}_{t-1}\mid X]\right)-T\ln\left(\frac{1}{2}\right)
\end{align*}
\endgroup
According to Neyman-Pearson Lemma [see e.g. Lehmann (1959), page 65], the
best test for testing $H_{0}$ against $H_1,$ based on $s(y_{1}),\cdots,s(y_{T}),$
rejects $H_{0}$ when%
\begin{equation*}
\ln \left\{ \frac{L_1\left( U(T),\bm{\beta} _{1},X\right) }{L_{0}\left( U(T),\bm{0},X\right) }\right\} \geq c
\end{equation*}%
or when%
\begingroup
\allowdisplaybreaks
\begin{align*}
\ln \left\{ \frac{L_1\left( U(T),\bm{\beta} _{1},X\right) }{L
_{0}\left( U(T),\bm{0}\right) }\right\}&=\sum\limits_{t=2}^{T}\sum\limits_{l=t-1}^{2}\ln c_{\sigma_{lt}t,\mid \sigma_{1t},\cdots,\sigma_{t-1,t}}+\sum\limits_{t=2}^{T}\ln c_{\sigma_{1t}t}\\\
&\textcolor{white}{=}+\sum\limits_{t=1}^{T}s_t\ln\left\{\frac{1-P_t[\varepsilon_t<-\bm{\beta}_1' \bm{x}_{t-1}\mid X]}{P_t[\varepsilon_t< -\bm{\beta}_1' \bm{x}_{t-1}\mid X]}\right\}>c_1(\bm{\beta}_1)
\end{align*}
\endgroup
The critical value, say $c_1(\beta_1)$ is given by the smallest constant $c_1(\beta_1)$ such that%
\begin{equation*}
P\left( \ln \left\{ \frac{L_1\left( U(T),\bm{\beta}_{1},X\right) }{L_{0}\left( U(T),\bm{0},X\right) }\right\} >c_1(\bm{\beta}_1)\mid
H_{0}\right) \leq \alpha .
\end{equation*}
\end{proof}
\begin{proof}[Algorithm for the likelihood function of the signs under the alternative hypothesis]

In this Section, we adapt the algorithm for the joint pmf for D-vine for discrete variables of \citet{panagiotelis2012pair} and \citet{joe2014dependence} to the context of our study. Let $U(n)=\left(s(y_1),s(y_2),\cdots,s(y_T)\right)'$ be a binary valued $T$-vector. Furthermore, for a vector of integers $\mathbf{i}$, let $\mathbf{S_{i}}=\{s(y_i), i\in \mathbf{i}\}$, where $\mathbf{s_i}$ is a mass point of $\mathbf{S_i}$ and $s_g$ is a mass point of $s(y_g)$. Let
\begingroup
\allowdisplaybreaks
\begin{align*}
F_{g\mid \mathbf{i}}^{+}&:=P\left[s(y_g)\leq s_g\mid \mathbf{S_i}=\mathbf{s_i},X\right],\quad F_{g\mid \mathbf{i}}^{-}:=P\left[s(y_g)< s_g\mid \mathbf{S_i}=\mathbf{s_i},X\right],\\
f_{g\mid \mathbf{i}}&:=P[s(y_g)=s_g\mid\mathbf{S_i}=\mathbf{s_i},X].
\end{align*}
 \endgroup 
noting that when $\mathbf{i}=\{\emptyset\}$, these conditional probabilities, correspond to marginal probabilities. Furthermore, let  $C_{gh\mid\mathbf{i}}$ be a bivariate copula for the conditional CDFs $F_{g\mid\mathbf{i}}$ and $F_{h\mid\mathbf{i}}$, and denote
\begingroup
\allowdisplaybreaks
\begin{align*}
C^{++}_{gh\mid\mathbf{i}}&:=C_{gh\mid\mathbf{i}}\left(F_{g\mid\mathbf{i}}^+,F_{h\mid\mathbf{i}}^+\right),\quad C^{+-}_{gh\mid\mathbf{i}}:=C_{gh\mid\mathbf{i}}\left(F_{g\mid\mathbf{i}}^+,F_{h\mid\mathbf{i}}^-\right),\\
C^{-+}_{gh\mid\mathbf{i}}&:=C_{gh\mid\mathbf{i}}\left(F_{g\mid\mathbf{i}}^-,F_{h\mid\mathbf{i}}^+\right),\quad C^{--}_{gh\mid\mathbf{i}}:=C_{gh\mid\mathbf{i}}\left(F_{g\mid\mathbf{i}}^-,F_{h\mid\mathbf{i}}^-\right).
\end{align*}
\endgroup
The main elements of the algorithm is the following recursions:
\begin{itemize}
\item[(\rom{1})] $F_{j-t\mid (j-t+1):(j-1)}^+=\left[C_{j-t,j-1\mid(j-t+1):(j-2)}^{++}-C_{j-t,j-1\mid(j-t+1):(j-2)}^{+-}\right]/f_{j-1\mid(j-t+1):(j-2)};$ 
\item[(\rom{2})] $F_{j-t\mid (j-t+1):(j-1)}^-=\left[C_{j-t,j-1\mid(j-t+1):(j-2)}^{-+}-C_{j-t,j-1\mid(j-t+1):(j-2)}^{--}\right]/f_{j-1\mid(j-t+1):(j-2)};$
\item[(\rom{3})] $f_{j-t\mid (j-t+1):(j-1)}=F_{j-t\mid (j-t+1):(j-1)}^+-F_{j-t\mid (j-t+1):(j-1)}^-;$
\item[(\rom{4})] $F_{j\mid (j-t+1):(j-1)}^+=\left[C_{j-t+1,j\mid(j-t+2):(j-1)}^{++}-C_{j-t+1,j\mid(j-t+2):(j-1)}^{-+}\right]/f_{j-t+1\mid(j-t+2):(j-1)};$
\item[(\rom{5})] $F_{j\mid (j-t+1):(j-1)}^-=\left[C_{j-t+1,j\mid(j-t+2):(j-1)}^{+-}-C_{j-t+1,j\mid(j-t+2):(j-1)}^{--}\right]/f_{j-t+1\mid(j-t+2):(j-1)};$
\item[(\rom{6})] $f_{j\mid (j-t+1):(j-1)}=F_{j\mid (j-t+1):(j-1)}^+-F_{j\mid (j-t+1):(j-1)}^-;$
\item[(\rom{7})] The values based on $C_{j-t,j\mid (j-t+1):(j-1)}$ is computed; 
\item[(\rom{8})] $t$ is incremented by $1$ and back to (\rom{1}).
\end{itemize}
The identity employed in the recursions is
\begingroup
\allowdisplaybreaks
\begin{align*}
\begin{array}{ll}
P\left[s(y_g)\leq s_g\mid s(y_h)=s_h,\mathbf{S_i}=\mathbf{s_i},X\right]=\\
\textcolor{white}{P[s(y_g)\leq s_g\mid s(y_h)=s_h,}\frac{P\left[s(y_g)\leq s_g, s(y_h)\leq s_h\mid\mathbf{S_i}=\mathbf{s_i},X\right]-P\left[s(y_g)\leq s_g, s(y_h)< s_h\mid\mathbf{S_i}=\mathbf{s_i},X\right]}{P\left[s(y_h)=s_h\mid \mathbf{S_i}=\mathbf{s_i},X\right]}.
\end{array}
\end{align*}
 \endgroup
The algorithm is as follows
\begin{enumerate}
\item Input $\mathbf{s}_T=\left(s_1,\cdots,s_T\right)$.
\item Allocate an $T\times T$ matrix $\pi$, where $\pi_{tj}=f_{(j-t+1):j}$ for $t=1,\cdots,T$ and $j=t+1,\cdots,T$ and the likelihood function $P[s(y_1)=s_1,\cdots,s(y_T)=s_T]$ under the alternative will appear as $\pi_{TT}$.
\item Allocate $C^{++}$, $C^{+-}$, $C^{-+}$, $C^{--}$, $U^{'+}$, $U^{'-}$, $U^{+}$, $U^{-}$, $u'$, $u$, $w'$, $w$, as vectors of length $T$.
\item Evaluate $F^{+}_j$, $F^{-}_j$, and $f_j=F^{+}_j-F^{-}_j$ using (\ref{eq: BernoulliCDF2}) and let $\pi_{1j}\leftarrow f_j$ for $j=1,\cdots,T$;  
\item Let $C^{++}_j\leftarrow C_{j-1,j}\left(F^{+}_{j-1},F^{+}_{j}\right)$, $C^{+-}_j\leftarrow C_{j-1,j}\left(F^{+}_{j-1},F^{-}_{j}\right)$, $C^{-+}_j\leftarrow C_{j-1,j}\left(F^{-}_{j-1},F^{+}_{j}\right)$, and $C^{--}_j\leftarrow C_{j-1,j}\left(F^{-}_{j-1},F^{-}_{j}\right)$ for $j=2,\cdots,T$;
\item  Set $P_{2j}\leftarrow C^{++}_j-C^{+-}_j-C^{-+}_j+C^{--}_j$ for $j=2,\cdots,T$;
\item \textbf{for} $j=2,\cdots,T:\left(\mathcal{T}_1\right)$ \textbf{do}
\item \hspace{10pt} \parbox[t]{\linegoal}{$U_j^{'+}\leftarrow F_{j-1\mid j}^{+}=\left(C^{++}_j-C^{+-}_j\right)/f_j$, $U_j^{'-}\leftarrow F_{j-1\mid j}^{-}=\left(C^{-+}_j-C^{--}_j\right)/f_j$, and $u_j'\leftarrow f_{j-1\mid j}=F_{j-1\mid j}^{+}-F_{j-1\mid j}^{-}$;}
\item \hspace{10pt} \parbox[t]{\linegoal}{$U_j^{+}\leftarrow F_{j\mid j-1}^{+}=\left(C^{++}_j-C^{-+}_j\right)/f_{j-1}$, $U_j^{-}\leftarrow F_{j\mid j-1}^{-}=\left(C^{-+}_j-C^{--}_j\right)/f_{j-1}$, and $u_j\leftarrow f_{j\mid j-1}=F_{j\mid j-1}^{+}-F_{j\mid j-1}^{-}$;}
\item \textbf{end for}
\item \textbf{for} $t=2,\cdots,T-1: \left(\mathcal{T}_2,\cdots,\mathcal{T}_{T-1}\right)$ \textbf{do}
\item\hspace{10pt} let $C^{\alpha\beta}_j\leftarrow C_{j-t,j\mid (j-t+1):(j-1)}\left(U^{'\alpha}_{j-1},U^{\beta}_j\right)$, for $j=t+1,\cdots,T$ and $\alpha,\beta\in\{+,-\}$; 
\item\hspace{10pt} let $w_{j}'\leftarrow u_{j}'$, $w_{j}\leftarrow u_{j}$ for $j=t,\cdots,T;$
\item\hspace{10pt} \textbf{for} $j=t+1,\cdots,T:$ \textbf{do}
\item\hspace{30pt} $U^{'+}_{j}\leftarrow \left(C^{++}_j-C^{+-}_j\right)/w_j$, $U^{'-}_{j}\leftarrow \left(C^{-+}_j-C^{--}_j\right)/w_j$ and $u_j'\leftarrow U^{'+}_{j}-U^{'-}_{j}$;
\item\hspace{30pt} $U^{+}_{j}\leftarrow \left(C^{++}_j-C^{-+}_j\right)/w_{j-1}'$, $U^{-}_{j}\leftarrow \left(C^{+-}_j-C^{--}_j\right)/w_{j-1}'$ and $u_j\leftarrow U^{+}_{j}-U^{-}_{j}$;
\item\hspace{10pt} \textbf{end for}
\item\hspace{10pt} let $\pi_{t+1,j}\leftarrow \pi_{t,j-1}\times u_j$ for $j=t+1,\cdots,T$.
\item\textbf{end for}
\item Return the likelihood function $\pi_{TT}$.
\end{enumerate}
\end{proof}

\begin{proof}[Proof of Theorem \protect\ref{Theorem1}]
 The characteristic function of the test statistic $SN_T(\bm{\beta}_0\mid\bm{\beta}_1)$ conditional on $X$ is given by
\begingroup
\allowdisplaybreaks
\begin{align*}
\phi_{SN_T}(u)&=\mathbb{E}_{X}\left[\exp(iu SN_T(\bm{\beta}_0\mid\bm{\beta}_1))\right]\\
&=\E_X\left[\exp\left(iu\left(\sum\limits_{t=1}^T R_{t,t-1}+\sum\limits_{t=1}^{T}\ln\left\{\frac{1-p_t[\bm{x}_{t-1},\bm{\beta}_0,\bm{\beta}_1\mid X]}{p_t[\bm{x}_{t-1},\bm{\beta}_0,\bm{\beta}_1\mid X]}\right\}s(\tilde{y}_t)\right)\right)\right],
\end{align*}
\endgroup
which may be expressed as
\begingroup
\allowdisplaybreaks
\begin{align*}
\phi_{SN_T}(u)=\E_X\left[\prod\limits_{t=1}^{T}\exp\left(iu\left(R_{t,t-1}+\ln\left\{\frac{1-p_t[\bm{x}_{t-1},\bm{\beta}_0,\bm{\beta}_1\mid X]}{p_t[\bm{x}_{t-1},\bm{\beta}_0,\bm{\beta}_1\mid X]}\right\}s(\tilde{y}_t)\right)\right)\right],
\end{align*}
\endgroup
with $R_{1,0}=0$, and $R_{t,t-1}=\sum\limits_{l=t-1}^{2}\ln c_{\tilde{\sigma}_{lt}t\mid \tilde{\sigma}_{1t},\cdots,\tilde{\sigma}_{t-1,t}}+\ln c_{\tilde{\sigma}_{1t}t}$ for $t=2,\cdots,T$, for the D-vine-array $\tilde{A}=(\tilde{\sigma}_{lt})_{1\leq l\leq t\leq T}$, such that $l=2,\cdots,T-1$ is the row with tree $\mathcal{T}_l$, and column $t$ has the permutation $\tilde{\underline{\sigma}}_{t-1}=(\tilde{\sigma}_{1t},\cdots,\tilde{\sigma}_{t-1,t})$ of the previously added variables, $p_t[\bm{x}_{t-1},\bm{\beta}_0,\bm{\beta}_1\mid X]=P_t[\varepsilon_t\leq f(\bm{x}_{t-1},\bm{\beta}_0)-f(\bm{x}_{t-1},\bm{\beta}_1)\mid X]$, and $s(\tilde{y}_t)=s(y_t-f(\bm{x}_{t-1},\bm{\beta}_0))$. Furthermore, $u\in\mathbb{R}$ and the complex number $i=\sqrt{-1}$. Unlike \citet{dufour2010exact}, $\tilde{y}_t$ for $t=1,\cdots,T$ are no longer necessarily independent conditional on $X$. Therefore, we follow \citet{heinrich1982factorization} by expressing the characteristic function $\phi_{SN_T}(u)$ as follows
\[
\phi_{SN_T}(u)=\prod\limits_{t=1}^{T}\varphi_t(u)
\]
where $\varphi_1(u)=\E_X\left[\exp(iu\left(\ln\left\{\frac{1-p_1[\bm{x}_{0},\bm{\beta}_0,\bm{\beta}_1\mid X]}{p_1[\bm{x}_{0},\bm{\beta}_0,\bm{\beta}_1\mid X]}\right\}s(\tilde{y}_1)\right))\right]$ and for $t=2,\cdots,T$
\[
\varphi_t(u)=\frac{f_t(u)}{f_{t-1}(u)},\quad\text{where},\quad f_t(u)=\E_X\left[\exp(iuSN_t(\bm{\beta}_0\mid\bm{\beta}_1))\right]
\] 
\citet{heinrich1982factorization} shows that $\varphi_t(u)$ can alternatively be expressed as
\[
\varphi_t(u)=\E_X\left[\exp\left(iu\left\{R_{t,t-1}+\ln\left\{\frac{1-p_t[\bm{x}_{t-1},\bm{\beta}_0,\bm{\beta}_1\mid X]}{p_t[\bm{x}_{t-1},\bm{\beta}_0,\bm{\beta}_1\mid X]}\right\}s(\tilde{y}_t)\right\}\right)\right]+\rho_t(u)
\]
where 
\begingroup
\allowdisplaybreaks
\begin{align*}
\rho_t(u)&=\bigg\{\E_X\left[\exp\left(iu\left\{SN_t\left(\bm{\beta}_0\mid\bm{\beta}_1\right)\right\}\right)\right]-\\
&\textcolor{white}{=}\textcolor{white}{=}\E_X\left[\exp\left(iu\left\{R_{t,t-1}+\ln\left\{\frac{1-p_t[\bm{x}_{t-1},\bm{\beta}_0,\bm{\beta}_1\mid X]}{p_t[\bm{x}_{t-1},\bm{\beta}_0,\bm{\beta}_1\mid X]}\right\}s(\tilde{y}_t)\right\}\right)\right]\times\\
&\textcolor{white}{=}\textcolor{white}{=}\E_X\left[\exp\left(iu\left\{SN_{t-1}\left(\bm{\beta}_0\mid\bm{\beta}_1\right)\right\}\right)\right]\bigg\}\bigg/ \E_X\left[\exp\left(iu\left\{SN_{t-1}\left(\bm{\beta}_0\mid\bm{\beta}_1\right)\right\}\right)\right].
\end{align*}
\endgroup
Therefore, the characteristic function of the PCC-POS test statistic can be expressed as 
\begingroup
\allowdisplaybreaks
\begin{align}
\begin{split}\label{eq: Fourier-inverse}
\phi_{SN_T}(u)&=\prod\limits_{t=1}^{T}\varphi_t(u)\\
&=\prod\limits_{t=1}^{T}\left(\E_X\left[\exp\left(iu\left\{R_{t,t-1}+\ln\left\{\frac{1-p_t[\bm{x}_{t-1},\bm{\beta}_0,\bm{\beta}_1\mid X]}{p_t[\bm{x}_{t-1},\bm{\beta}_0,\bm{\beta}_1\mid X]}\right\}s(\tilde{y}_t)\right\}\right)\right]+\rho_t(u)\right),
\end{split}
\end{align}
\endgroup
where $\rho_1(u)=0$, $R_{1,0}=\rho_1(u)=0$. 

Let $Z_t=R_{t,t-1}+\ln\left\{\frac{1-p_t[\bm{x}_{t-1},\bm{\beta}_0,\bm{\beta}_1\mid X]}{p_t[\bm{x}_{t-1},\bm{\beta}_0,\bm{\beta}_1\mid X]}\right\}s(\tilde{y}_t)$ for $t=1,\cdots,T$. Then following \citet{heinrich1982factorization}, and by assuming that $Z_1,\cdots,Z_T$ are weakly dependent, the term $\rho_t(u)$ can further be factorized. For instance, a case of such weakly dependent random variables for which a Theorem exists is the regularity Markov type process (i.e. RMT-process). Let $\mathcal{B}_s^{s+m}=\sigma(Z_s,\cdots,Z_{s+m})$ be the Borel $\sigma$-field generated by $\{Z_t, t=s,\cdots,s+u\}$. The process $\{Z_t\}_{t=1,2,\cdots}$ is an RMT-process, if for $1\leq s\leq t$, the uniform mixing coefficient $\phi(m)\leq\gamma(s,t)$ with probability one, where 
\begin{equation*}
\phi(m)\equiv\sup_{s\geq 1}\phi(\mathcal{B}_{1}^s,\mathcal{B}_{s+m}^\infty)
\end{equation*}
and where $\phi(\mathcal{B}_{1}^s,\mathcal{B}_{s+m}^\infty)$
\begin{equation*}
\phi(\mathcal{B}_{1}^s,\mathcal{B}_{s+m}^\infty)\equiv \sup_{G\in\mathcal{B}_{s+m}^{\infty},H\in \mathcal{B}_{1}^s}\lvert P[H\mid G]-P[H]\lvert,
\end{equation*}
with $\sup_{s\geq1}\gamma(s,s+m)\rightarrow0$ as $m\rightarrow \infty$. Given such dependence, $\rho_t(u)$ can be factorized using the results of Theorem 2 of \citet{heinrich1982factorization}.

The conditional CDF of $SN_T(\bm{\beta}_0\mid\bm{\beta}_1)$ evaluated at a constant $c_1(\bm{\beta}_0,\bm{\beta}_1)$, where $c_1(\bm{\beta}_0,\bm{\beta}_1)\in\mathbb{R}$, given by the conditional characteristic functions $\phi_{SN_T}(u)$ can then be obtained using the Fourier-inversion formula [see \citet{gil1951note}] as follows
\[
P[SN_T(\bm{\beta}_0\mid \bm{\beta}_1)\leq c_1(\bm{\beta}_0,\bm{\beta}_1)]=\frac{1}{2}-\frac{1}{\pi}\int_{0}^{\infty}\frac{\Im\{\exp(-iuc_1(\bm{\beta}_0,\bm{\beta}_1))\phi_{SN_T}(u)\}}{u}du
\] 
where $\forall u \in \mathbb{R}$, the conditional characteristic function $\phi_{SN_T}(u)$ is expressed by (\ref{eq: Fourier-inverse}) and $\Im{z}$ denotes the imaginary part of the complex number $z$. Therefore, the power function can be obtained as follows
\[
\Pi(\bm{\beta}_0,\bm{\beta}_1)=P[SN_T(\bm{\beta}_0\mid \bm{\beta}_1)> c_1(\bm{\beta}_0,\bm{\beta}_1)]=\frac{1}{2}+\frac{1}{\pi}\int_{0}^{\infty}\frac{\Im\{\exp(-iuc_1(\bm{\beta}_0,\bm{\beta}_1))\phi_{SN_T}(u)\}}{u}du
\] 

\end{proof}
\newpage
\begin{proof}[Additional simulations]

\begin{figure}[tbph]
\caption{Power comparisons: different tests. Student's $t(\nu)$ error distributions, with
different degrees of freedom $\nu$, $\protect\rho=0$ in (\protect\ref{eq: errorsim}) and $\protect%
\theta =0.9$ in (\protect\ref{eq: theta})}
\begin{center}
\subfigure{\includegraphics[scale=0.58]{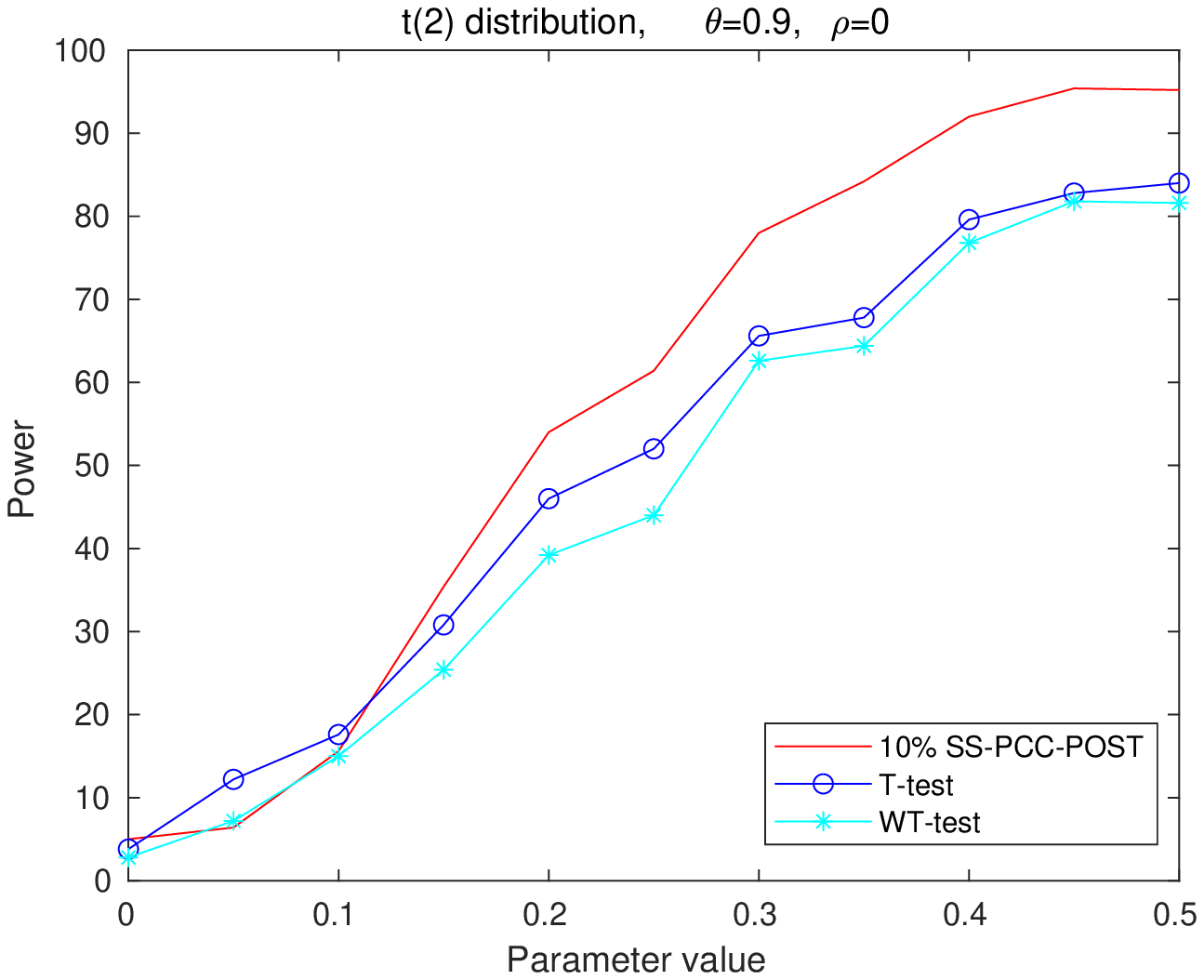}} %
\subfigure{\includegraphics[scale=0.58]{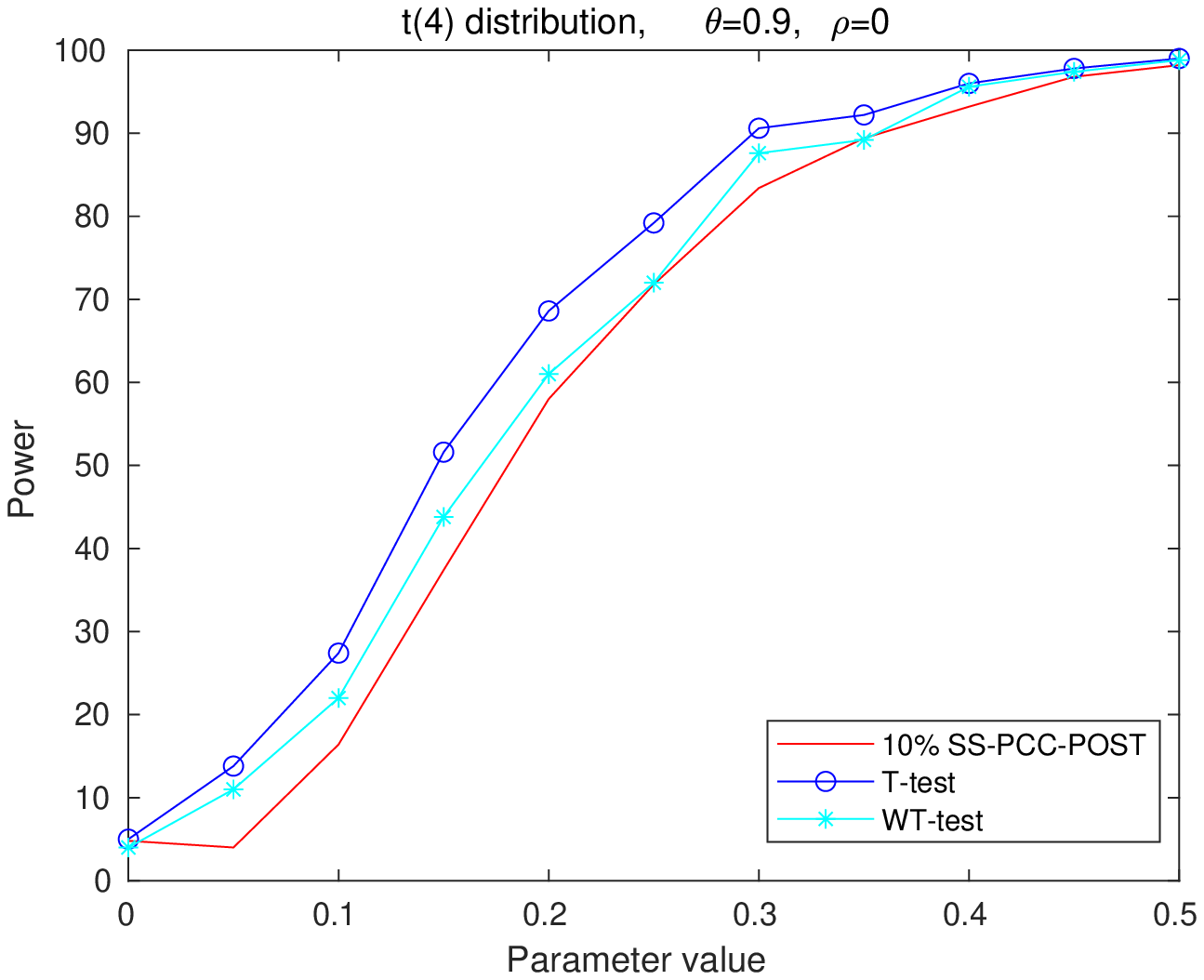}} \\[0pt]
\subfigure{\includegraphics[scale=0.58]{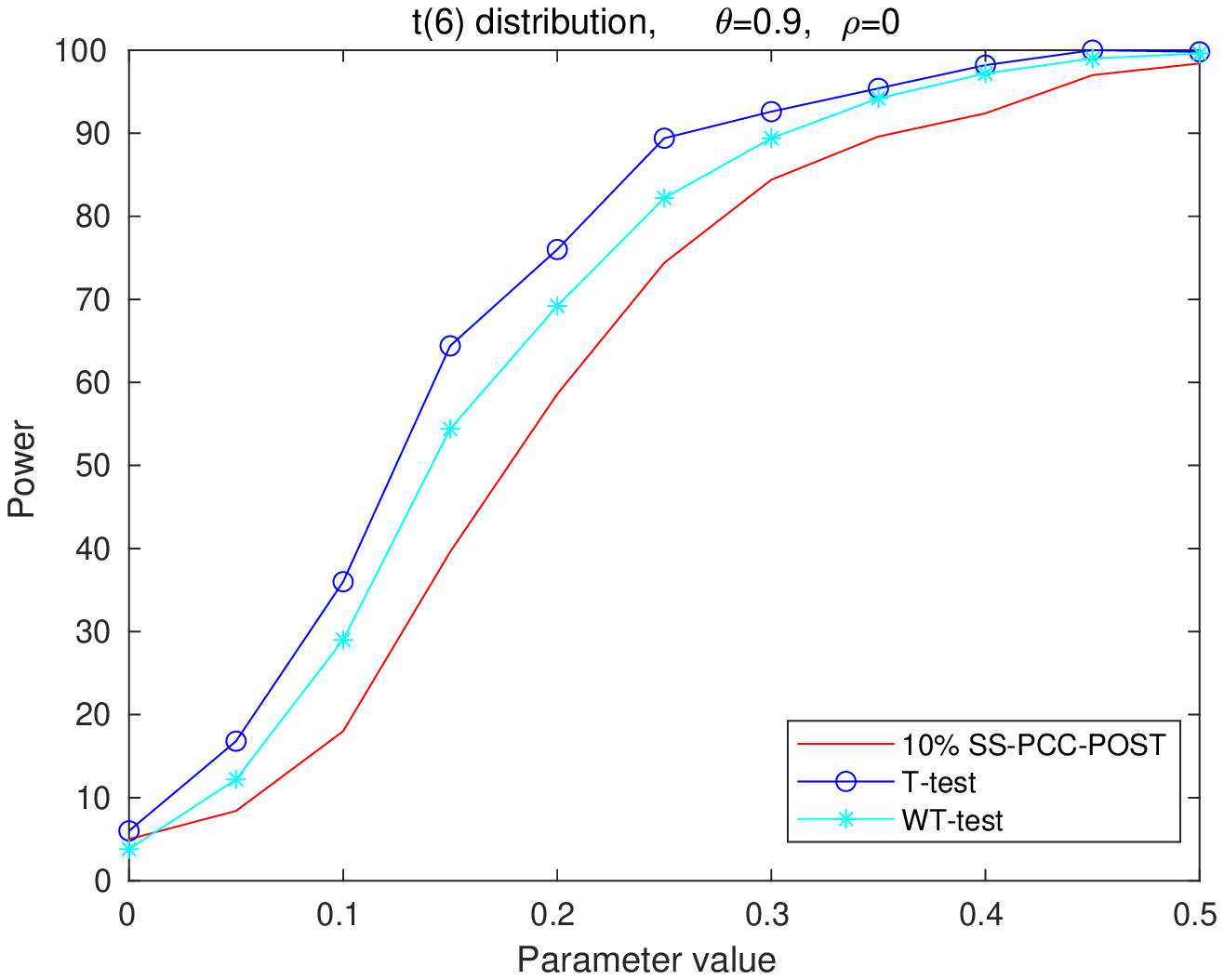}} %
\end{center}
\doublespacing
Note: These figures compare the power curves of the 10\% split-sample PCC-POS test
[10\% SS-PCC-POS test] with: (1) the \textit{t}-test and (2) the \textit{t}-test based
on White's (1980) variance correction [WT-test]. 
\label{fig: c21}
\end{figure}
\FloatBarrier

\begin{figure}[tbph]
\caption{Power comparisons: different tests. Student's $t(\nu)$ error distributions, with
different degrees of freedom $\nu$, $\protect\rho=0.1$ in (\protect\ref{eq: errorsim}) and $\protect%
\theta =0.9$ in (\protect\ref{eq: theta})}
\begin{center}
\subfigure{\includegraphics[scale=0.58]{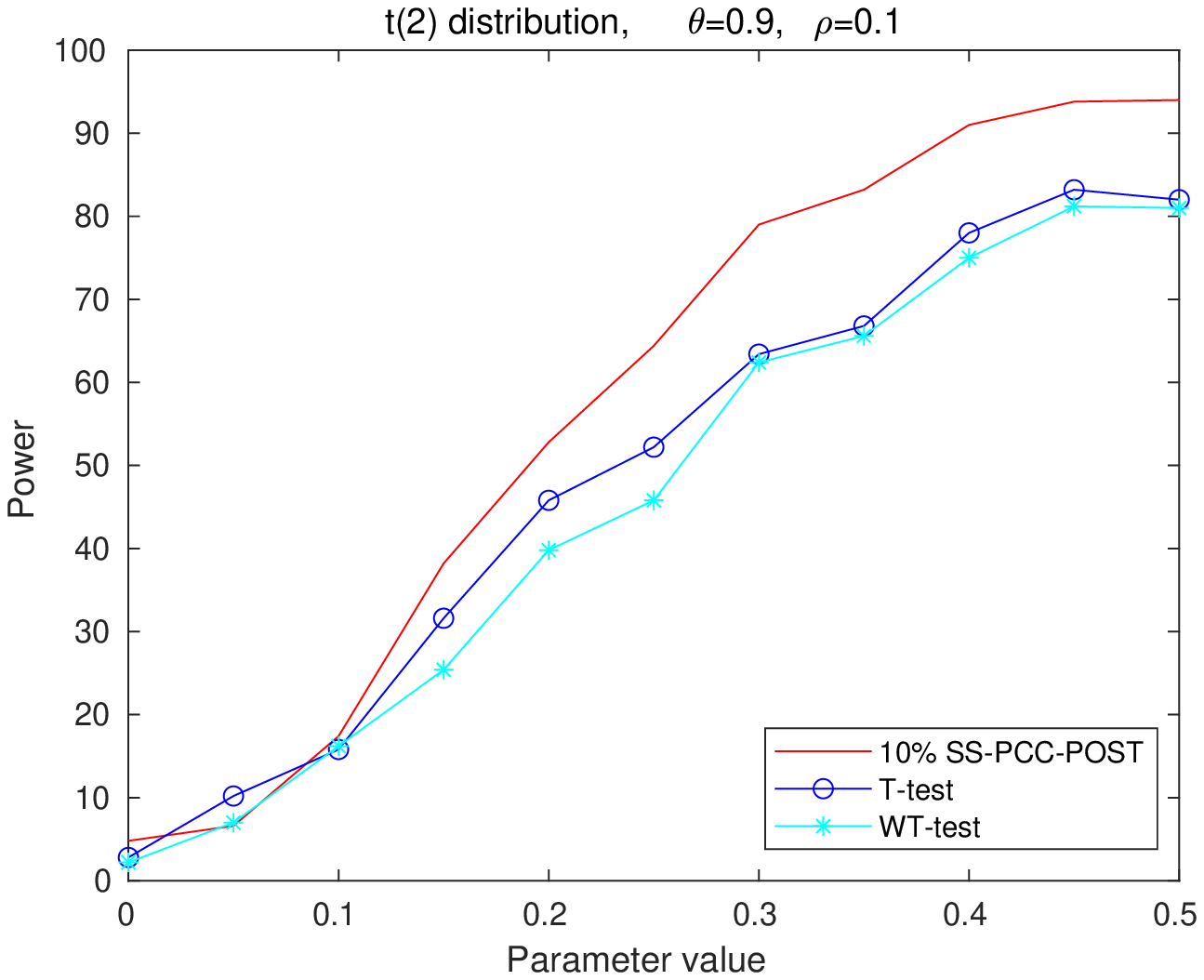}} %
\subfigure{\includegraphics[scale=0.58]{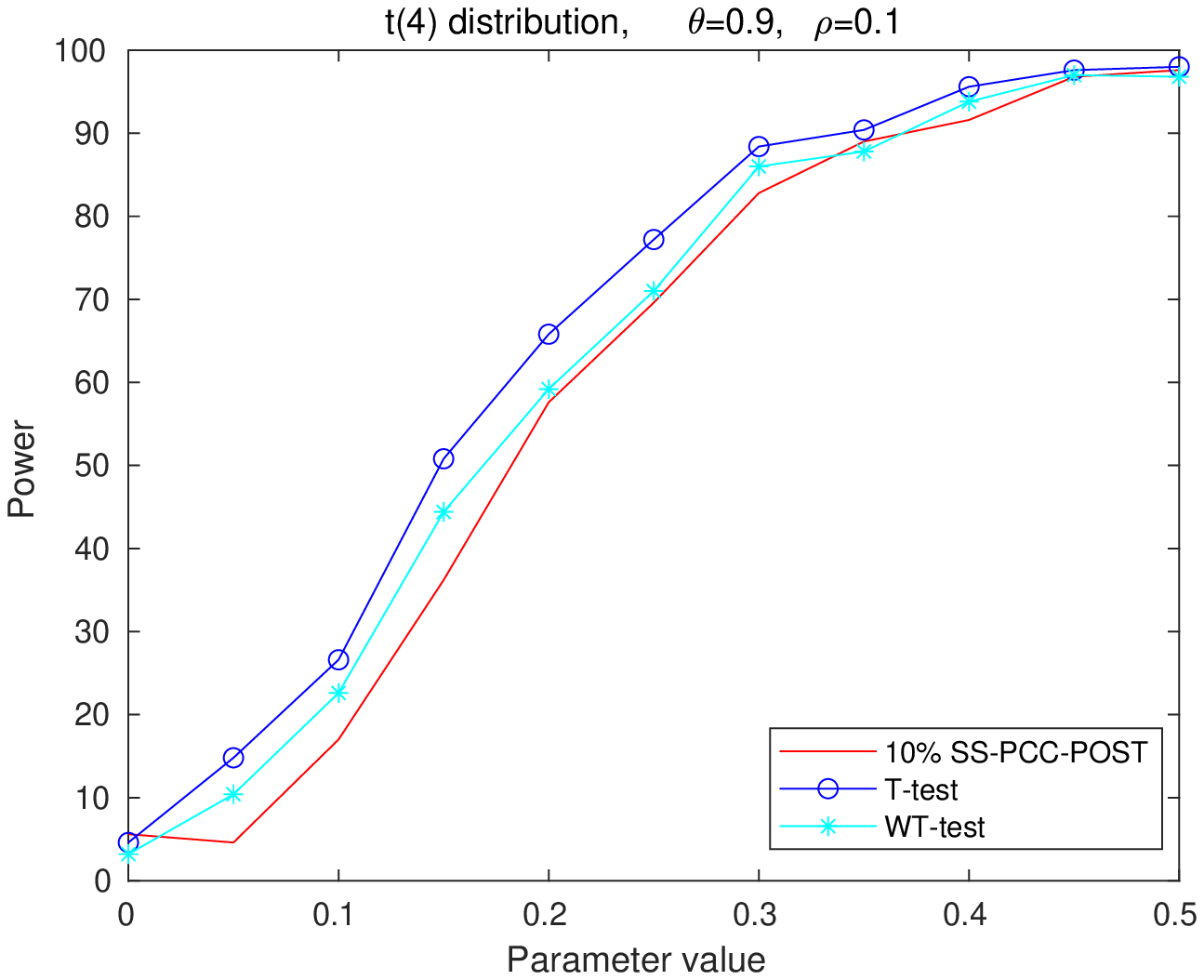}} \\[0pt]
\subfigure{\includegraphics[scale=0.58]{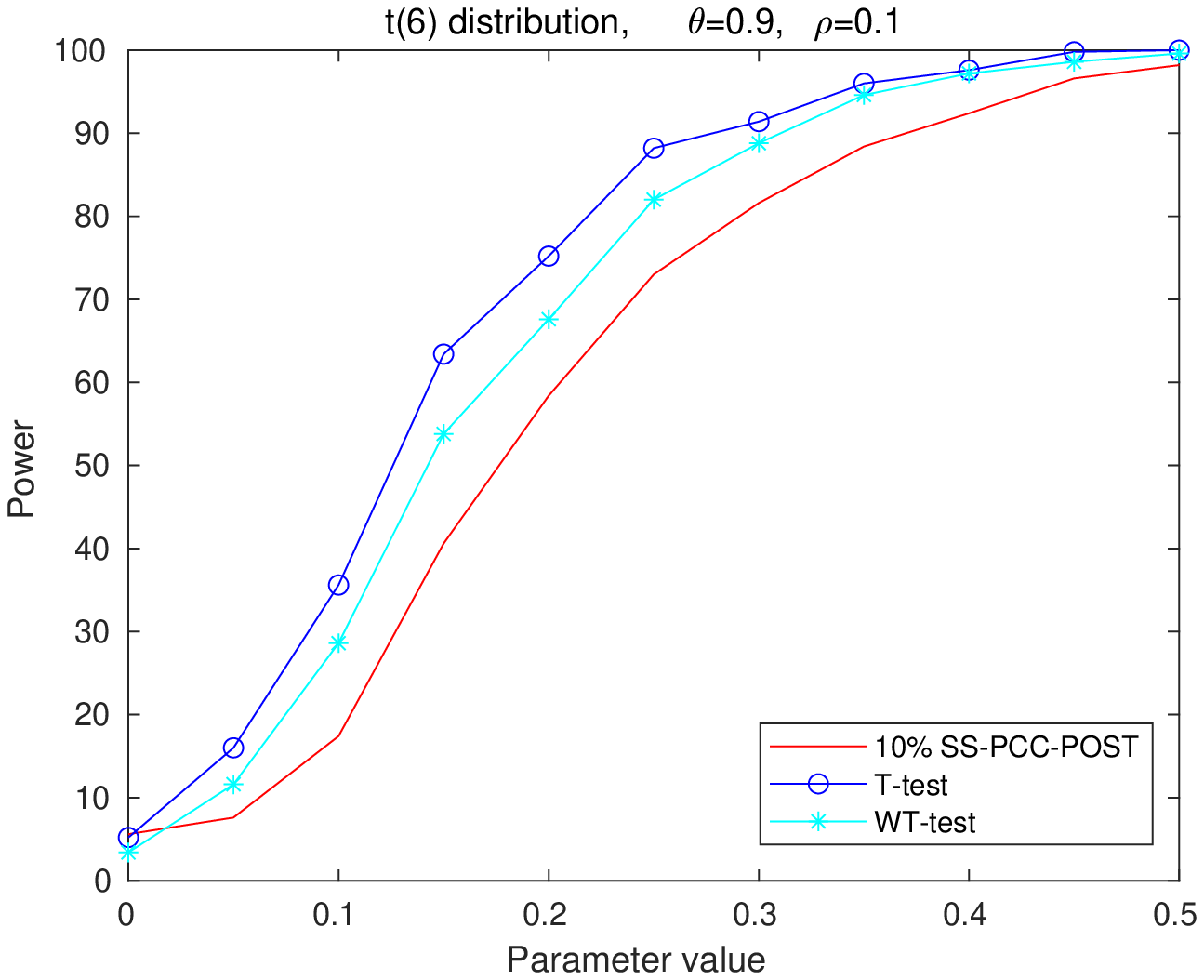}} %
\end{center}
\doublespacing
Note: These figures compare the power curves of the 10\% split-sample PCC-POS test
[10\% SS-PCC-POS test] with: (1) the \textit{t}-test and (2) the \textit{t}-test based
on White's (1980) variance correction [WT-test]. 
\label{fig: c22}
\end{figure}
\FloatBarrier

\begin{figure}[tbph]
\caption{Power comparisons: different tests. Student's $t(\nu)$ error distributions, with
different degrees of freedom $\nu$, $\protect\rho=0.5$ in (\protect\ref{eq: errorsim}) and $\protect%
\theta =0.9$ in (\protect\ref{eq: theta})}
\begin{center}
\subfigure{\includegraphics[scale=0.58]{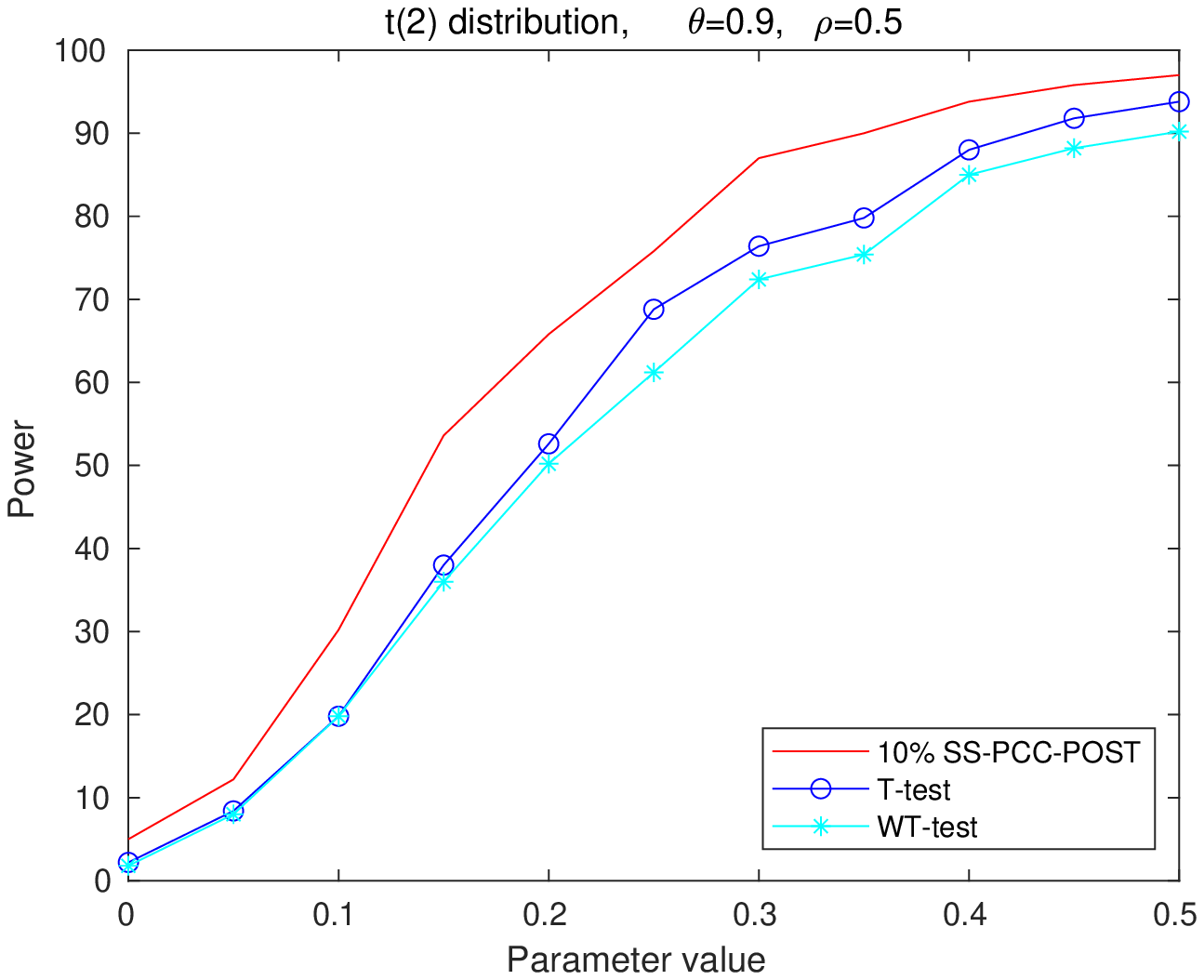}} %
\subfigure{\includegraphics[scale=0.58]{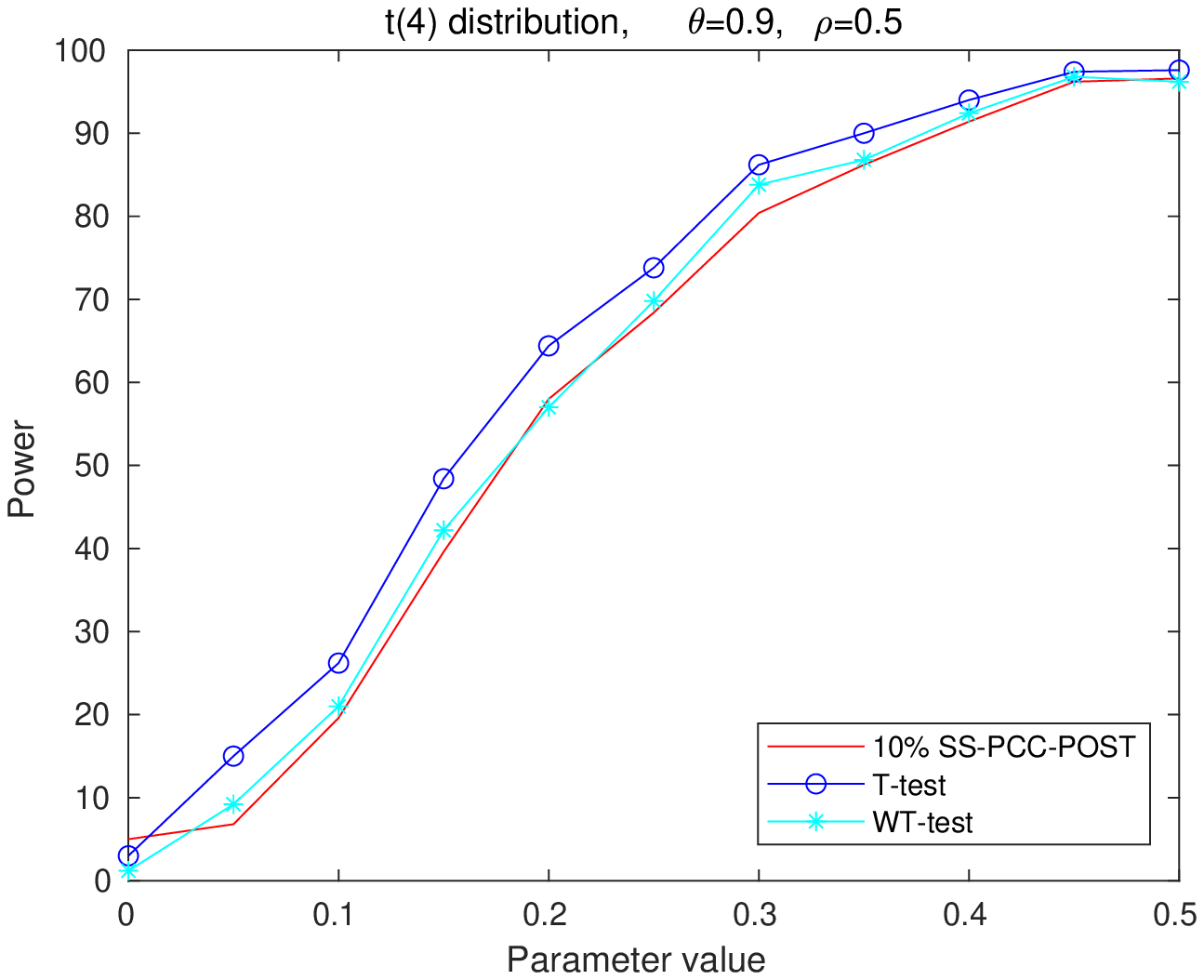}} \\[0pt]
\subfigure{\includegraphics[scale=0.58]{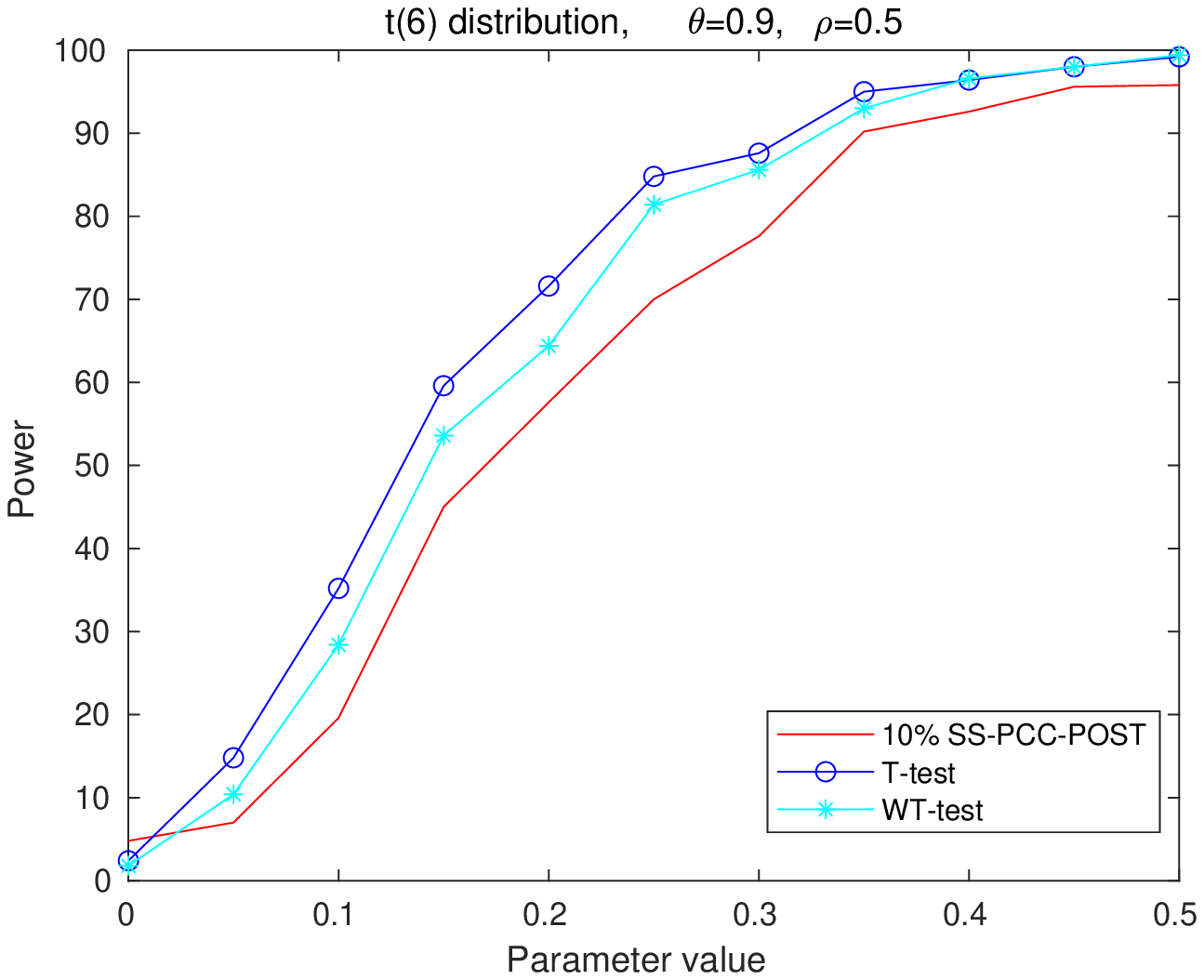}} %
\end{center}
\doublespacing
Note: These figures compare the power curves of the 10\% split-sample PCC-POS test
[10\% SS-PCC-POS test] with: (1) the \textit{t}-test and (2) the \textit{t}-test based
on White's (1980) variance correction [WT-test].  
\label{fig: c23}
\end{figure}
\FloatBarrier

\begin{figure}[tbph]
\caption{Power comparisons: different tests. Student's $t(\nu)$ error distributions, with
different degrees of freedom $\nu$, $\protect\rho=0.9$ in (\protect\ref{eq: errorsim}) and $\protect%
\theta =0.9$ in (\protect\ref{eq: theta})}
\begin{center}
\subfigure{\includegraphics[scale=0.58]{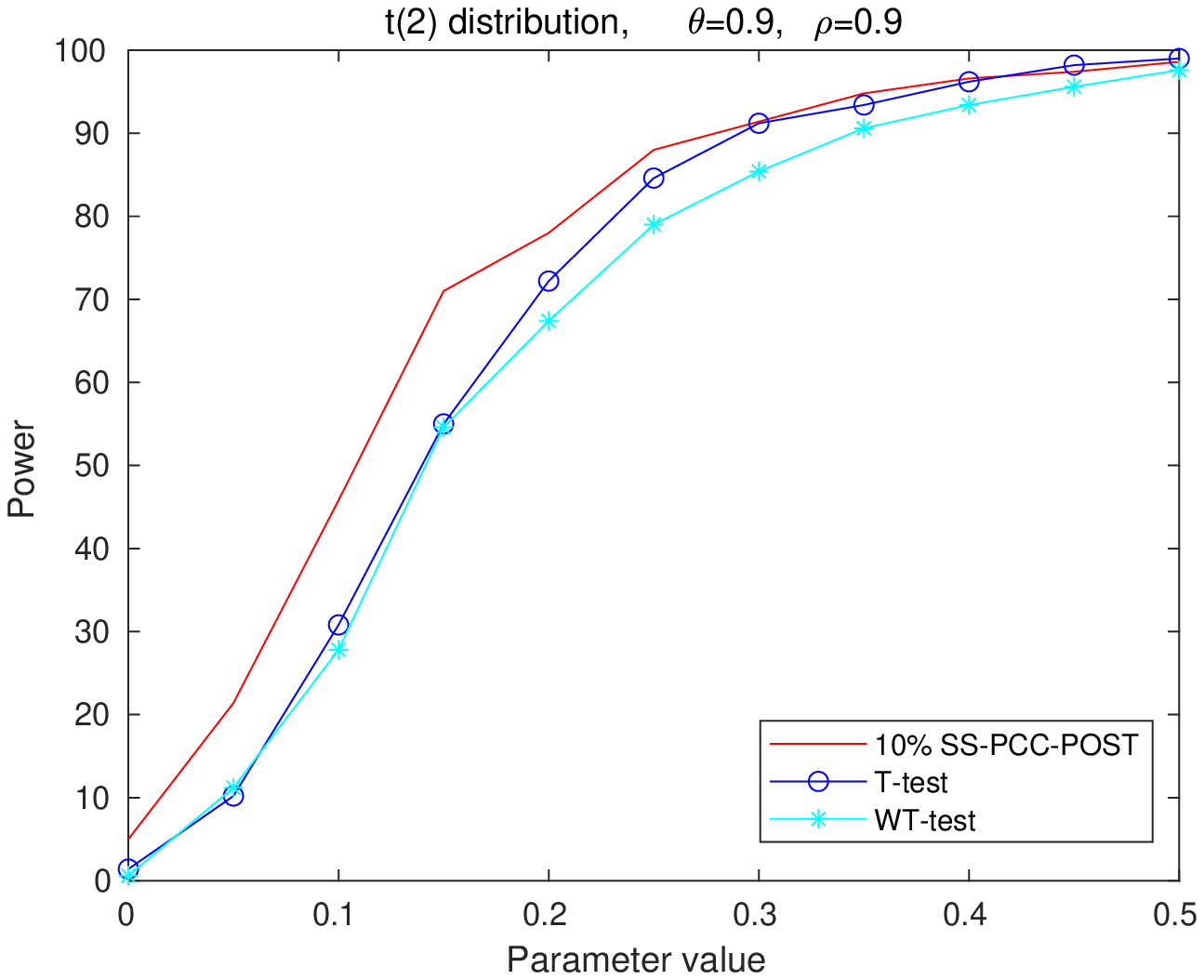}} %
\subfigure{\includegraphics[scale=0.58]{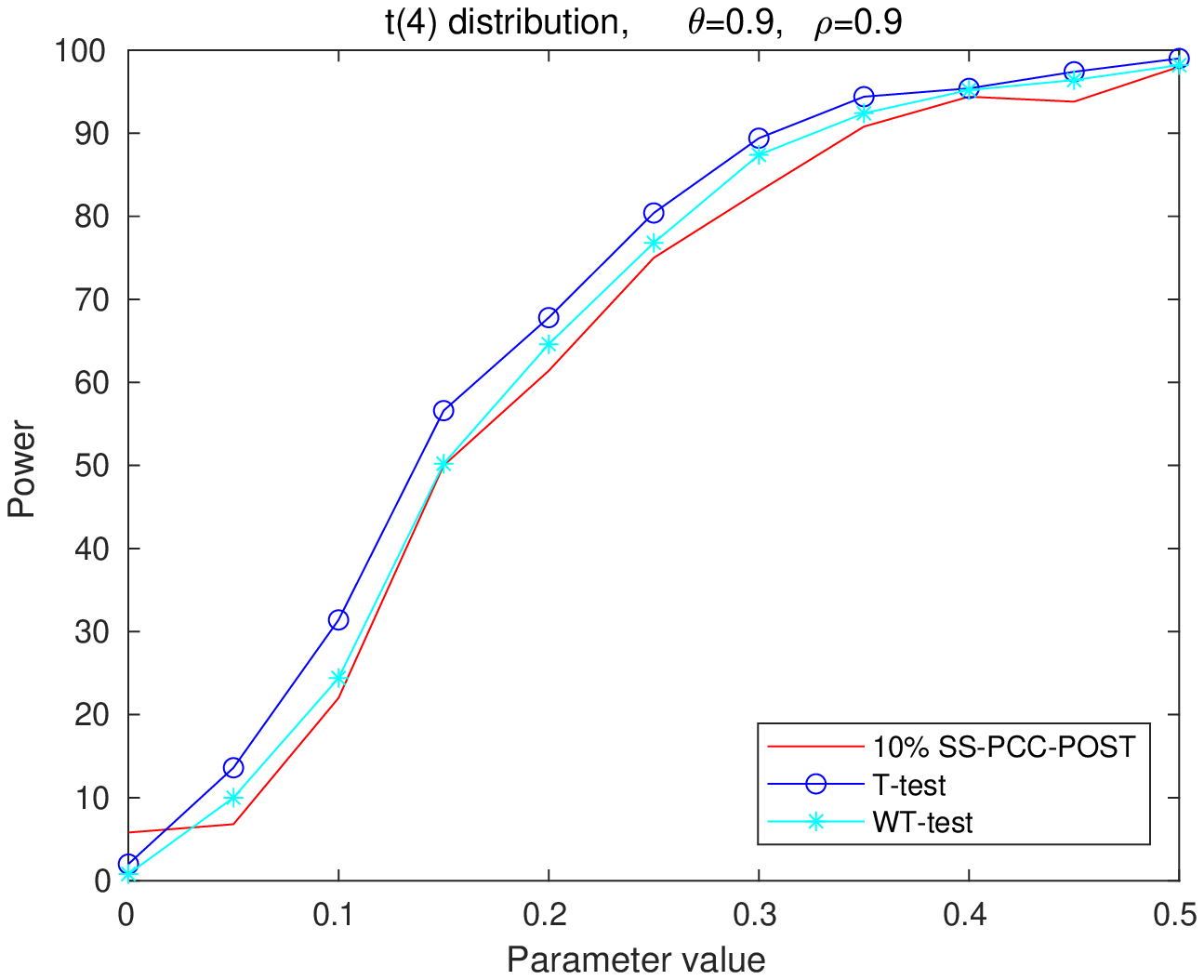}} \\[0pt]
\subfigure{\includegraphics[scale=0.58]{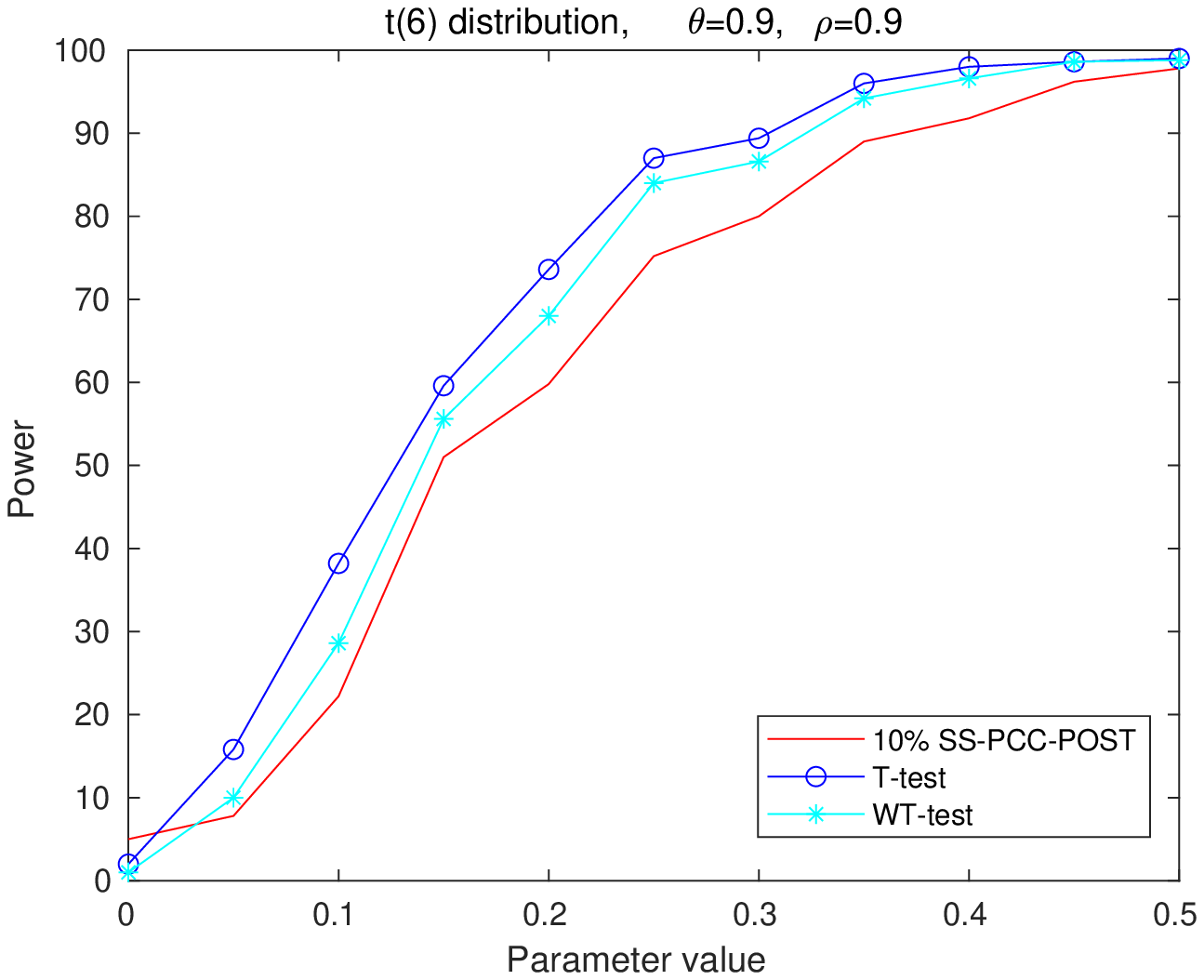}} %
\end{center}
\doublespacing
Note: These figures compare the power curves of the 10\% split-sample PCC-POS test
[10\% SS-PCC-POS test] with: (1) the \textit{t}-test and (2) the \textit{t}-test based
on White's (1980) variance correction [WT-test]. 
\label{fig: c24}
\end{figure}
\FloatBarrier
\begin{figure}
\caption{Comparison of the student's $t$ distribution with various degrees of freedom to the normal distribution}
\includegraphics[width=\textwidth]{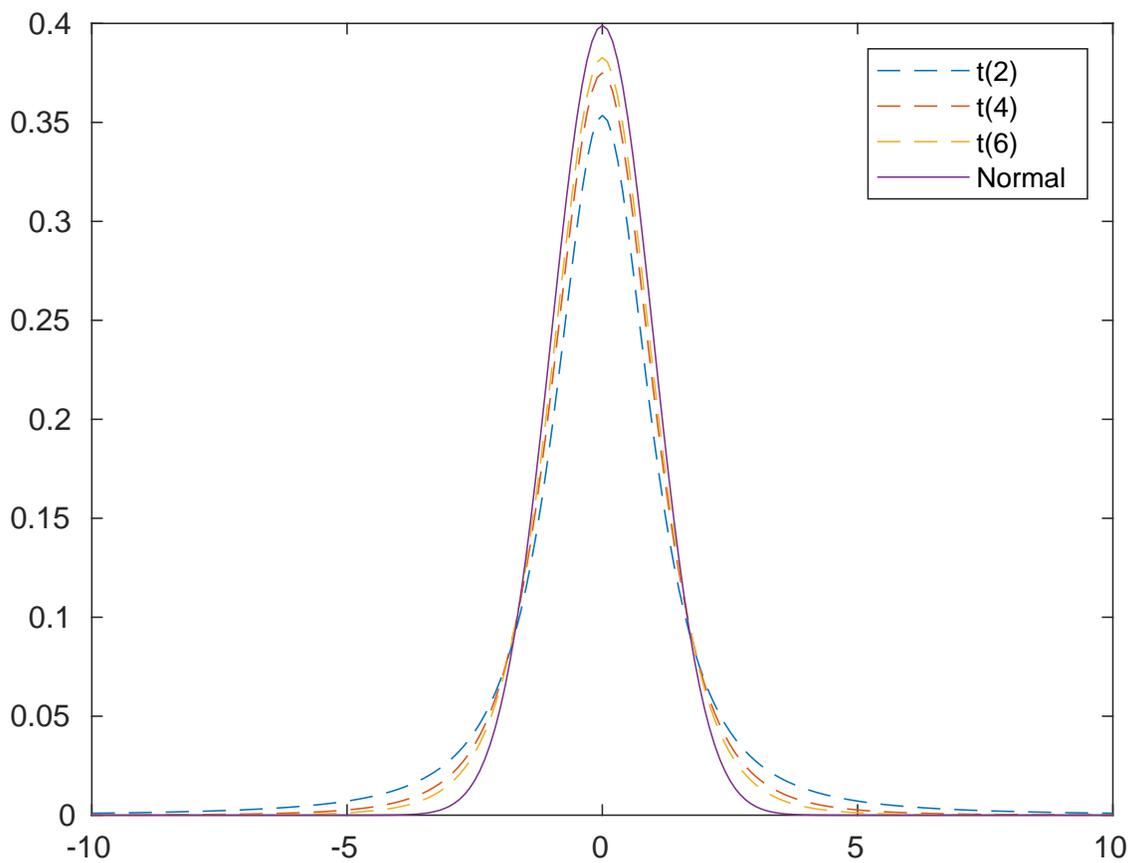}\label{fig: distos1}
Note: In this figure, we compare the Normal and Student's distribution with two, four and six degrees of freedom - i.e. $\nu=2,\quad \nu=4,\quad \nu=6$.
\label{fig: c25}
\end{figure}
\FloatBarrier
\end{proof}
\newpage
\bibliographystyle{chicago}
\bibliography{References_Final}

\end{document}